\documentclass[12pt]{article}

\usepackage{wrapfig}
\usepackage[utf8]{inputenc}
\usepackage{amsmath}
\usepackage{amsfonts}
\usepackage[english]{babel}
\usepackage{amsthm}
\usepackage{amssymb}
\usepackage[ruled,vlined]{algorithm2e}
\usepackage{graphicx}
\usepackage{bbm}
\usepackage{xcolor}
\usepackage{bibentry}
\usepackage[round]{natbib}
\usepackage{hyperref}

\newcommand{\footremember}[2]{%
    \footnote{#2}
    \newcounter{#1}
    \setcounter{#1}{\value{footnote}}%
}

\newcommand{\R}{\mathbb{R}}

\newcommand{\var}{\text{Var}}
\newcommand{\pr}{\mathbb{P}}
\newcommand{\E}{\mathbb{E}}

\newcommand{\cov}{\text{Cov}}
\newcommand{\im}{\sqrt{-1}}
\newcommand{\mt}{\tilde{m}}
\newcommand{\FF}{{\cal F}}

\newtheorem{lemma}{Lemma}[section]
\newtheorem{remark}{Remark}
\newtheorem{corollary}{Corollary}
\newtheorem{proposition}{Proposition}
\newtheorem{theorem}{Theorem}[section]

\newtheorem{assumption}{Assumption}

\newtheorem{example}{Example}

\begin{document}

	\title{Frequency Detection and Change Point Estimation for Time Series of Complex Oscillation}

\author{Hau-Tieng Wu\footremember{1}{Departments of Mathematics and Department of Statistical Science, Duke University} \ and Zhou Zhou\footremember{2}{Department of Statistics, University of Toronto} }
\date{}

\maketitle

\begin{abstract}
We consider detecting the evolutionary oscillatory pattern of a signal when it is contaminated by non-stationary noises with complexly time-varying data generating mechanism. A high-dimensional dense progressive periodogram test is proposed to accurately detect all oscillatory frequencies. A further phase-adjusted local change point detection algorithm is applied in the frequency domain to detect the locations at which the oscillatory pattern changes. Our method is shown to be able to detect all oscillatory frequencies and the corresponding change points within an accurate range with a prescribed probability asymptotically. A Gaussian approximation scheme and an overlapping-block multiplier bootstrap methodology for sums of complex-valued high dimensional non-stationary time series without variance lower bounds are established, which could be of independent interest. This study is motivated by oscillatory frequency estimation and change point detection problems encountered in physiological time series analysis. An application to spindle detection and estimation in electroencephalogram recorded during sleep is used to illustrate the usefulness of the proposed methodology.
 
\end{abstract}

\noindent%
{\it Keywords:} oscillation frequency detection; oscillation change point detection; non-stationary time series; time-frequency analysis; high dimensional Gaussian approximation; spectral domain methods.

\section{Introduction}
A major task in physiological time series analysis is to detect and estimate the complex oscillatory pattern of the observed stochastic process over time. In the past century, researchers have established various physiological knowledge about the complex oscillatory signals and its clinical applications; see, for example, dynamics in the breathing signal \citep{benchetrit2000} and electrocardiogram \citep{malik1996} for a far from exhaustive list of reference. However, there are still a lot left unknown when we encounter a physiological time series, probably due to its complicated characteristic structure. Specifically, common characteristics shared by physiological time series include but not exclusively the following. First, the time series is usually composed of multiple oscillatory components, and each component usually oscillates with time-varying frequency, amplitude, or even oscillatory morphology. Second, the signal is usually contaminated by non-stationary noise, and various artifacts. Moreover, the frequency or amplitude of an oscillatory component might abruptly jump from one to another. 

There have been quite a few analysis tools developed in the time-frequency (TF) analysis society (\cite{daubechies1992}, \cite{flandrin1998time}) toward studying this kind of time series with various clinical applications. However, TF analysis tools, particularly those nonlinear-type tools, have not been widely considered in the time series society and their statistical properties are largely unknown, except few current efforts; for example, \cite{Adak1998,Nason2000,chen2014,bruna2015intermittent,sourisseau2019,liu2022asymptotic}. It is interesting to ask if it is possible to apply existing TF analysis tools and the underlying ideas, to design more suitable statistical analysis and inferential tools. We shall mention that the local change point detection statistic can be viewed as detecting discontinuity by the Haar wavelet.
In this study, motivated by the clinical needs of detecting oscillatory components and quantifying its dynamics, we focus on two critical problems with a direct clinical interest. Specifically, under the oscillatory signal model with non-stationary noise model that we will introduce soon, we design a statistic to determine if there is an oscillatory component in a given physiological signal, and provide a strategy applying the wavelet analysis to decide if an oscillatory component has a change point behavior in its frequency over time.
To our knowledge, there exists no statistics literature on oscillation change point detection, except some remotely related literature (\cite{last2008detecting}, \cite{lavielle2000multiple} and \cite{preuss2015detection}). 
A detailed literature review of oscillation frequency and change point detection can be found in Section \ref{sec:lr} of the supplementary.

\subsection{Challenges in statistical analysis of complex oscillation}\label{sec:challenges}

Let $\{X_{i,n}:=\mu_{i,n}+\epsilon_{i,n}\}_{i=1}^n$, where $\E(\epsilon_{i,n})=0$, be an observed non-stationary time series and $\mu_{i,n}$ is a deterministic signal. A commonly used tool for oscillatory frequency detection is the {\em periodogram} defined as $I_{n,X}(\omega)=|\sum_{j=1}^nX_{j,n}e^{\sqrt{-1}\omega j}|^2/n$, $\omega\in[0,\pi]$. In principle, $I_{n,X}(\omega_0)$ should be large if $X_{i,n}$ is oscillatory at frequency $\omega_0$. However, non-stationarity in the noises $\epsilon_{i,n}$ as well as possible change points in the oscillation bring great challenges to the rigorous statistical analysis of complex oscillatory signals. We will discuss these challenges in detail in the next two subsections. 

\vspace{-0.5cm}

\subsubsection{Spectral dependency }\label{sec:sd_cp}
Let $\omega^*_j=2\pi j/n$, $j=1,2,\cdots,\lfloor n/2\rfloor$ be the {\em canonical frequencies}. A fundamental result for classic oscillatory frequency detection is that  $I_{n,X}(\omega^*_j)$, $j=1,2,\cdots,\lfloor n/2\rfloor$, are asymptotically independent when $\{X_{i,n}\}$ is stationary under some mild conditions. See for instance \cite{davis1999} and \cite{lin2009}. Consequently a Gumbel-type limiting distribution can be derived for the maximum deviation of the periodogram on the  canonical frequencies under the null hypothesis that there is no oscillation. Nevertheless, the periodograms on the canonical frequencies are no longer asymptotically independent for non-stationary time series.  
See, for instance, \cite{dwivedi2011} and \cite{zhou2014} for detailed calculations and discussions. Consequently, it has been a difficult and open problem to derive the maximum deviation of the periodogram on a dense set of frequencies for non-stationary time series. Without such results, it is difficult to distinguish peaks of the periodogram which reflect the oscillation from those caused by the noise. We refer to Sections \ref{sec:sd} and \ref{sec:class_freq} in the supplementary for numerical experiments on spectral dependency and the performance of classic frequency detection algorithms under non-stationarity.

\subsubsection{Spectral energy leak}\label{sec:sel}
One major challenge in oscillation change point detection is the spectral energy leak phenomenon.  
For $\omega\in[0,\pi]$, let $\{L_n(i,\omega):=\sum_{k=1}^i X_{k,n}e^{\sqrt{-1}\omega k}\}_{i=1}^n$ be the partial sum process of the Fourier transform of $X_{i,n}$ at frequency $\omega$. One of the most popular change point detection algorithms is the cumulative sum (CUSUM) test, which utilizes  $C_n(i,\omega):=[L_n(i,\omega)-iL_n(n,\omega)/n]/\sqrt{n}$.  In principle $C_n(i,\omega)$ should be small uniformly across $i$ if there is no oscillation change point at frequency $\omega$. Consider an example where $X_{i,n}=\cos(\omega^o i )+\epsilon_i$ with $\omega^o\in(0,\pi]$ and $\epsilon_i$ i.i.d. standard normal. Clearly $X_{i,n}$ is oscillating at $\omega^o$ without any change points. See Figure \ref{fig:energy_leak}  for an example of the heat map of $|C_n(i,\omega)|$.  We observe large values of $|C_n(i,\omega)|$ in a frequency band around  $\omega^o$,  although $C_n(i,\omega)$ are indeed uniformly small at frequency $\omega^o$. 
Therefore, the CUSUM test fails in this case as it produces strong false positive information near an oscillatory frequency. This problem persists if we apply other change point detection algorithms such as those based on binary segmentation or dynamic programming to the spectral domain. The cause of the problem is spectral energy leak in the sense that 
$\sum_{k=1}^i \cos (\omega^o k)e^{\sqrt{-1}\omega k}$ is a nonlinear function of $i$ with magnitude $O(|\omega-\omega^o|^{-1})$ if $i$ is large and $\omega$ is close (but not too close) to $\omega^o$. See Section \ref{sec:el} in the supplementary material for more discussions and a numerical experiment.

\subsection{Proposed methodology and its theoretical property}
In this paper, we devise a two-stage methodology for oscillatory frequency detection and change point estimation respectively which addresses the aforementioned challenges.

\subsubsection{The dense progressive periodogram test}\label{sec:DPPT}
One limitation of the classic oscillatory frequency detection algorithms is that only the canonical frequencies are considered. Hence the estimation accuracy is $O_\pr(1/n)$ which is slower than the parametric rate $O_\pr(n^{-3/2})$ for oscillatory frequency estimation \citep{genton2007}.  To address the latter issue as well as phase cancellation, we propose to investigate the progressive periodogram defined as $\{|L_n(i,\omega)|/\sqrt{n}\}_{i=1}^n$ on a dense grid of frequencies with mesh size no larger than $O(n^{-3/2})$ for the oscillatory frequency detection. {The reason for selecting a mesh size no larger than $O(n^{-3/2})$ is to make it possible for our method to detect the oscillation frequencies within an optimal $O_{\pr}(n^{-3/2})$ range}.
Note that in principle $\max_{1\le i\le n}|L_n(i,\omega)|/\sqrt{n}$ should be large if there are oscillations of sufficient length with possible occasional abrupt phase changes at frequency $\omega$.  

The key to the successful implementation of the above-mentioned dense progressive periodogram test (DPPT) is to investigate the maximum deviation 
\begin{eqnarray}\label{dppt}
F(W):=\max_{\omega\in W}\max_{1\le i\le n}|L_n(i,\omega)|/\sqrt{n}
\end{eqnarray} 
under the null hypothesis that there is no oscillation, where $W$ is a dense collection of frequencies with cardinality $p$. As we mentioned in Section \ref{sec:challenges}, the latter is a difficult problem due to spectral dependency caused by non-stationarity of $\{\epsilon_{i,n}\}$. In this paper, we tackle this problem by viewing (\ref{dppt}) as the maximum of a complex-valued high-dimensional dependent random vector of dimensionality $pn$. Then we utilize Stein's method of expectation approximation \citep{stein1986} to show that the law of $F(W)$ can be well approximated by that of the DPPT of a non-stationary Gaussian process which preserves the covariance structure of $\{X_{i,n}\}$. A high-dimensional extension to the overlapping-block multiplier bootstrap (OBMB) in \cite{zhou2013} is proposed to approximate the behaviour of the latter non-stationary Gaussian DPPT. We will show that the DPPT with the OBMB is able to detect oscillatory frequencies within an $O_\pr(n^{-1})$ range if there are changes in the phases of the oscillation. For oscillatory frequencies whose phases do not change over time, the DPPT is able to detect them within a nearly optimal  $O_\pr(n^{-3/2}\log n)$ range. In the special case where there is no oscillation, the DPPT will be shown to be able to accept the null hypothesis of no oscillation with a prescribed probability asymptotically.

In the literature, Gaussian approximations to the maximum of high-dimensional sums using Stein's method were investigated by, among others, \cite{chernozhukov2013} for independent data and \cite{zhang2017} and \cite{zhang2018} for time series. An important assumption in the latter papers is that the variances of the vector components should be bounded from both above and below and hence be balanced. Consequently, their results cannot be used directly for the DPPT since the variances of $L_n(i,\omega)$ are proportional to $i$ and hence are highly unbalanced across time.  In this paper, we generalize Nazarov's anti-concentration inequality \citep{nazarov2003} to the unbalanced variance case and extend the results of the aforementioned papers to complex-valued high-dimensional time series without variance lower bound. As many high-dimensional problems are without balanced variances among the vector components, our result may be of separate interest.  
  

\subsubsection{The phase-adjusted local change point detection algorithm}\label{sec:palcpd}
As mentioned in Section \ref{sec:sel}, the spectral energy leak phenomenon has to be carefully addressed in order to perform oscillation change point detection. 
One remedy to the latter phenomenon is to reduce error amplification of the Fourier transforms. 
This observation inspires us to consider a {\it local} change point detection algorithm. Specifically, for any oscillatory frequency $\hat \omega_k$ estimated from stage 1 and time point $i$, we utilize the norm of the difference between phase-adjusted local Fourier transforms 
\begin{eqnarray*}
T(i)=\frac{1}{\sqrt{2\tilde m}}\left|\sum_{l=i-\tilde m}^{i}e^{\sqrt{-1}\hat\omega_k(l-i)}X_l-\sum_{l=i+1}^{i+\tilde m+1}e^{\sqrt{-1}\hat\omega_k(l-i)}X_l\right|
\end{eqnarray*} 
to test whether there is an oscillation change point at time $i$ and frequency {$\omega_k$}, where $\tilde m$ is a bandwidth controlling the size of the local neighborhood. 
$T(i)$ should be large if the oscillatory pattern at frequency {$\omega_k$} changes at time $i$. Since $T(i)$ performs Fourier transforms only in a radius $\tilde m$ neighborhood of $i$, the {phases of the Fourier transforms} are amplified at most $\tilde m$ times uniformly over time. As we will require $\tilde m/n\rightarrow 0$, 
the energy leak problem is greatly reduced and will be shown to be asymptotically negligible. We adopt an extension of the OBMB with adjusted phases to approximate the maximum deviation of $T(i)$ uniformly across time $i$. We will show that the local change point detection algorithm has the correct Type-I error rate asymptotically if there is no change point at the oscillatory frequency. If there are oscillation change points, using a delicate empirical process theory for Fourier transforms of non-stationary time series, we will show that the latter algorithm is able to detect all change points within an $O(\log \mt)$ range with a pre-specified probability asymptotically, where the $O(\log \mt)$ rate is almost the parametric $O(1)$ rate for change point detection except a factor of logarithm.

The rest of the paper is organized as follows. Section \ref{sec:PLS} introduces a flexible class of non-stationary time series models for the noises $\{\epsilon_{i,n}\}$. In Section \ref{sec:two-stage}, we introduce the two-stage methodology in detail. Section \ref{sec:theory} investigates the consistency and accuracy of the proposed methodology. In particular, some optimality properties of the methodology are established. Section \ref{sec:simu} performs numerical experiments to investigate the finite sample property of the two-stage methodology. In Section \ref{sec:data}, we apply our methods to an electroencephalogram (EEG) recorded during sleep for spindle detection and estimation. Finally, literature review, additional simulation results and figures, Gaussian approximation and comparison schemes, and proofs of the theoretical results are put in the online supplementary material.

\section{Preliminaries: models for the signal and the noise}\label{sec:PLS}
Motivated by the aforementioned complex oscillatory pattern detection and estimation problems, consider an observed real time series $X=\{X_{i,n}\}_{i=1}^n$ that follows the model
 \begin{eqnarray}\label{eq:model}
    X_{i,n} = \mu_{i,n} + \epsilon_{i,n},
\end{eqnarray} 
where $\{\epsilon_{i,n}\}_{i = 1}^n$ is a centred non-stationary noise process whose data generating mechanism may evolve both smoothly and abruptly over time. The mean $\mu_{i,n}$ is assumed to be
\begin{eqnarray}\label{eq:mean}
\mu_{i,n} = \sum_{k = 1}^{|\Omega|} \sum_{r = 0}^{M_k} ( A_{r,k} \cos(  \omega_k i ) + B_{r,k} \sin( \omega_k i )) \mathbb{I}({b_{k,r} < i \leq b_{k,r+1} }) + f(i/n)\,,
\end{eqnarray}
where $\omega_k \in \Omega$, $\Omega$ is a finite set of  unknown oscillatory frequencies, $M_k\in \mathbb{N}$ is the number of segments for the $k$-th oscillatory component, $A_{r,k}, B_{r,k}\in \mathbb{R}$, $\mathbb{I}$ is the indicator function, $b_{k,1}<\cdots<b_{k,M_{k}}$ are the unknown change points corresponding to the oscillatory frequency $\omega_k$ with the convention $b_{k,0}=0$ and $b_{k,M_{k}+1}=n$ and $f$ is assumed to be a smooth function that ``oscillates slowly''. $f$ is usually understood as the {\em trend} or {\em baseline wandering} in biomedical signal processing. Here we assume that all $\omega_k\in\Omega$ are sufficiently high in the sense that $\min_{1\le k\le |\Omega|}\omega_k\ge \delta_0$ for some positive constant $\delta_0$ to distinguish the oscillation from the smooth trend $f$ that is represented by oscillations with frequencies much lower than $\delta_0$. {Note that while only the amplitude and phase jumps are modeled in \eqref{eq:mean}, the frequency jump has been captured implicitly. Indeed, if there is an oscillatory component with the frequency jump at time $c_{k}\in (0,1)$, it is equivalent to a summation of two oscillatory components with amplitude and phase jumps at time $c_k$. Thus, the frequency jump is a special case of \eqref{eq:mean} with different $\Omega$ and $M_k$.} The purpose of this paper is to test whether $\Omega$ is empty and if not, we would like to accurately estimate all $\omega_k$ and then test and locate the corresponding change points $\{b_{k,r}\}_{r=1}^{M_k}$. 

The rest of this section is devoted to the modelling of $\{\epsilon_{i,n}\}$. Many physiological time series, like breathing signal and photoplethysmogram, are oscillatory and contaminated by non-stationary noises with complex generating mechanisms. Signals like EEG is usually understood as stochastic, and the non-stationarity can be seen clearly from their time series plots, while {\em spindles} showing up during deep sleep can be modelled by oscillatory components. See Section \ref{sec:data}. For the sake of better modelling these time series, it is desirable to have a flexible non-stationary time series model for the stochastic process $\{\epsilon_{i,n}\}$, either noise or the stochastic part of EEG or others, which allows the underlying data generating mechanism to change both smoothly and abruptly over time.  To this end, we shall adopt the {\em piecewise locally stationary} (PLS) time series model proposed in \cite{zhou2013}.

We say $\{\epsilon_{i,n}\}_{i = 1}^n$ is PLS with $r$ break points (PLS($r$)) if there exist constants $0 = s_0 < s_1 <\ldots< s_r < s_{r+1} = 1$ and nonlinear filters (measurable functions) $\mathcal G_0,\ldots,\mathcal G_{r}$ such that 
\[	
\epsilon_{i,n} = \mathcal G_j(t_i,\mathcal{F}_i), \quad \text{if} \quad s_j < t_i \leq s_{j+1},	
\]	
$j=0,1,\cdots, r$, where $t_i = i/n$,  $\mathcal{F}_i = (\dots,e_0,...,e_i)$ and $e_i $ are i.i.d. random variables.  The function $\mathcal G_j$ is assumed to be smooth between $s_j$ and $s_{j +1}$ in some appropriate sense so that the time series is locally stationary between $s_j$ and $s_{j+1}$ in the sense that the data generating mechanism changes slowly. And the data generating mechanism changes abruptly from $\mathcal G_{j-1}$ to $\mathcal G_{j}$ at $s_j$, $j=1,2,\cdots, r$. 

A locally stationary model is more general compared with a stationary model as shown in the versatility of the generating function. A PLS($r$) model allows for additional abrupt changes in the mechanism of data generation adding to a more generalized model assumption. 
Assume that $\max_{1\le i\le n}\|\epsilon_{i,n}\|_q\le C_q$ for some finite constant $C_q$ and some $q\ge 4$, where $\|\cdot\|_q:=(\E|\cdot|^q)^{1/q}$ is the ${\cal L}^q$ norm of a random variable. We define the {\em physical dependence measures} for PLS time series as 
\[
\delta_{q}(k) =\max_{ 0 \leq j \leq r } {\sup_{i:s_j < t_i \le s_{j+1}} \| \mathcal{G}_{j}(t_i,\mathcal{F}_i) -\mathcal{G}_{j}(t_i,\mathcal{F}_{i,i-k}) \|_q},
\] 
where $\mathcal{F}_{i,i-k}$ is defined as $\mathcal{F}_{i,i-k} := (\dots,\hat{e}_{i - k},...,e_i)$, and $\hat{e}_{i - k}$ is an identically distributed copy of $e_{i-k}$ and is independent of $\{e_{j}\}_{j\in\mathbb{Z}}$.
Using the idea of coupling, $\delta_q(k)$ measures the influence of the innovations of the underlying data generating mechanism $k$ steps ahead on the current time series observation uniformly across time.   
We refer to \cite{zhou2013} and \cite{zhou2014} for more discussions and examples of PLS processes and the associated dependence measures.

\section{The two-stage methodology}\label{sec:two-stage}

\subsection{First stage: oscillatory frequency detection}\label{sec:stage1}

Given time series data $X=\{X_{i,n} \}_{i = 1}^n$, our first goal is to detect and then estimate the oscillatory frequencies. Recall from \eqref{dppt} that we require $W$ to be a dense set of possible oscillatory frequencies with mesh size no greater than $O(n^{-3/2})$. Denote $p := |W|$ as the size of the set of potential frequencies $W$.  Throughout this article we assume that $p$ is proportional to $n^{3/2}\log n$. In practice, one could pick 
\[ 
W = \{ \delta_0 n^{3/2} \log(n) ,  \delta_0 n^{3/2}\log(n) + 1,\dots, \delta_0 n^{3/2} \log(n)+\lfloor (\pi-\delta_0)n^{3/2}\log(n)\rfloor \}/(n^{3/2}\log(n) ),  
\] 
where $\delta_0$ is some small positive constant. The set of potential frequencies $W$ has its lowest frequency at $\delta_0$ as a rule of thumb in order to avoid small frequencies as the oscillation at the latter frequencies represents the smooth trend $f(\cdot)$. Regarding $\delta_0$, we could choose $\delta_0=n^{-1}$ if $f=0$. When $f\neq 0$ and is smooth, one could first difference the time series and obtain $\tilde X_{i,n}=X_{i,n}-X_{i-1,n}$ to reduce the impact of the trend $f(\cdot)$, since $|f(i/n)-f((i-1)/n)|=O(1/n)$. Note that $\frac{\cos(\omega i ) - \cos(\omega (i-1))}{\sin( \omega /2 )} = -2\sin( \omega ( i - 1/2) )$ and $\frac{\sin(\omega i ) - \sin(\omega (i-1))}{ \sin( \omega /2 )} = 2\cos( \omega ( i - 1/2) )$.
Hence, differentiation will not influence the oscillation as long as every $\omega\in\Omega$ is sufficiently high and satisfies $\omega\ge \delta_0$. Furthermore, we point out that in many real applications (such as the sleep spindle example considered in this article) the oscillatory frequencies of interest are typically known to be restricted to a certain range. In this case, $\delta_0$ can be chosen at the lower bound of the target range of the oscillation frequencies.  

The dense progressive periodogram test (DPPT) statistic
$F(W)$ in \eqref{dppt} should be large if there is an oscillation at some $\omega\in W$. In particular, $\bar{F}(\omega):=\max_{ 1 \leq k \leq n } |L_n(k,\omega)|/\sqrt{n}$ will show peaks at the oscillatory frequencies. Figure \ref{fig:2} shows a typical plot of $\bar{F}(\omega)$ for a time series with two oscillatory frequencies contaminated by PLS noises.

As we mentioned in the introduction, spectral dependency caused by the high density of $W$ as well as the non-stationarity of $\{\epsilon_{i,n}\}$ makes it difficult to investigate the limiting distribution of $F(w)$. As a solution to this conundrum we will use an extension of the OBMB \citep{zhou2013} to approximate the critical values of the DPPT under the null hypothesis of no oscillation. The bootstrap is simple to implement and will be shown to be able to detect all oscillatory frequencies accurately with a prescribed probability asymptotically. 

Define $S_{j,m}(w) = \sum_{i = j}^{j+m -1 } \sin( i w) X_{i,n}$, $C_{j,m}(w) =$ $\sum_{i = j}^{j+m -1 } \cos( i w) X_{i,n}$ and $E_{j,m}(w)=C_{j,m}(w)+\sqrt{-1}S_{j,m}(w)$ for an integer bandwidth $m<n$. %
The {\em OBMB} statistic is defined as 
\[ 
 \tilde{F}_{m,l}(W) = \max_{ w \in W} \tilde F_{m,l}(\omega),
\]
where 
\[
\tilde F_{m,l}(\omega):=\max_{ 1 \leq k \leq n -m }  \left\{   \left| \sum_{j= 1}^{k} E_{j,m}(w)G_{j,l} \right|  \right\} / \sqrt{m (n-m) }, 
\]
where $\{G_{j,l}\}_{j=1,2,\cdots, n-m, \, l=1,2,\cdots}$ are i.i.d. standard normal random variables independent of $\{{X_{i,n}}\}_{i=1}^n$. 
The DPPT is performed by simulating the distribution of $F(W)$ by that of the multiplier bootstrap statistic $\tilde{F}_{m,l}(W)$. Let $\{ \tilde{F}_{m,l}(W)\}_{l = 1}^K$ be $K$ (say, $K=1,000$) multiplier bootstrap statistics that we generate. Under a pre-specified significance level $\alpha \in (0,1)$ we estimate the $\alpha$-th critical value for $F$ using the empirical $(1- \alpha)$-th quantile 
\[ 
crit_{\alpha,1}(W) := Quantile(\{ \tilde{F}_{m,l}(W)\}_{l = 1}^K ,1 - \alpha) .
\] 
If $F(W) \leq crit_{\alpha,1}(W)$, then we claim that there is no oscillation. Otherwise, the first oscillatory frequency can be estimated by taking the frequency that maximizes $F(W)$: 
\[ 
\hat{\omega}_1 :=\{ \omega \in W | \bar F(\omega) = F(W) \} .
\]  

\begin{algorithm}[H]
\SetAlgoLined
\KwResult{estimated frequencies  $\hat{\Omega}$ }
 Let $W_1 = W$, $k = 1$ and $\hat{\Omega}_0 = \emptyset$\;
 Compute $F(W_1)$ and $crit_{\alpha,1}(W_1)$\;
 \While{$F(W_k) > crit_{\alpha,k}(W_k)$}{
 $\hat{\omega}_k = \arg\max_{\omega \in W_k} \bar F(\omega) $\;
 $\hat{\Omega}_k = \hat{\Omega}_{k-1} \cup \{\hat{w}_k\} $\;
  $W_{k+1} = W_k \setminus [\hat{\omega}_k - \log(m)/(4m^{1/2}),\hat{\omega}_k + \log(m)/(4m^{1/2}) ] $\;
  Compute $F(W_{k+1})$ and $crit_{\alpha,k+1}(W_{k+1})$\;
 Increase $k$ by 1\;  
 }
 \caption{Frequency Estimation}
\end{algorithm}

If there are multiple maximizers, let $\hat{\omega}_1$ be the smallest among them. Take $W_2 := W \setminus [\hat{\omega}_1- \log(m)/(4m^{1/2}), \hat{\omega}_1 + \log(m)/(4m^{1/2}) ]$. The second frequency estimate is taken to be  
\[ 
\hat{\omega}_2 := \{ \omega \in W_2 |\bar F(\omega) = F(W_2) \} 
\] 
if $F(W_2)> crit_{\alpha,2}(W_2)$, where $crit_{\alpha,2}(W_2)$ is defined in the same way. Repeat the previous steps until $F(W_k) \leq crit_{\alpha,k}(W_k)$ for some integer $k$. A heuristic reason that the multiplier 
bootstrap statistic works is because with high probability the conditional covariance structure of $\tilde{F}_{m,l}(W)$ can approximate that of the DPPT. See Theorem \ref{thm:stage1boots} in Section \ref{sec:theory} for a rigorous treatment.  
The algorithm is described in Algorithm 1.

In Algorithm 1, we remove a $\log m/(4m^{1/2})$ neighborhood from each estimated oscillatory frequency before we perform the next iteration. The reason is that it can be shown that 
\[
\tilde F_{m,l}(\omega)\gg Q_{1-\alpha}(\max_{ \theta \in W} \max_{ 1 \leq k \leq n } |\sum_{j=1}^k \epsilon_{j,n}e^{\sqrt{-1}\theta j}|/\sqrt{n})
\] 
with high probability if $|\omega-\omega_k|\ll m^{-1/2}$ for some $\omega_k\in\Omega$, where $Q_{1-\alpha}(Z)$ is the $(1-\alpha)$ quantile of a random variable $Z$. Here $a_n\gg b_n$ for positive sequences $a_n$ and $b_n$ means that $a_n/b_n\rightarrow\infty $ as $n\rightarrow\infty$. $a_n\ll b_n$ is defined similarly. Therefore, the bootstrap critical values will be too large if $\omega$ is within an $O(m^{-1/2})$ neighborhood of an oscillatory frequency. Hence, we remove a $\log m/(4m^{1/2})$ neighborhood of the estimated frequencies to avoid reduced sensitivity for oscillatory frequency detection. On the other hand, we need the oscillatory frequencies to be well separated in order for Algorithm 1 to detect all of them. In particular, we require that $|\omega_i-\omega_{j}|\gg m^{-1/2}\log m$ for all $\omega_i,\omega_j\in\Omega$, $\omega_i\neq\omega_j$.

\subsection{Second stage: oscillation change point testing and estimation}\label{sec:seconds}
Given the set of estimated frequencies $\hat{\Omega}$ (not empty), the second stage aims at testing the existence and then estimating the locations of any change points at each oscillatory frequency $w \in \Omega$. To this end, we propose a phase-adjusted local change point detection algorithm that compares the phase-adjusted Fourier transforms at frequency $\hat{w}$ before and after each time point. Specifically,  let $\hat{w} \in \hat{\Omega}$ and let $B_k$ be a set of potential change points at step $k$, $k=1,2,\cdots$. 
Let $B_1 =\{\mt + m',\mt+m'+1,\cdots,n - \mt -m' \}$, where $\mt,m'\in \mathbb{N}$.
The statistic for the second stage is defined as, \begin{align}\label{eq:stage2stat}  T(B_1,\hat w) =  &  \max_{i \in B_1 }  \left| \sum_{l = i - \mt  }^{i}  \exp( \sqrt{-1} \hat{w}( l- i)) X_{l,n} -\sum_{l = i +1 }^{i + \mt +1} \exp( \sqrt{-1} \hat{w} ( l- i) ) X_{l,n} \right| /\sqrt{2\mt},
\end{align} 
where $\mt$ is a bandwidth satisfying $\mt \rightarrow\infty$ with $ \mt /n\rightarrow 0$.

A phase-adjusted OBMB is employed here to approximate the critical values of $T(B_1,\hat w)$. Specifically, for $j \in B_1$, define 
\begin{align}
\Upsilon_{i}(j,\hat w) &\,= \left(\sum_{l = i -m'  }^{i }  e^{\sqrt{-1} \hat{w} (l -j )} X_{l,n} -\sum_{l = i +1 }^{i + m' +1} e^{\sqrt{-1} \hat{w} (l -j ) } X_{l,n} \right) /\sqrt{2m'},\label{eq:stage2boots1}
\end{align}
where $m'$ satisfies $m'\rightarrow\infty$ with $m'/\mt\rightarrow 0$. Now let $G_{i,j}$ be i.i.d standard Gaussian random variables independent of the data, where $i,j\in \mathbb{Z}$. We define the bootstrap statistic as 
\begin{align*} 
\hat{T}_{l}(B_1,\hat w) = & \max_{j \in B_1 }   \left| \sum_{i = j -\mt  }^{j }  \Upsilon_{i}(j,\hat w) G_{i,l} -\sum_{i = j +1 }^{j + \mt +1} \Upsilon_{i}(j,\hat w) G_{i,l}\right|/\sqrt{2\mt}\quad \mbox{where}\quad l\in \mathbb{N},
\end{align*} 
and simulate the distribution of $ T(B_1,\hat w)$ by generating $K_0$ bootstrap statistics $ \{\hat{T}_{l}(B_1,\hat w)\}_{l = 1}^{K_0}$. 

\begin{algorithm}[H]
\SetAlgoLined
\KwResult{estimated change points $\hat B$ at oscillatory frequency $\omega$. }
 Let $B_1 =\{\mt+ m',\mt +m'+1,\dots,n - \mt-m' \}$, $k = 1$ and $\hat{B}_0 = \emptyset$\;
 Compute $T(B_1,\hat w)$ and $crit_{\beta,1}(B_1)$\;
 \While{$T(B_k,\hat w) > crit_{\beta,k}(B_k)$}{
 obtain $\hat b_k$ using (\ref{eq:changepoint})\;
 $\hat{B}_k = \hat{B}_{k-1} \cup \hat{b}_k $\;
  $B_{k+1} = B_k \setminus [\hat{B}_k - \mt ,\hat{B}_k + \mt] $\;
  Compute $T(B_{k+1},\hat w)$ and $crit_{\beta,k+1}(B_{k+1})$\;
 Increase $k$ by 1\;
  }
 \caption{Change Point Estimation}
\end{algorithm}

Next find the estimated critical value under $\beta$ significance level for the first change point:
\begin{align*}
&crit_{\beta,1}(B_1) 
= Quantile \left( \{\hat{T}_{l}(B_1,\hat w)\}_{l = 1}^{K_0},{1 - \beta} \right). 
\end{align*}
We claim that there is no change point at frequency $\omega$ at level $\beta$ if $T(B_1,\hat w) \le crit_{\beta,1}(B_1)$. Define the first change point estimator as
\begin{align}\label{eq:changepoint}  
\hat{b}_{1}:=  & \arg\max_{i \in B_1 } \left| \sum_{l = i -\mt  }^{i }  \exp(\sqrt{-1}\hat{\omega}(l-i)) X_{l,n} -\sum_{l = i +1 }^{i + \mt +1} \exp(\sqrt{-1}\hat{\omega}(l-i)) X_{l,n} \right|/\sqrt{2\mt} 
\end{align}
provided that $T(B_1,\hat w) > crit_{\beta,1}(B_1)$. If there are multiple maximizers, then let $\hat b_1$ be the smallest among them. Define $B_2 := B_1 \setminus [\hat{b}_{1} - \mt,\hat{b}_{1} + \mt] $, which is the second potential set of change point. We iterate the aforementioned procedure until $T(B_k,\hat w) \le crit_{\beta,k}(B_k)$ for some $k$. The detailed algorithm is listed in Algorithm 2.

\begin{remark}
Let $Z_{i,n}$ be a zero-mean short memory non-stationary process and define its partial sum process 
$\{S_Z(i)\}_{i=1}^n=\{\sum_{j=1}^iZ_{j,n}/\sqrt{n}\}_{i=1}^n. 
$
The essence behind the OBMB and the multiplier bootstrap in \cite{zhou2013} is the observation that the covariance structure of $\{S_Z(i)\}$ can be well approximated by the conditional covariance structure of 
\begin{eqnarray}\label{eq:114}
\{\chi_Z(i)\}:=\{\sum_{j=1}^i\mathcal{S}_{j,m}G_j/\sqrt{n}\},
\end{eqnarray}
where $\mathcal{S}_{j,m}=\sum_{k=j}^{j+m-1}Z_{k,n}/\sqrt{m}$ and $G_j$ are i.i.d. standard normal random variables independent of $\{Z_{i,n}\}$. Careful observations of the bootstrap statistics proposed in Sections \ref{sec:stage1} and \ref{sec:seconds} reveal that they are essentially localized, differenced or phased-adjusted versions of \eqref{eq:114} with $Z_{k,n}$ replaced by $X_{k,n}\exp(\sqrt{-1}\omega k)$.
\end{remark}

\subsection{Tuning parameter selection}\label{sec:tuning}
To implement our methodology, one needs to select three tuning parameters: $m$, $\mt$ and $m'$. In order to choose $m$ in Stage 1, we observe that the accuracy of the bootstrap is determined by how well the conditional covariance of $\sum_{j= 1}^{k} E_{j,m}(\omega)G_{j,l}/\sqrt{m(n-m)}$ approximates the covariance of $L_n(k,\omega)/\sqrt{n}$. The block size $m$ determines the latter accuracy. We propose and utilize a plug-in block size selector that minimizes the (asymptotic) mean squared error (MSE) of the aforementioned approximation. Due to page constraints, a detailed description of the plug-in block size selector for general non-stationary time series can be found in Section \ref{tps} of the supplementary material. In order to utilize the plug-in block size selector to select $m$ in stage 1, we vectorize $L_n(n,\omega)$ over all candidate frequencies and write
${\cal L}_n=(\Re{L_n(n,\omega_1)},\Im{L_n(n,\omega_1)},\cdots, \Re{L_n(n,\omega_p)},\Im{L_n(n,\omega_p)})^\top$, where $\omega_1<\cdots<\omega_p$ are the elements of $W$ and $\Re$ and $\Im$ represent the real and imaginary parts of a complex number, respectively. Observe that the accuracy of the OBMB is determined by how well the quantity $\sum_{j=1}^{n-m}{\cal R}_{j,m}{\cal R}_{j,m}^\top/(m(n-m))$ approximates $\E[{\cal L}_n{\cal L}^\top_n]/n$, where 
${\cal R}_{j,m}=(C_{j,m}(\omega_1),S_{j,m}(\omega_1),\cdots,C_{j,m}(\omega_p),S_{j,m}(\omega_p))^\top$. The latter formulation matches exactly that in Section \ref{sec:mts} of the supplementary. Therefore the plug-in block size selector can be applied to selecting $m$ in stage 1. We refer the readers to the supplementary for the detailed implementation. Similarly, the tuning parameter $m'$ in stage 2 can also be selected by the plug-in method. In this case, we choose $m'$ that minimizes the (asymptotic) MSE of $\sum_{j=1}^{n-2m'}{\cal V}_{j,m'}{\cal V}_{j,m'}^\top/(2m'(n-2m'))$ as an estimator of the long-run covariance of the sequence $\{(\cos(\hat\omega i)X_{i,n},\sin(\hat\omega i)X_{i,n})^\top\}_{i=1}^n$, where ${\cal V}_{j,m'}=(C_{j,m'}(\hat\omega)-C_{j+m',m'}(\hat\omega),S_{j,m'}(\hat\omega)-S_{j+m',m'}(\hat\omega))^\top$. As a result, the plug-in block size selection method in Section \ref{tps} of the supplementary, particularly in Sections \ref{sec:dlrv} and \ref{sec:mts} can be applied to select $m'$. We again refer the readers to the supplementary for the detailed implementation.

We now discuss the choice of $\mt$ in stage 2. Here we shall use a penalized (BIC) method for the purpose. The rationale is to balance the goodness-of-fit and the complexity of the change point model. Observe that if $\mt$ is too small, the change point algorithm is unstable due to the fact that the phase-adjusted local change point algorithm is performed over 
narrow local neighborhoods with large smoothing errors. In this case, one will observe either a large lack-of-fit due to some undetected true change points or a large model complexity due to some spurious change points detected. On the other hand, when $\mt$ is too large, the effect of energy leak kicks in and several spurious change point will be detected which increases the model complexity without improving the goodness-of-fit. To implement the algorithm, we first set the candidate pool of $\mt$ as $(\lfloor n^{16/29}\rfloor, \lfloor n^{2/3}\rfloor)$ as discussed in Lemmas \ref{lem:b9} and \ref{lem:10} in the supplementary. For each $m^*$ in this candidate pool, define the penalized goodness-of-fit measure (BIC)
$$GM(m^*)=n\log(SSE(m^*)/n)+M(m^*)\log n,$$
where $M(m^*)$ and $SSE(m^*)$ are the number of change points detected and the residual sum of squares of the segmentation when the block size is selected at $m^*$, respectively. In particular, to avoid the energy leak problem, at each point $i$ the fitted value $\widehat{e^{\sqrt{-1}\hat\omega}X_{i,n}}$is defined as the average of $\{e^{\sqrt{-1}\hat\omega}X_{i,n}\}$ {over the nearest neighborhood} of $i$ of length $\lfloor n^{2/3}\rfloor$ such that the latter neighborhood is within the same segment as $i$. Here the distance between a point and an interval is defined as the distance between the point and the middle point of the interval. The $SSE(m^*)$ is then defined as $\sum_{i=1}^n|e^{\sqrt{-1}\hat\omega}X_{i,n}-\widehat{e^{\sqrt{-1}\hat\omega}X_{i,n}}|^2$.   
We propose to select $\mt$ that minimizes $GM(m^*)$. If there are multiple minimizers $m^*_1<m^*_2<\cdots<m^*_d$ with $d>1$, then select $\mt=m^*_{\lfloor (d+1)/2\rfloor}$. In our simulation studies and data analysis, the plug-in block size selector and the penalized (BIC) selector work well. We also refer to Section A.5 in the supplementary for a sensitivity analysis of our tuning parameter selection methods.

\section{Theoretical results}\label{sec:theory}

In Section \ref{sec:nullboots}, we demonstrate the asymptotic consistency of the OBMB for both stages. 
Those results establish that the OBMB can well approximate the probabilistic behavior of the DPPT and the phase-adjusted local change point detection algorithms asymptotically under the corresponding null hypotheses of no oscillation at frequencies $\ge\delta_0$ and no change points.   
In the literature, $L^\infty$ Gaussian approximations to high dimensional time series  are typically realized by the {\em non-}overlapping multiplier bootstrap (cf. \cite{jirak2015} and \cite{dette2018}). The non-overlapping multiplier bootstrap could suffer from a relatively small number of blocks for time series of moderate length and hence may be numerically unstable under such circumstances. To remedy the latter problem, the OBMB is implemented and theoretically justified. Under the alternative hypotheses where oscillations at frequencies $\ge\delta_0$ or oscillatory change points exist, in Section \ref{sec:ea} we further investigate the estimation accuracy of our bootstrap-assisted two-stage algorithm. We establish that the estimation accuracy of our two-stage algorithm is of nearly parametric rate except a factor of logarithm.

{Before stating our theorems, some assumptions for the time series are placed in order. 
In the sequel, we shall omit the subscript $n$ in $X_{i,n}$ and $\epsilon_{i,n}$ in the time series model \eqref{eq:model} to simplify notation. Meanwhile, the symbol $C$ denotes a generic positive and finite constant which may vary from place to place. 

\begin{assumption} \label{assumption1}
Assume  that $\max_{1\leq i \leq n}  \E X^4_{i} < c_1$ for $c_1>0$ and there exists $\mathcal{D}_n >0$ such that one of the following two conditions holds:
	\begin{align}\label{eq12}
	\max_{1 \leq j \leq n} \E \exp( |X_j| / \mathcal{D}_n ) \leq 1,
	\end{align}	 
	\begin{align}\label{eq11}
	\max_{1 \leq j \leq n} \E g( |X_j| / \mathcal{D}_n ) \leq 1,
	\end{align}
    for some strictly increasing convex function $g$ defined on $[0,\infty)$ satisfying $g(0)=0$.
    
\end{assumption}

\begin{assumption} \label{assumption2}  Assume there exist $M = M(n) >0$ and $\gamma = \gamma(n) \in (0,1)$ such that 
	\begin{align}
	& n^{3/8} M^{-1/2} l_n^{-5/8} \geq C_1 \max \{\mathcal{D}_n l_n ,l_n^{1/2} \}  \textit{ under Condition (\ref{eq12})}
	\\
	&n^{3/8} M^{-1/2} l_n^{-5/8} \geq C_2 \max \{\mathcal{D}_n g^{-1}(n/\gamma) ,l_n^{1/2} \}  \textit{ under Condition (\ref{eq11})}
	\end{align}	
	for $C_1,C_2 > 0$, where $\mathcal{D}_n$ is given in Assumption 1, and $l_n = \log (pn/\gamma_n) \vee 1$ with $p=|W|$. In both cases, suppose $n^{7/4} M^{-2} l_n^{-9/4} \geq C_3 >0$.
\end{assumption}

\begin{assumption}  \label{assumption3}
Assume that 
\begin{align}
	&\max_{1 \leq i \leq n} \|\epsilon_{i}\|_q \le C_q \quad \text{ and } \quad  \delta_q(k)=O((k+1)^{-d})
	\end{align} 
	for some finite constants $C_q>0$, $q\ge 4$ and $d\ge 5$. Define $\Theta_{k,q} := \sum_{l = k}^{+\infty} \delta_q(l) $. Further assume that, for some finite constant $C>0$,
	\begin{eqnarray}\label{eq:SLC}
	\|\mathcal G_j(t,\FF_0)-\mathcal G_j(s,\FF_0)\|_4\le C|t-s|,\quad t,s\in [s_j,s_{j+1}],\quad j=0,1,\cdots, r.
	\end{eqnarray}

\end{assumption}

\begin{assumption}  \label{assumption4}
$\min_{\omega\in W}\omega= \delta_0$ for some positive constant $\delta_0$. 
\end{assumption}

\begin{assumption}  \label{assumption5}
$f(\cdot)$ is twice differentiable on [0,1] with Lipschitz continuous $f''(\cdot)$. \end{assumption}

\begin{assumption}  \label{assumption6}
Let $v(t, \omega)=\sum_{k=-\infty}^\infty\mbox{Cov}(\mathcal G_j(t,\FF_0), \mathcal G_j(t,\FF_{k}))\exp(\sqrt{-1}k\omega)$, where $t\in(s_j,s_{j+1}]$, $j=0,1,\cdots, r$, be the spectral density of $\{\epsilon_{i}\}$ at time $t$ and frequency $\omega$. 
Assume that there is a positive constant $\delta_1$ such that $v(t, \omega)\ge\delta_1$ for all $t\in[0,1]$ and $\omega\in (0,\pi).$ 
\end{assumption}

\begin{assumption}  \label{assumption7}
For $k=0,1,\cdots,$ and $j=0,\cdots,r$, let $\gamma_{j,k}(t)=\cov(\mathcal G_j(t,\FF_0), \mathcal G_j(t,\FF_{k}))$. Assume that, for each $j$ and $k$, 
$\gamma_{j,k}(t)$ is twice continuous differentiable on $[s_{j},s_{j+1}]$ with Lipschitz continuous second derivatives.
\end{assumption}

\begin{assumption}  \label{assumption8}
If $\Omega$ is not empty, then there exist a positive constant $\delta_2$ such that $\delta_0\le \omega_i\le \pi-\delta_2$, $\forall \omega_i\in\Omega$. And there exists a positive constant $c$ such that $|\omega_i-\omega_j|\ge c/\log n$ for any $\omega_i$, $\omega_j\in\Omega$ and $\omega_i\neq \omega_j$.
\end{assumption}

\begin{assumption} \label{assumption9}
 If $\Omega$ is not empty, then $\forall \omega_k\in\Omega$, assume that $|A_{0,k}|+|B_{0,k}|\ge \delta_3>0$ if there is no change point at frequency $\omega_k$. If there are change points at frequency $\omega_k$, then we assume that $b_{k,r}=c_{k,r}n$, where $0<c_{k,1}<\cdots<c_{k,M_k}<1$ are constants, and $[|A_{r,k}-A_{r-1,k}|^2+|B_{r,k}-B_{r-1,k}|^2]^{1/2} \ge \delta_4$, $r=1,2,\cdots, M_{k}+1$, where $\delta_4>0$ is a constant. 
\end{assumption}

Some discussions of these assumptions are in order. 
Assumptions \ref{assumption1} and \ref{assumption2} are regularity conditions for the high dimensional Gaussian approximation. They correspond to Assumptions (A1) and (A2) in Section \ref{sec:A}. Typical choice of $g$ in \eqref{eq11} is the power function $g=x^q$, $q\ge 1$. 
$M$ in Assumption \ref{assumption2} is the dependence truncation constant used in the Gaussian approximation proof, and we put a relatively weak constraint on $M(n)$, $\gamma(n)$ and $l_n$ relative to the data length $n$. Assumption \ref{assumption3} puts constraints on the moments and dependence of the noise sequence $\{\epsilon_{i}\}$. It requires that $\epsilon_{i}$ has finite $q$-th moment for $q\ge 4$ and weak dependence which decays at a sufficiently fast algebraic rate. Equation \eqref{eq:SLC} is a piece-wise stochastic Lipschitz continuity condition which guarantees that the data generating mechanism of the noise process is smooth between adjacent jump points. Assumption \ref{assumption4} requires that our candidate frequencies to be separated from 0 in order to avoid detecting variations caused by the smooth trend $f(\cdot)$.  Observe that $v(t,\omega)$ in  Assumption \ref{assumption6} is the instantaneous spectral density of the noise process $\{\epsilon_{i}\}$ at time $t$ and frequency $\omega$. Hence Assumption \ref{assumption6} is a mild condition that requires that the  instantaneous spectral density of the noise process $\{\epsilon_{i}\}$ is uniformly positive. Assumption \ref{assumption7} requires that the auto-covariances of the noise process are piece-wise twice differentiable with respect to time with piece-wise Lipschitz continuous derivatives. Assumption \ref{assumption8} puts a positive lower bound on the oscillatory frequencies to distinguish the oscillation from the smooth trend $f(\cdot)$. Meanwhile, a mild technical assumption is put in Assumption \ref{assumption8} to separate the oscillatory frequencies from $\pi$. Assumption \ref{assumption8} also requires that the oscillatory frequencies to be well separated by at least $O(1/\log n)$. Assumption \ref{assumption9} requires that the amplitudes of the oscillations have a positive lower bound. Meanwhile, the oscillation change points, when they exist, are required to be well separated with jump sizes larger than a positive constant.}

\subsection{Bootstrap consistency}\label{sec:nullboots}

Recall the definitions of $F(W)$ and $\tilde{F}(W)$ in Section \ref{sec:stage1} where $\tilde{F}(W)$ is short for $\tilde{F}_{m,l}(W)$. Observe that both $F(W)$ and $\tilde{F}(W)$ can be viewed as the maximum of the coordinate-wise sums of a high-dimensional vector. 
Define $n$-dimensional vectors $C^{(\epsilon)}(\omega)$ and $S^{(\epsilon)}(\omega)$ with the $k$-th vector coordinate 
\[
C_{k}^{(\epsilon)}(\omega) := \sum_{j= 1}^k  \cos(  j \omega) \epsilon_j \quad \text{ and } \quad S_{k}^{(\epsilon)}(\omega) := \sum_{j= 1}^k  \sin(  j \omega) \epsilon_j, 
\] 
where $k=1,2,\cdots, n$ respectively. 
Let  
\[ 
\Theta^{(\epsilon)}(W) := [C^{(\epsilon)}(\omega_1)^\top,S^{(\epsilon)}(\omega_1)^\top,\dots, C^{(\epsilon)}(\omega_p)^\top,S^{(\epsilon)}(\omega_p)^\top]^\top/\sqrt{n}\in \mathbb{R}^{2np},
\]
where $W=\{\omega_1,\cdots,\omega_p\}$. Recall the definition of $W$ at the beginning of Section \ref{sec:stage1}. For any given $w\in W$, define the $2n$-dim vector $S(w)$ with the $k$-th vector coordinate
\[ S_k(w) := \begin{cases}
\ C_{1,m}(w)G_1 &\text{ if } k\leq m\\
\ \sum_{j= 1}^{k - m+1} C_{j,m}(w)G_j &\text{ if }  m < k\leq n \\
\ S_{1,m}(w)G_1 &\text{ if } n + 1 \leq k\leq n + m +1\\
\ \sum_{j= 1}^{k - m+1} S_{j,m}(w)G_j &\text{ if }  n + m +2 \leq k\leq 2n \,,  \\
\end{cases} 
\]
where we omit the subscript $l$ in $G_{j,l}$ to simplify notation. Finally, we define $ S(W) =  [S(w_1)^\top,\dots,S(w_p)^\top]^\top/\sqrt{m(n-m)}\in \mathbb{R}^{2np}$. 
Define the following measure of difference in the covariance structure between two statistics $F(W)$ and $\tilde{F}(W)$ as 
\[ 
\Delta := \max_{1\le i,j \leq 2pn}|[\cov(\Theta^{(\epsilon)}(W)) - \cov(S(W)|X)]_{ij}|, 
\]
where $\cov(A)$ denotes the covariance matrix of a random vector $A$, and for a matrix $D$, $D_{ij}$ denotes the entry of $D$ at its $i$-th row and $j$-th column. 
Now we are ready to present the first main theorem.

 \begin{theorem}\label{thm:stage1boots}
	Suppose Assumptions \ref{assumption1}-\ref{assumption7} 
	hold true. Assume that $\Omega$ is empty in \eqref{eq:mean}. Further assume that
	\begin{eqnarray}\label{eq:thmcon}	
		m\asymp n^\theta \mbox{ with } 0<\theta<1\,\, \mbox{ and }\,\, q>8\log p/((1-\theta)\log n).
\end{eqnarray}
Then, we have 
	$	
	\Delta  =   O_\pr(p^{4/q} \sqrt{m/n}  + 1/m ) \,.
	$
Define the sequence of events $A_n = \{ \Delta  \leq (p^{4/q} \sqrt{m/n} + 1/m  ) h_n \}$, where $h_n>0$ is a sequence diverging at an arbitrarily slow rate so that $\pr(A_n) = 1 - o(1)$. {For every outcome $\boldsymbol{\omega}\in A_n$}, we have 
	\begin{align}\label{eq:boundthm1}
	&\sup_{|x| > d^\circ_{n,p} } \left| \pr( F(W) \leq x) -  \pr( \tilde{F}(W) \leq x | X({\boldsymbol{\omega})} ) \right|\nonumber\\ 
	\lesssim\,&  (p^{4/q} \sqrt{m/n} + 1/m  )^{1/3}h_n^{1/3}  \log^2(pn) + G^*(n,np), \mbox{ where }
	\end{align} 
	$$
	d^\circ_{n,p}=C[(p^{4/q} \sqrt{m/n} + 1/m  )^{1/3}h_n^{1/3} \log^{1/6}(pn)+  \log^{-1/2}(np)+d^*_{n,np}]
	$$ 
	with some finite constant $C$ that does not depend on $n$, and $G^*(n,np)$ as well as $d^*_{n,np}$ are defined in Proposition \ref{prop:complex_approximation} in the supplementary material with $h$ therein replaced by $np$. 
\end{theorem}

In \eqref{eq:boundthm1}, the term $G^*(n,np)$ corresponds to the Gaussian approximation error of $F(W)$ by 
$\max_{ \theta \in W} \max_{ 1 \leq k \leq n } |\sum_{j=1}^k y_{j,n}e^{\sqrt{-1}\theta j}|/\sqrt{n}$, where $\{y_{i,n}\}$ is a centered Gaussian time series that preserves the covariance structure of $\{\epsilon_{i}\}$. The term $(p^{4/q} \sqrt{m/n} + 1/m  )^{1/3}h^{1/3}(n)  \log^2(pn)$ corresponds to the bootstrap error when approximating the distribution of $\max_{ \theta \in W} \max_{ 1 \leq k \leq n }$ $|\sum_{j=1}^k y_{j,n}e^{\sqrt{-1}\theta j}|/\sqrt{n}$ by that of $\tilde{F}(W)$.
The requirement $q>4\log p/((1-\theta)\log n)$ in Theorem \ref{thm:stage1boots} is to ensure that $p^{4/q} \sqrt{m/n}  + 1/m $ converges to 0 polynomially fast. For the DPPT, we set $p\asymp n^{3/2}\log n$, in which case the above requirement is equivalent to $q>12/(1-\theta)$.
As we will discuss after the proof, $G^*(n,np)$ and $d^\circ_{n,p}$ converge to $0$ when $n\to \infty$.

As a result, Theorem \ref{thm:stage1boots} asserts that under the null hypothesis of no oscillation at frequencies $\ge \delta_0$ the conditional cumulative distribution function of the bootstrap well approximates that of $F(W)$ with high probability if $|x|>d^\circ_{n,p}$ as the right hand side of \eqref{eq:boundthm1} converges to 0. The restricted range $|x|>d^\circ_{n,p}$ is due to the unbalanced variances of $L_n(i,\omega)$ across $i$, where there is no positive lower bounds. 
The unbalanced variances of $L_n(i,\omega)$  lead to the possible failure of the Gaussian approximation when $|x|$ is very small. 
On the other hand, Lemma \ref{lem6} in the supplementary material assures that, for $\alpha\in (0,1-\alpha_0]$ where $\alpha_0$ is any positive constant, the $(1-\alpha)$ quantile of $F(W)$ is no less than $c_\alpha\sqrt{\log n}$ for sufficiently large $n$ and some positive constant $c_\alpha$. Hence $d^\circ_{n,p}$ is dominated by the latter quantile and Theorem \ref{thm:stage1boots} can be used for the DPPT.

The next theorem shows the validity of the phase-adjusted OBMB for the second stage statistics under the null hypothesis of no change points at an oscillatory frequency $\omega$. We first need to introduce some notation. 
Let $\Theta^{(2)} \in \R^{ 2(n - \mt - m') }$ be the vectorized stage 2 statistics at the true oscillatory frequency $w$; that is,  for $\mt + m' \leq i \leq n - \mt - m'$ we define the coordinate of $\Theta^{(2)}$ as \[ \Theta^{(2)}_{i - (\mt + m') + 1 } =  \left( \sum_{k = i -\mt  }^{i }  \cos({ w (k - i ) }) \epsilon_k -\sum_{k = i +1 }^{i + \mt +1} \cos({ w (k - i ) }) \epsilon_k \right)/\sqrt{2\mt}  \] and \[ \Theta^{(2)}_{i + n- 2(\mt + m') + 1  } =  \left( \sum_{k = i - \mt  }^{i }  \sin({ w (k - i ) }) \epsilon_k -\sum_{k = i +1 }^{i + \mt +1} \sin({ w (k - i ) }) \epsilon_k \right)/\sqrt{2\mt}. \] 
We can define the vectorized multiplier bootstrap statistics at the true oscillatory frequency $w$ in a similar way. Let $\tilde{S}^{(2)} \in \R^{ 2(n - \mt - m') }$, which is defined as 
\[
\tilde{S}^{(2)}_{i - (\mt + m') + 1 } = \left( \sum_{k = i - \mt  }^{i }  \Phi_{k}(i, w) G_k -\sum_{k = i +1 }^{i + \mt +1}  \Phi_{k}(i, w) G_k \right)/\sqrt{2\mt} 
\] 
and 
\[ 
\tilde{S}^{(2)}_{i + n- 2(\mt + m') + 1  } = \left( \sum_{k = i - \mt  }^{i }  \Psi_{k}(i, w) G_k -\sum_{k = i +1 }^{i + \mt +1}  \Psi_{k}(i, w) G_k \right)/\sqrt{2\mt}
\]  
for $\mt + m' \leq i \leq n - \mt - m'$. Lastly, define a measure of difference in covariance structure for the second stage statistics: 
\begin{eqnarray}\label{eq:Delta'}
\Delta' := \max_{1\leq i,j \leq 2(n - \mt - m')}|[\cov(\Theta^{(2)}) - \cov(\tilde{S}^{(2)}|X)]_{ij}|. 
\end{eqnarray}

\begin{theorem} \label{thm:s2_consistancy}
Assume that $\Omega$ is not empty, $\omega\in\Omega$ and there is no change point at frequency $\omega$. Suppose that Assumptions \ref{assumption3} to \ref{assumption8} 
hold true. Further assume that $\mt\asymp n^{\gamma_1}$ with $16/29<\gamma_1<1$, $m'\asymp n^\eta$ with $0<\eta<\gamma_1$ and $q>4/(\gamma_1-\eta)$. Then, one can find a sequence of events $B_n$ with probability at least $1-C/\log^{q/2}n$, where $C>0$ does not depend on $n$, such that {for every outcome $\boldsymbol{\omega}\in B_n$}, we have $\Delta'({\boldsymbol{\omega}})\lesssim 1/m' + n^{2/q} \sqrt{m'/\mt}\log n$ and  
\begin{equation}\label{eq:s2null}
\sup_{x\in\mathbb{R} } \left| \pr( T(B_1,\hat w) \leq x) -  \pr( \hat{T}(B_1,\hat w) \leq x | X({\boldsymbol{\omega}}) ) \right|=o(1). 
\end{equation}	
\end{theorem}

Theorem \ref{thm:s2_consistancy} implies that the phase-adjusted OBMB achieves the correct Type-I error rate $\beta$ asymptotically if there is no change point at frequency $\omega$. The assumption $q>4/(\gamma_1-\eta)$ ensures that $n^{2/q}\sqrt{m'/\mt}\log n$, hence $\Delta'$, converges to 0 algebraically fast. The $o(1)$ error in \eqref{eq:s2null} is composed of three parts. First, we show that 
$\sup_{x\in\mathbb{R} }| \pr( T(B_1,\hat w) \leq x)$ $-  \pr( T(B_1,w) \leq x) |$ converges to 0 and, conditional on the data, 
$\sup_{x\in\mathbb{R} } | \pr( \hat T(B_1, \hat w) \leq x| X) -  \pr( \hat T(B_1,w) \leq x| X) |$ converges to 0 algebraically fast with high probability. Second, we derive that
$$\sup_{x\in\mathbb{R} } | \pr( T(B_1, w) \leq x) -  \pr( T^{(y)}(B_1,w) \leq x) |=o(1), \mbox { where }$$ 
$T^{(y)}(B_1,w )$ is the version of $T(B_1, w)$ with $\{X_{l}\}$ therein replaced by a centered Gaussian process $\{y_{l,n}\}$ with the same covariance structure. Then utilizing comparison of distribution results for complex Gaussian random vectors, we show that, with high probability, 
$$\sup_{x\in\mathbb{R} } \left| \pr( T^{(y)}(B_1,w) \leq x) -  \pr( \hat{T}(B_1,w) \leq x | X ) \right|=O((\Delta')^{1/3}\log^{7/6} n).$$

\subsection{Estimation accuracy}\label{sec:ea}
Under the alternative hypotheses where oscillations at frequencies $\ge \delta_0$ or oscillatory change points exist, this subsection investigates the accuracy of our bootstrap-assisted algorithm.
Here the accuracy of our methodology is composed of two parts. First, for given $\alpha,\beta\in (0,1)$, we hope that our methodology is able to estimate the correct number of oscillatory frequencies with probability $1-\alpha$ and the correct number of change points with probability $1-\beta$ asymptotically. Second, we would like our methodology to estimate the oscillatory frequencies and change points, if they exist, within an accurate range. The following Theorem \ref{thm:s1_accuracy} and Proposition \ref{prop:nophasechange} establish the desired result for our stage 1 methodology.

\begin{theorem}[Accuracy of frequency estimation]	\label{thm:s1_accuracy}
	Assume that Assumptions \ref{assumption1} to \ref{assumption9} in Section \ref{Section:AppenC} 
	hold true.  Then  when $n\to \infty$,

	1. if $\Omega$ is empty, then $\pr( |\hat{\Omega}| = 0) \rightarrow 1-\alpha $;
	
	2. if $|\Omega| > 0$, 
	then 
	$  
	 \pr \left( \max_{k} |\hat{w}_k - \omega_k| \lesssim  n^{-1}h_n,\,\, |\hat{\Omega}| = |\Omega|  \right) \rightarrow 1 - \alpha
	 $
for any sequence $h_n>0$ that diverges to infinity arbitrarily slowly.	
\end{theorem}

The rate $n^{-1}h_n$ is slower than the parametric rate $n^{-3/2}$ for oscillatory frequency estimation (cf. \cite{genton2007}). A further investigation reveals that this is caused by possible changes in the phase of the oscillation. When there is an abrupt change in the phase, the Fourier transformation of the oscillation curve is not maximized at the true oscillatory frequency. Instead, it will be maximized in an $O(1/n)$ neighborhood of the oscillatory frequency. The following is a detailed example.

\begin{example}
Let $\mu_i=C_1\cos(\omega_0 i+\alpha_1)$, $i=1,2,\cdots, n_1$, and $\mu_i=C_2\cos(\omega_0 i+\alpha_2)$, $i=n_1+1,\cdots,n$, where $C_1,C_2,\alpha_1,\alpha_2>0$, $0<\omega_0<\pi/2$ and $0<\alpha_1-\alpha_2<\pi/2$. Write $n_2=n-n_1$ and assume $n_1=c_1 n$ for some $c_1\in(0,1)$. Let $\epsilon_{i}=0$.
Then $F(\omega_0)=C_1^2n_1^2+C_2^2n_2^2+2C_1C_2n_1n_2\cos(\alpha_1-\alpha_2)+O(1)$. For any $\omega'$ such that $|\omega'-\omega_0|\asymp 1/(nh_n)$, where $h_n$ is diverging at an arbitrarily slow rate, elementary but tedious calculations yield 
$F(\omega')-F(\omega_0)=C_1C_2\sin(\alpha_1-\alpha_2)n_1n_2^2(\omega'-\omega_0)(1+o(1))$.
Hence, if $\omega'-\omega_0>0$, $F(\omega')>F(\omega_0)$ for sufficiently large $n$. Clearly, if $|\omega'-\omega_0|\gg 1/n$,  $F(\omega')<F(\omega_0)$. Therefore, we conclude that for this example, $F(W)$ is maximized at a point $\omega^*$ such that $1/(nh_n)\le|\omega^*-\omega_0|\le h_n/n$
for sufficiently large $n$. This example also shows that the rate $n^{-1}h_n$ in Theorem \ref{thm:s1_accuracy} cannot be improved (except for a factor of an arbitrarily slowly diverging function) for the DPPT if there exist changes in the phase of the oscillation.  
\end{example}
The following Proposition shows that if the phase of the oscillation does not change over time; that is, if only the amplitude of the oscillation is allowed to change over time, then the DPPT can detect the oscillatory frequencies at the $n^{-3/2}$ parametric rate except for a factor of logarithm. 
\begin{proposition}\label{prop:nophasechange}
Suppose the assumptions of Theorem \ref{thm:s1_accuracy} and \eqref{eq:thmcon} hold true. Further assume $\Omega\neq\emptyset$, and for each $\omega_k\in\Omega$ and all $r=0,1,\cdots,M_k$,  if $A_{r,k}^2+B_{r,k}^2\neq 0$, we have
\begin{eqnarray}\label{eq:nophasechange}
A_{r,k}/\sqrt{A_{r,k}^2+B_{r,k}^2}=c_k\ \mbox{ and }\ B_{r,k}/\sqrt{A_{r,k}^2+B_{r,k}^2}=d_k \mbox{ for some constants } c_k, d_k
\end{eqnarray}
independent of $r$. Then  
$  \pr \left( \max_{k} |\hat{w}_k - \omega_k| \lesssim  n^{-3/2}\log(n), |\hat{\Omega}| = |\Omega|  \right) \rightarrow 1 - \alpha. $
\end{proposition}
If we write $A_{r,k}+\sqrt{-1}B_{r,k}=C_{r,k}\exp(\sqrt{-1}\theta_{r,k})$, $C_{r,k}\ge 0$, $0\le \theta_{r,k}<2\pi$, then \eqref{eq:nophasechange} is equivalent to $\theta_{r,k}=\theta_{k}$ for some $\theta_k$; that is, there is no change in the phase of the oscillation. On the other hand, note that $C_{r,k}$, the magnitude of the oscillation, is allowed to change with respect to $r$. For the sleep EEG data analyzed in this paper, we are interested in detecting spindles which could be modeled as an oscillation at a fixed frequency that occurs for a short period of time and then vanishes. It can be easily seen that \eqref{eq:nophasechange} is suitable to model such short-term oscillations. 

\begin{remark}
The level $\alpha$ in Theorem \ref{thm:s1_accuracy} and Proposition \ref{prop:nophasechange} can be chosen as $\alpha=\alpha_n\rightarrow 0$ as long as $\alpha_n$ converges to 0 slower than the right hand side of \eqref{eq:boundthm1} and $1-P(A_n)$. In this case, it can be shown that the probabilities in Theorem \ref{thm:s1_accuracy} and Proposition \ref{prop:nophasechange} equal $1-\alpha_n(1+o(1))$ asymptotically.
\end{remark}

The following theorem investigates the asymptotic accuracy of the phase-adjusted local change point detection algorithm. Observe that the established $O_\pr(\log (\mt))$ rate of abrupt change point estimation is nearly the parametric $O_\pr(1)$ rate except a factor of logarithm.

\begin{theorem}[Accuracy of change point estimation]\label{thm:step2acc} %
Write $\Omega=\{\omega_1,\cdots,\omega_K\}$ and let $D_k:=\{b_{1,k},\cdots,b_{M_k,k}\}$ be the set of change points associated with $\omega_k$ and $\hat{D}_k$ be the set of all estimated change points by Algorithm 2 using $\hat{\omega}_k$, $k=1,2,\cdots, K$. For each $k=1,\cdots,K$, suppose that $\gamma_1<2/3$ if the oscillatory phase changes at frequency $\omega_k$. Further assume that all assumptions of Theorem \ref{thm:s2_consistancy} and Assumption \ref{assumption9} 
hold true. Then, for each $k$, $k=1,2,\cdots,K$, we have 
	
	1. If $D_k$ is empty, then $\pr( |\hat{D}_k| = 0) \rightarrow 1 - \beta$.
	
	2. If $|D_k| > 0 $, then $\pr \left( \max_r |\hat{b}_{r,k} - b_{r,k}| \leq \log(\tilde{m}) h_n, \,\,  |\hat{D}_k| = |D_k| \right) \rightarrow 1- \beta$
for any sequence $h_n>0$ that diverges to infinity arbitrarily slowly.
\end{theorem}

\section{Simulation study}\label{sec:simu}
We shall perform our simulation studies under four models, (M1)-(M4), for the non-stationary noise $\{\epsilon_{i,n}\}$ listed as follows. 
(M1) A locally stationary model with  $\epsilon_{k,n} := 0.5\cos(  k/n) \epsilon_{k-1,n} + e_k+0.3(k/n)e_{k-1},$ where $e_k$ are  i.i.d standard normal. 
    (M2) A piece-wise locally stationary model with $\epsilon_{k,n} := [0.5\cos(  k/n)\mathbbm{1}_{(k/n) < 0.75} + (k/n -0.5)\mathbbm{1}_{(k/n) \geq 0.75}]\epsilon_{k-1,n} +e_k$, where $e_k$ are  i.i.d standard normal. 
    (M3) A piece-wise locally stationary model with multiple breaks with 
    $\epsilon_{k,n} := [0.5\cos(  k/n)\mathbbm{1}_{(k/n) < 0.3} + (k/n -0.3)^2 \mathbbm{1}_{0.3 \leq (k/n) < 0.75}+ 0.3 \sin(  k/n)\mathbbm{1}_{(k/n)\ge 0.75}]\epsilon_{k-1,n} +e_k$, where $e_k$ are i.i.d standard normal. 
    (M4) A locally stationary model with heavy tails: $\epsilon_{k,n} :=0.5 \cos(  k/n)\epsilon_{k-1,n}+ e_k+0.3(k/n)e_{k-1} $, where $e_k$ are i.i.d $t(5)/\sqrt{5/3}$ where $t(5)$ denotes the $t$ distribution with 5 degrees of freedom. Note that $\sqrt{5/3}$ is the standard deviation of  $t(5)$. 
In the following simulations, the tuning parameters are selected according to the methods described in Section \ref{sec:tuning}. The simulations are performed with 1000 repetitions and for each repetition, the OBMB is performed using 1000 pseudo samples. 
\subsection{Test accuracy for Stage 1}
	The simulated data sets are generated from $X_{k,n} = \mu_{k,n} + \epsilon_{k,n}$, where $\mu_k = k/n $. Observe that $X_{i,n}$ has a smoothly time-varying mean and no oscillation. 
	The noises $\epsilon_{i,n} $ are generated according to (M1)-(M4) described above. 
	The  simulated rejection rates are reported in Table \ref{tab:s1 null}, which shows that the DPPT has reasonably accurate rejection rates under the null hypothesis of no oscillation for various kinds of non-stationary noises $\epsilon_{i,n}$.

\begin{table}[htb!]
\begin{center}
		\begin{tabular}{ |c| c| c| c| c |c|c| }
			\hline
			& \multicolumn{6}{|c|}{ Simulated Rejection Rates } \\\hline\hline
			 & \multicolumn{3}{|c|}{ $\alpha$ = 0.05 } & \multicolumn{3}{|c|}{ $\alpha$ = 0.10 }\\ \hline
			 Model &n=500& n = 1000 & n = 2000 &n=500&n = 1000 & n = 2000\\ \hline
			 M1 &0.061& 0.044 & 0.051 &0.121& 0.104 & 0.110\\ \hline
			 M2 &0.058& 0.057 &0.053 &0.126& 0.118 & 0.110\\ \hline
			 M3 &0.056& 0.046 & 0.050 &0.118& 0.108 & 0.110\\ \hline
			 M4 &0.060& 0.045 & 0.048 &0.120&0.098 & 0.113\\ \hline
		\end{tabular}
\end{center}
\caption{Simulated rejection rates under Stage 1 null conditions when $\mu_{k,n} = k/n $.\label{tab:s1 null}}
\end{table}

\subsection{Test accuracy for Stage 2}

	The second stage Type-I error simulation is performed under the null conditions a): $\mu_{k,n} = 2\sin(  \omega k)$ with $\omega = \pi/15$ and b): $\mu_{k,n} = 2.5\sin(  \omega_1 k)+2\sin(  \omega_2 k)$ with $\omega_1 =0.17(2\pi)$ and $\omega_2 = 0.3805(2\pi)$. The noises $\epsilon_{k,n}$ are generated from model (M1)-(M4).  Observe that the mean function $\mu_{k,n}$ does not have change points in its oscillatory behaviour. The simulated rejection rates  for a) and b) are recorded in Tables 2 and 3, respectively. Based on Tables 2 and 3, the simulated rejection rates are reasonably close to the nominal level $\beta$.
	
	\begin{table}[htb!]
	
\centering
		\begin{tabular}{ |c| c| c| c| c | c|c|}
			\hline
			& \multicolumn{6}{|c|}{ Simulated Rejection Rates}\\\hline\hline
			 & \multicolumn{3}{|c|}{ $\beta$ = 0.05 } & \multicolumn{3}{|c|}{ $\beta$ = 0.10 }\\ \hline
			 Model &$n=500$& $n=1000$ & $n=2000$ &$n=500$& $n=1000$ & $n=2000$  \\ \hline
			 M1 &0.058&0.047 & 0.053  &0.121& 0.102 & 0.116   \\ \hline
			 M2 &0.04&0.037& 0.046 &0.078 &0.083 & 0.105   \\ \hline
			 M3 &0.055& 0.048& 0.054 &0.112& 0.099 &  0.106  \\ \hline
			 M4 &0.068& 0.060 & 0.048 &0.135&  0.124  & 0.113  \\ \hline
		\end{tabular}
\caption{Simulated rejection rates for the proposed stage 2 algorithm when $\mu_{k,n} = 2\sin(\omega k)$ with $\omega = \pi/15$.\label{tab:s2 null}}
\end{table}

	\begin{table}[htb!]
	
	\centering
		\begin{tabular}{ |c| c| c| c| c |c|c|c|c| }
			\hline
			& \multicolumn{8}{|c|}{ Simulated Rejection Rates } \\\hline\hline
   & \multicolumn{4}{|c|}{ $n=500$ } &\multicolumn{4}{|c|}{ $n=1000$ }\\\hline
			 & \multicolumn{2}{|c|}{ $\beta$ = 0.05 } & \multicolumn{2}{|c|}{ $\beta$ = 0.10 }& \multicolumn{2}{|c|}{ $\beta$ = 0.05 } & \multicolumn{2}{|c|}{ $\beta$ = 0.10 }\\ \hline
			 Model & $\omega_1$ & $\omega_2$ & $\omega_1$ & $\omega_2$ & $\omega_1$ & $\omega_2$ & $\omega_1$ & $\omega_2$\\ \hline
			 M1 & 0.058 &0.040  &0.130  &  0.125& 0.048 &0.042  &0.116  &  0.110  \\ \hline
			 M2 & 0.031 &0.030 &   0.080 & 0.075 & 0.040 &0.041 &   0.101 & 0.085       \\ \hline
			 M3 & 0.058 &0.035 & 0.132 & 0.088& 0.048 &0.050 & 0.123 & 0.09   \\ \hline
			 M4 & 0.059 &0.0375 &  0.121  & 0.075& 0.048 &0.0475 &  0.102  & 0.085    \\ \hline
		\end{tabular}
		
\caption{Simulated rejection rates for the proposed stage 2 
algorithm when $\mu_{k,n} = 2.5\sin(  \omega_1 k)+2\sin(  \omega_2 k)$ with $\omega_1 =0.17(2\pi)$ and $\omega_2 = 0.3805(2\pi)$.\label{tab:s2 null 2}}
\end{table}

\subsection{Estimation accuracy for short-term oscillations}

We are concerned with short-term oscillations since spindles in sleep EEG signal studied in Section \ref{sec:data} could be modeled as short-term oscillations within a given frequency band.
In this subsection we would like to investigate the estimation accuracy of our methodology in this situation.
Specifically, we emulate short-term oscillations at two oscillatory frequencies and  different time locations, where $n=1000$, 
\begin{eqnarray}\label{eq:simu2}
\mu_k = 2 \cos(  \omega_1 k) \mathbb{I}_{0.2 n \leq k \leq 0.45n } + 2.5 \cos(  \omega_2 k) \mathbb{I}_{0.45n \leq k \leq 0.65n }+2 \cos(  \omega_1 k) \mathbb{I}_{0.65 n \leq k \leq 0.85n } 
\end{eqnarray} 
with $\omega_1 = 0.17007(2\pi)$ and $\omega_2 = 0.38007(2\pi)$.  The noises $\epsilon_{k,n}$ are generated from models (M1)-(M4).
Tables 4 and 5 report the accuracy of the estimators by computing their mean squared errors (MSE) and the probability of estimating the accurate number of oscillatory frequencies and change points, and they show a high  simulated accuracy. Furthermore, since there is no change in phase, Proposition \ref{prop:nophasechange} implies that the estimation precision $|\hat{w} - w| \approx n^{-3/2}\log(n) \approx 2\times10^{-4}$ which implies that $MSE(\hat{w}) \approx 4\times10^{-8}$. It can be seen that the results from Table 4 are consistent with this theoretical accuracy. Similarly, comparing the result from Table 5 with the theoretical accuracy of Theorem \ref{thm:step2acc}, we have $MSE(\hat{b}_i) \approx \log(n)^2 \approx 10^{1}$. The two results stay consistent. Finally, additional simulation results on the power performance as well as the estimation accuracy of our methodology can be found in Section A of the supplementary material.

\begin{table}[htp!]\label{tab:s1_spindle2}
\centering
	\begin{tabular}{ |c| c| c|c|c| c| c|c| }
		\hline
		&\multicolumn{3}{|c|}{ $\alpha$ = 0.05 }&\multicolumn{3}{|c|}{ $\alpha$ = 0.1 } \\
		\hline
		Model  &  $ MSE( \hat{w}_1)$ & $ MSE( \hat{w}_2)$ & $ \pr (|\hat{\Omega}|  = 2 ) $  &  $ MSE( \hat{w}_1)$ & $ MSE( \hat{w}_2)$ & $ \pr (|\hat{\Omega}|  = 2 ) $  \\ 
		\hline
		M1   &3.4e-08 & 6.7e-07 & 0.940 & 3.4e-08 & 6.7e-07 & 0.905 \\ 
		M2   &3.1e-08 & 6.5e-07 & 0.939 & 3.2e-08 & 6.5e-07 & 0.898 \\ 
		M3  & 3.1e-08 & 6.4e-07 & 0.944 & 3.1e-08 & 6.4e-07 & 0.892 \\ 
		M4  & 3.2e-08 & 6.8e-07 & 0.941 & 3.2e-08 & 6.8e-07 & 0.903 \\ 
		\hline
	\end{tabular}\label{tab:s2_spindle2}
	\caption{Simulated stage 1  estimation accuracy for the oscillatory frequencies specified in \eqref{eq:simu2}. } 
\end{table}

\section{Real data example: detecting sleep spindles}\label{sec:data}

    We demonstrate how the proposed two-stage algorithm can be applied to detecting spindles in the EEG signal recorded during sleep. Spindles are bursts of neural oscillatory activity during sleep that are captured by the EEG. They are generated by the complicated interplay of the thalamic reticular nucleus and other thalamic nuclei \citep{DEGENNARO2003423} during the N2 sleep stage, which is defined based on the AASM sleep stage classification system \cite{Iber_Ancoli-Isreal_Chesson_Quan:2007}. Spindles oscillate in a frequency range of about 11 to 16 Hz with a duration of 0.5 seconds or greater (usually 0.5-1.5 seconds). The EEG signal during the N2 stage serves a good example for the change point detection problem. The spindle might exist from time to time, and there might be multiple spindles during the N2 stage. The dynamics of spindles encode important physiological information \citep{DEGENNARO2003423}. While it is possible to have experts labeling it, it might not be feasible if the data size is large. We thus need an automatic labeling algorithm. To apply our proposed two-stage algorithm we check if the model \eqref{eq:mean} is reasonable. Although the spindle frequency might change from time to time, the frequency changes slowly and the spindles exist only for a relatively short period, so we can reasonably assume that the frequency $\omega_k$ in \eqref{eq:mean} is fixed. On the other hand, the appearance of spindle can be well captured by the amplitude $A_{r,k}$ and $B_{r,k}$ in \eqref{eq:mean}. Moreover, the EEG signal other than the spindle is non-stationary, and we assume that it can be well captured by the PLS model. We emphasize that how well the PLS model captures this non-stationarity is out of the scope of this paper.

\begin{table}[htb!]
\centering
	\begin{tabular}{|c | c| c| c|c|c|c| }
		\hline
		\multicolumn{7}{|c|}{ $\omega_1$ }\\
		\hline
		&\multicolumn{3}{|c|}{ $\beta$ = 0.05 }&\multicolumn{3}{|c|}{ $\beta$ = 0.1 } \\
		\hline
		Model & $ MSE( \hat b_{1,1})$ & $ MSE(\hat b_{1,2}) $ & $ \pr( |\hat{B}| = 4 )$  & $ MSE( \hat b_{1,1})$ & $ MSE(\hat b_{1,2}) $ & $ \pr( |\hat{B}| = 4 )$ \\ 
		\hline
		M1  & 20.24 & 12.93 & 0.928 & 19.96 & 12.76 & 0.870 \\ 
		M2  & 18.62 & 17.68 & 0.931 & 19.18 & 17.31 & 0.872 \\ 
		M3  & 26.34 & 19.85 & 0.923 & 26.63 & 19.90 & 0.862 \\ 
		M4  & 21.11 & 13.06 & 0.930 & 21.21 & 13.14 & 0.872\\
		\hline
		Model & $ MSE( \hat b_{1,3})$ & $ MSE(\hat b_{1,4}) $ &   & $ MSE( \hat b_{1,3})$ & $  MSE(\hat b_{1,4}) $ &  \\ 
		\hline
		M1  & 13.12 & 19.75 & & 13.54 & 19.88 & \\ 
		M2  & 19.52 & 20.88 & & 19.43 & 20.05 &  \\ 
		M3  & 18.25 & 25.41 & & 18.30 & 25.01 & \\ 
		M4  & 13.55 & 20.21 & & 13.77 & 20.30 & \\
		\hline
		\hline
		\multicolumn{7}{|c|}{  $\omega_2$}\\
		\hline
		&\multicolumn{3}{|c|}{ $\beta$ = 0.05 }&\multicolumn{3}{|c|}{ $\beta$ = 0.1 } \\
		\hline
		Model & $ MSE( \hat b_{2,1})$ & $ MSE(\hat b_{2,2}) $ & $ \pr( |\hat{B}| = 2 )$  & $ MSE( \hat b_{2,1})$ & $ MSE(\hat b_{2,2}) $ & $ \pr( |\hat{B}| = 2 )$ \\ 
		\hline
		M1  & 9.2 & 8.61 & 0.935 & 9.96 & 8.73 & 0.880 \\ 
		M2  & 7.88 & 6.84 & 0.945 & 7.88 & 6.89 & 0.887 \\ 
		M3  & 8.16 & 9.27 & 0.933 & 8.19 & 9.39 & 0.876 \\ 
		M4  & 9.88 & 8.64 & 0.937 & 9.76 & 8.61 & 0.883\\
		\hline
	\end{tabular}
	\caption{Simulated stage 2  estimation accuracy for $\mu_{k,n}$ specified in \eqref{eq:simu2}, where $b_{1,1}=0.2n$, $b_{1,2}=0.45n$, $b_{1,3}=0.65n$, $b_{1,4}=0.85n$, $b_{2,1}=0.45n$, and $b_{2,2}=0.65n$.}
\end{table}

	In this section, EEG was recorded from the standard polysomnogram signals on patients suspicious of sleep apnea syndrome at the sleep center in Chang Gung Memorial Hospital (CGMH), Linkou, Taoyuan, Taiwan, under the approval of the Institutional Review Board of CGMH (No. 101-4968A3). All recordings were acquired on the Alice 5 data acquisition system (Philips Respironics, Murrysville, PA). EEG is sampled at 200Hz.
    The sleep stages, including wake, rapid eyeball movement (REM) and N1, N2 and N3 of NREM, were annotated by two experienced sleep specialists according to the AASM 2007 guidelines \citep{Iber_Ancoli-Isreal_Chesson_Quan:2007} with consensus and the sleep specialists provide annotation for 30-seconds long epochs. Below, we focus on those epochs labeled as the N2 stage.

\begin{figure}[bht!]
\centering
	\includegraphics[trim=0 10 20 50,clip,width=.75\textwidth]{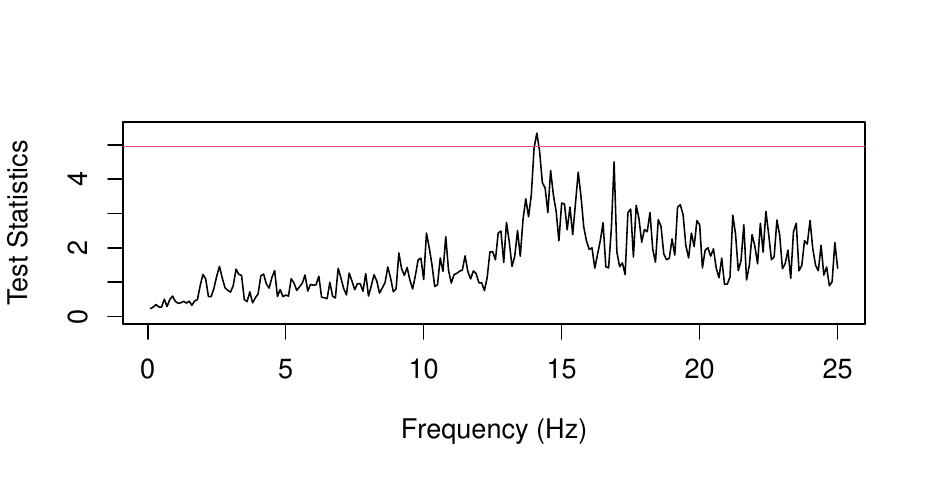}
	\includegraphics[trim=0 10 20 50,clip,width=.75\textwidth]{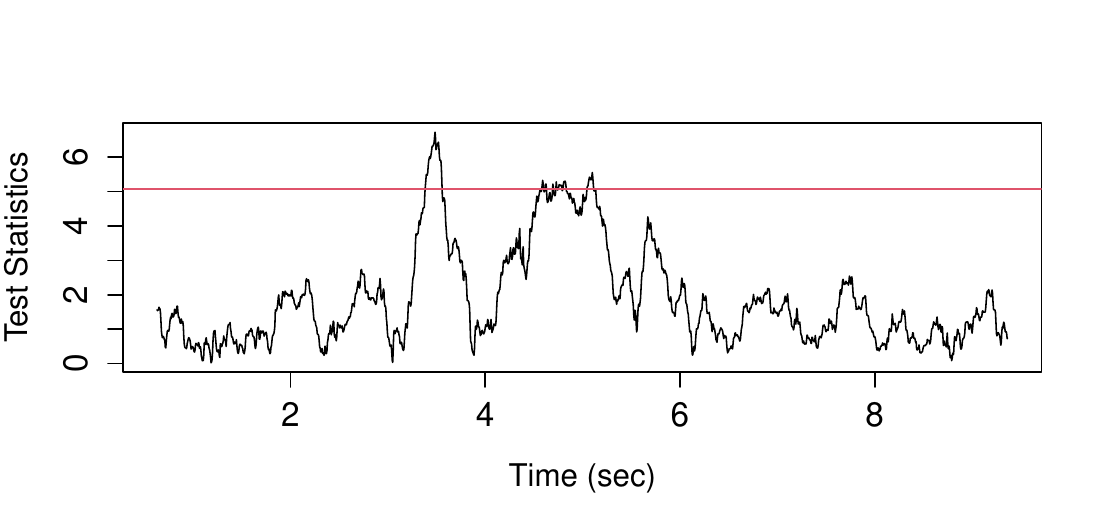}
	\includegraphics[trim=0 10 20 50,clip,width=.75\textwidth]{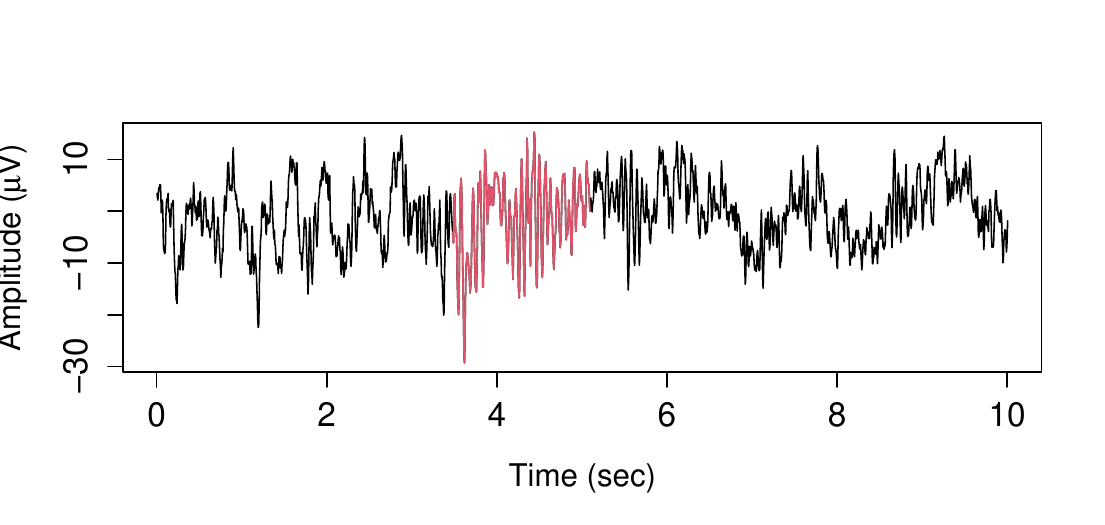}
	\caption{{First panel: the DPPT statistics $\bar{F}(\cdot)$. Second panel: the second-stage statistics $T(\cdot,14.1)$. Third panel: the EEG signal under analysis with the detected spindle superimposed in red}.\label{Figure:EEG example1}} 
	\end{figure}
	
    The data analysis results for a 10-sec segment of the EEG signal from channel C4A1 are shown in Figure \ref{Figure:EEG example1}.
	In the two-stage algorithm, we first want to find and estimate the frequencies of existing oscillations in the signal.
	The first-stage, the DPPT statistic with $m=14$ selected by our plug-in method is shown in the first panel against a grid of test frequencies in Hz, where the horizontal line show 90th simulated quantile value of the first iteration step by using the OBMB. Clearly, there is a peak around 14Hz, which suggests that there exists an oscillatory component inside the EEG signal. Note that the size of the peaks is affected by the amplitude of the oscillatory pattern and sample size of the time series data. Under significant level $\alpha = 0.1$, we are confident that there is one oscillatory component with the estimated frequency 14.1Hz. Furthermore, the DPPT is not significant after removing the first oscillatory frequency, indicating that the signal is oscillating at only one frequency.
	For the detected oscillatory component, next we estimate if there is any change point. The second panel shows the second-stage statistics with $\tilde{m}=112$ and $m'=13$ plotted against time position. Both tuning parameters stated above are selected by methods proposed in Section \ref{sec:tuning}. The horizontal line indicates the 90-th percent simulated quantile of the first iteration step under estimated frequency 14.1Hz. The estimated break points are at 696-th, 905-th and 1019-th units under $\beta = 0.1$. Again, the size and shape of each peak in the second-stage statistics are affected by the choice of the bandwidth parameter $\mt$ and amplitude of the oscillatory pattern existing in the data. 
	In the third panel, the detected spindle is colored in red, which coincides with the experts' annotation. 
	\begin{figure}[bht!]
	\centering
	\includegraphics[trim=0 10 20 50,clip,width=.75\textwidth]{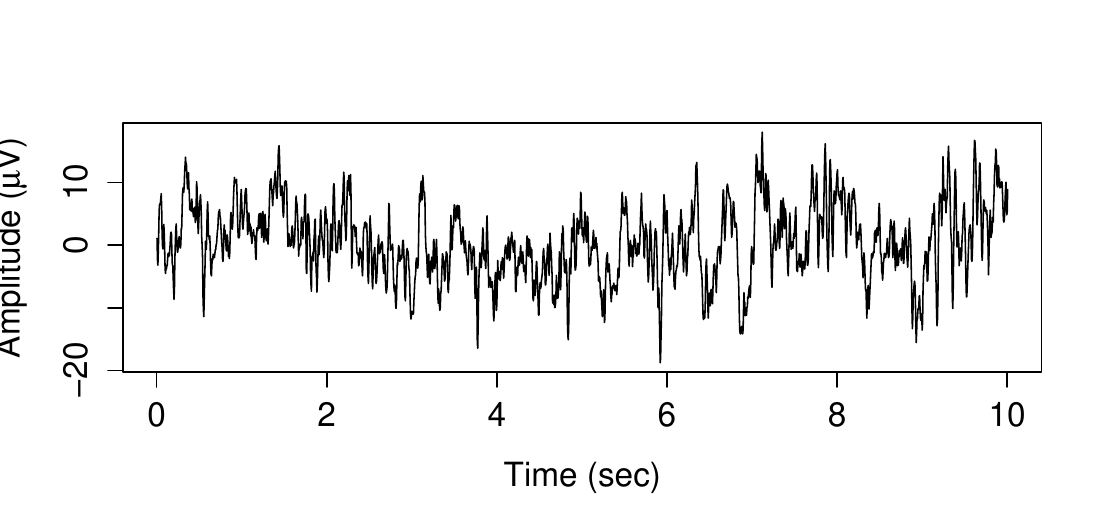}
	\includegraphics[trim=0 10 20 50,clip,width=.75\textwidth]{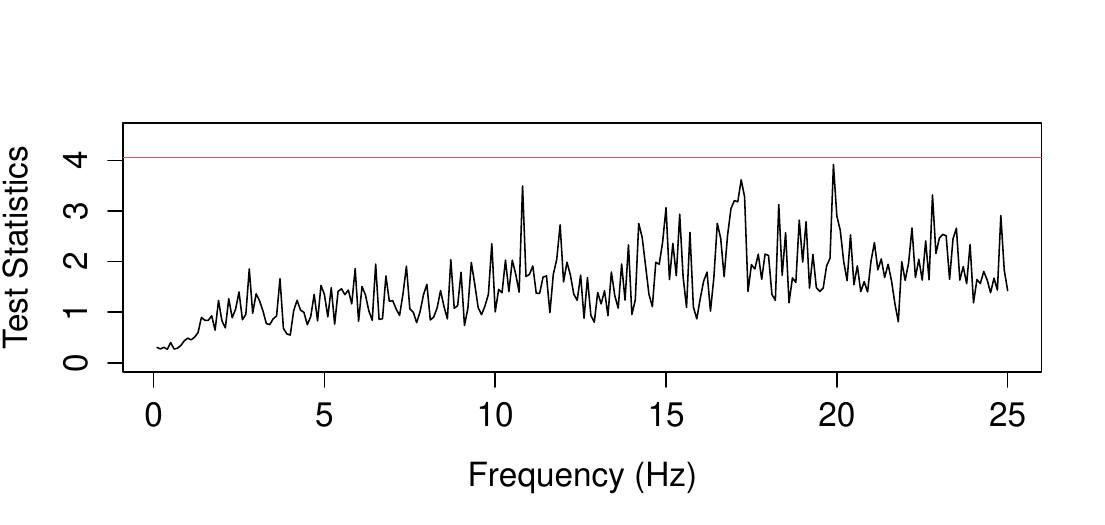}
	\caption{First panel: The recorded EEG signal. Second panel: the DPPT statistics $\bar{F}(\cdot)$.\label{Figure:EEG example2}} 
	\end{figure}
	
	 Another segment of the EEG signal from channel C4A1 without any annotated spindle is shown in the first panel in Figure \ref{Figure:EEG example2}. 
	The first-stage, the DPPT statistic, is shown in the second panel against a grid of test frequencies in Hz, where the horizontal line show 90th percent simulated quantile value by using the OBMB. While there seem to be peaks around 20Hz, 26Hz and 32Hz, they are not significant under significant level $\alpha = 0.1$. This finding suggests that there does not exist an oscillatory component inside the EEG signal that is sufficiently strong or long.

	The above preliminary analysis results indicate the potential of the proposed two-stage algorithm. A systematic application of the proposed algorithm to physiological signals for clinical applications will be reported in our future work.

\bibliography{reference}

	\setcounter{equation}{0}
	\setcounter{proposition}{0}
	\setcounter{corollary}{0}
	\setcounter{lemma}{0}
	\setcounter{theorem}{0}
	\renewcommand{\thepage}{S.\arabic{page}}
	\renewcommand{\thesection}{S.\arabic{section}}
	\renewcommand{\theequation}{S.\arabic{equation}}
	\renewcommand{\thelemma}{S.\arabic{lemma}}
	\renewcommand{\thetable}{S.\arabic{table}}

\clearpage
\appendix

	\renewcommand{\thepage}{S.\arabic{page}}
	\renewcommand{\thecorollary}{\Alph{section}.\arabic{corollary}}
	\renewcommand{\thelemma}{\Alph{section}.\arabic{lemma}}
	\renewcommand{\thetheorem}{\Alph{section}.\arabic{theorem}}
	\renewcommand{\theproposition}{\Alph{section}.\arabic{proposition}}

\title {\large\bf Supplementary}
	
\section{Literature review, additional simulation results and figures}\label{sec:asimu}

First of all, we present two figures. The first one shown in Figure \ref{fig:energy_leak} is an illustration of the energy leak phenomenon, and the second one shown in Figure \ref{fig:2} is about $\bar{F}(\omega)$. See the main article referring to these figures for detail.

\begin{figure}[htbp]
\centering
\includegraphics[width=.6\textwidth]{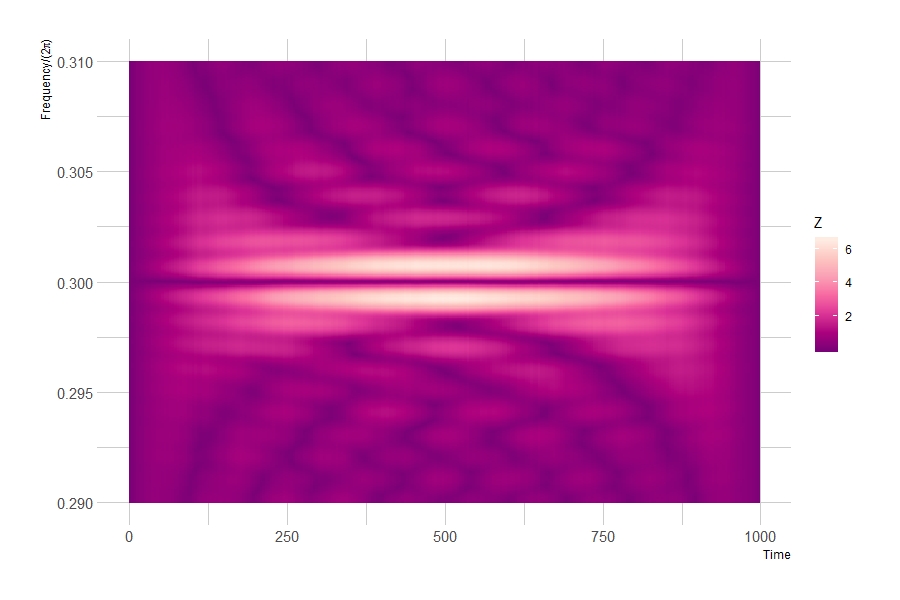}
\caption{Heat map of CUSUM statistics $|C_n(i,\omega)|$ with $n=1000$, mean function $\mu_{i,n} = \cos(0.6\pi i)$ and $\epsilon_{i,n}$ i.i.d. $N(0,1)$.  Plot shows energy leak around frequency $\omega = 0.6\pi$.\label{fig:energy_leak}}
\end{figure}

\begin{figure}[htb!]
\center
\includegraphics[scale = 0.5]{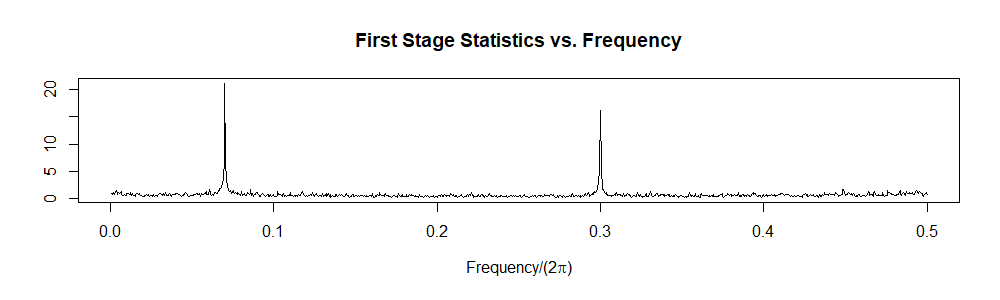}
\caption{Example of $\bar{F}(\omega)$ with $n=2000$, $\mu_{i,n} = 2 \cos(2\pi i 0.07) + 1.5 \cos(2\pi i 0.3)$ and $\epsilon_{i,n} = 0.5 \cos(i/n)\epsilon_{i-1,n}+e_i$, $e_i$ i.i.d. standard normal.}\label{fig:2} 
\end{figure}

\subsection{Literature review}\label{sec:lr}
The statistics literature of unknown periodicity detection dates back at least to \cite{fisher1929}, who considered testing the existence of a sinusoidal signal under i.i.d. Gaussian noise. Fisher's test was based on the maximum of the periodogram over the canonical frequencies and was later generalized to accommodate stationary and dependent noises and multiple oscillatory frequencies; see for instance  \cite{hannan1961}, \cite{chiu1989} and \cite{lin2009}. See also \cite{paraschakis2012} for the estimation of time-varying oscillatory frequencies and phases. To our knowledge, none of the previous spectral domain literature on periodicity detection considered non-stationary noises due to the difficulties mentioned in the main article. On the other hand, there are a few statistics papers, see for instance \cite{oh2004} and \cite{genton2007}, considering time domain estimation of unknown periodicity with application to light curve estimation of variable stars. See also the related astronomy literature cited therein. These papers typically assume that there is an oscillatory signal and aim at estimating the period and the corresponding periodic function. Another recent contribution is \cite{dahlhaus2017} where the authors consider estimating a generalized state-space model with an unknown periodic pattern function and a hidden stochastic and integrated phase process. It is worth mentioning that classic time series analysis of seasonal processes typically assumes that the periodicity is known and then removes the periodic trend by differencing the process at an appropriate order \citep{box2015}.

To our knowledge, there exists no statistics literature on oscillation change point detection except for a small number of related contributions in frequency domain change point detection for the spectral density function; see for instance \cite{last2008detecting}, \cite{lavielle2000multiple} and \cite{preuss2015detection} among others. We point out that the focus of those papers are essentially different from the current paper. The main difference is that these papers are all about detecting changes in the spectral density or power spectral associated with the covariance structure of the time series, whereas the current paper deals with changes in the oscillation in mean. In particular, their algorithms aim to detect change points in the covariance structure of our error process $\epsilon_{i,n}$, while we aim to detect the change point of the deterministic oscillation. 

The huge recent statistics literature on change point detection typically focuses on problems in the time domain and hence are free from the energy leak phenomenon in the spectral domain. In particular, we point out that local change point detection algorithms are common for time domain change point detection; see for instance \cite{preuss2015detection} and \cite{yau2016inference}. However, those local algorithms are not motivated by the energy leak problem in the frequency domain. For reviews of recent advances in change point detection, we refer the readers to \cite{aue2013structural} and \cite{niu2016review}. 

In the past decades, time-frequency (TF) analysis has been widely studied in the applied mathematics and application fields due to its flexibility to handle complicated and nonstationary time series \citep{wu2020current}. TF tools can be roughly classified into three types -- linear, bilinear and nonlinear \citep{flandrin1998time}. Since the linear-type TF analysis are directly related to this work, we focus on it. The basic idea is dividing the signal into segments and evaluating the spectrum for each segment, where how the signal is partitioned distinguishes different methods. When a fixed window is chosen and the Fourier transform is applied, it is the short-time Fourier transform (STFT) \citep{flandrin1998time}; when the segments depend on a dilated mother wavelet, it is the continuous wavelet transform (CWT) \citep{daubechies1992}. Due to the fundamental difference of their associated group structure (Heisenberg group for STFT and the affine group for the CWT), their fundamental differences manifest in various aspects. What concerns us in this work is their capability to study different functional spaces \citep{daubechies1992}. Particularly, the CWT can be applied to characterize local regularity, and hence has been widely applied to study singularities in the signal processing \citep{jaffard1996wavelet}. Note that the local change point detection statistic can be viewed as detecting discontinuity by the Haar wavelet. %

\subsection{Numerical experiments on spectral dependency}\label{sec:sd}
As we pointed out in Introduction, a classic result in frequency-domain time series analysis is that the periodogram at the canonical frequencies are approximately independent for a wide class of weakly stationary processes. In this subsection, we shall illustrate numerically that the latter result is no longer valid for PLS time series. To this end, we simulated model (M1) in Section 5 of the main article 1000 times with $n=1000$. For comparison, we also simulated 1000 times a stationary ARMA(1,1) model $X_i-0.6X_{i-1}=\epsilon_i+0.3\epsilon_{i-1}$ with i.i.d. standard normal $\epsilon_i$ and time series length $n=1000$.  The sample correlations between the periodograms at frequency lags 1 and 2 are depicted in Figure \ref{fig:sd} below. While the periodogram looks uncorrelated at lags 1 and 2 for the stationary ARMA(1,1) model as indicated by the classic theory, it is clear that the correlations between $I_n(\omega_j^*)$ and $I_n(\omega_{j+1}^*)$ are non-zero and varying with respect to $j$ for our non-stationary model (M1).  As we indicated in Section \ref{sec:sd_cp} of the main article, such spectral dependency makes it difficult theoretically to investigate the maximum deviation of the periodogram on a dense set of frequencies for non-stationary time series.  
1

	\begin{figure}[bht!]
	\includegraphics[width=\textwidth]{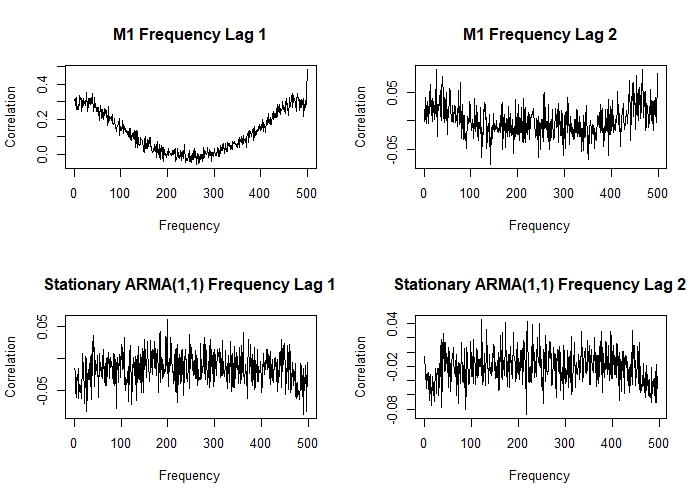}
	\caption{Sample correlations $\mbox{Corr}(I_n(\omega_j^*), I_n(\omega_{j+1}^*))$ and  $\mbox{Corr}(I_n(\omega_j^*), I_n(\omega_{j+2}^*))$ plotted against $j$ for (M1) and a stationary ARMA(1,1) model, $j=1,2,\cdots, 500$.\label{fig:sd}} 
	\end{figure}

\subsection{More discussions on energy leak and its influences on a global frequency change point algorithm}\label{sec:el}

The spectral energy leak problem persists if we estimate the oscillatory frequencies first and then apply change point detection algorithms directly to the Fourier transforms at the estimated frequencies. It is well-known that the parametric rate for oscillatory frequency estimation is  $O_\pr(n^{-3/2})$ (cf. \cite{genton2007}). Even at this very fast convergence rate, it can be shown that $|L_n(i,\omega)-L_n(i,\hat\omega)|=O_\pr(\sqrt{n})$ for sufficiently large $i$, where $\omega$ is a true oscillatory frequency and $\hat\omega$ is its estimate. The estimation error is not negligible asymptotically and change point detection algorithms will behave erratically if we simply plug-in the estimated oscillatory frequencies. See Lemma D.10 and Remark 2 in this supplementary material for more detailed calculations and discussions. 

In this section, we shall conduct a numerical experiment on how the spectral energy leak phenomenon distorts the accuracy of the CUSUM frequency change point test, a classic global change point detection algorithm (see Section \ref{sec:sel} for a description). In particular, we shall show that a seemingly slight estimation error in the oscillating frequency will have a devastating impact on the CUSUM test. To this end, we generate $ X_{k,n} = \mu_{k,n} + \epsilon_{k,n}$, $n=1000$, where $\epsilon_{k,n}$ are given by models (M1)-(M4) in Section \ref{sec:simu} of the main article and the mean function $\mu_{k,n} = 2\sin(\omega k)$ with $\omega = 2\pi^2/20 \in [0,2\pi]$. Note that there is no frequency change in $X_{k,n}$. We shall first estimate the oscillating frequency using our stage 1 algorithm. Then we shall use the OBMB technique developed in \cite{zhou2013} to calculate the critical values of the CUSUM test. It can be shown that the OBMB is asymptotically consistent if the estimated oscillating frequency exactly equals to $\omega = 2\pi^2/20$. With $1000$ replicates, the simulated rejection rates of the CUSUM test as well as the MSE of our stage 1 estimation are reported in Table \ref{tab:CUSuM null} below.

\begin{table}[htb!]
    \centering
	\begin{tabular}{ |c| c| c|c|  }
		\hline
		& \multicolumn{2}{|c|}{ Simulated Rejection Rates} &  Frequency Estimation \\
		\hline\hline
		& $\alpha$ = 0.05 &  $\alpha$ = 0.10& MSE \\ \hline
		M1 & 0.738 & 0.828 & 1.65E-06 \\ \hline
		M2 & 0.744 & 0.832 & 1.51E-06 \\ \hline
		M3 & 0.74  & 0.837 & 1.58E-06\\ \hline
		M4 & 0.739 & 0.81   & 1.53E-06 \\ \hline
	\end{tabular}
	\hspace{.1cm}
        \caption{Simulated rejection rates of the CUSUM test with estimated oscillating frequency. The table demonstrates distortion of the test accuracy due to spectral energy leak.\label{tab:CUSuM null}}
\end{table} 
We see from Table  \ref{tab:CUSuM null} that the MSE of estimation is consistent with the our theoretical $O_{\pr}(n^{-3/2}\log n)$ estimation error rate. On the other hand, the simulated rejection rates are all very large. As we discussed in the introduction of the main article, the loss of accuracy for the CUSUM test is due to spectral energy leak that separates the periodograms at the estimated frequency and the true frequency even if those two frequencies are within an $O(n^{-3/2}\log n)$ range from each other.

\subsection{Test power for Stage 1}\label{sec:ps1}
The first stage power simulation is performed under the setting where $\mu_{k,n} =  A \cos(  \omega k )$ for $\omega= 0.1(2\pi)$ and $\epsilon_{k,n} = 0.5\cos( k/n) \epsilon_{k-1,n} + e_k$ with $e_k$ i.i.d. standard normal. The significance level is set at $\alpha = 0.05$ and $0.1$ and $n$ = 1000. Figure \ref{fig:Power Plot 1} shows the simulated rejection rates plotted against the amplitude $A$. The variance of the generated time series data is approximately one and the signal strength is quantified by the oscillatory amplitude $A$. Thus $A$ is approximately the signal to noise ratio and Figure \ref{fig:Power Plot 1} show that the simulated rejection rates increase quite fast and approach 1 when $A$ is larger than 0.35. 

\begin{figure}[htb!]
\centering
\includegraphics[width=12cm,height=7cm,trim=0 0 0 40,clip]{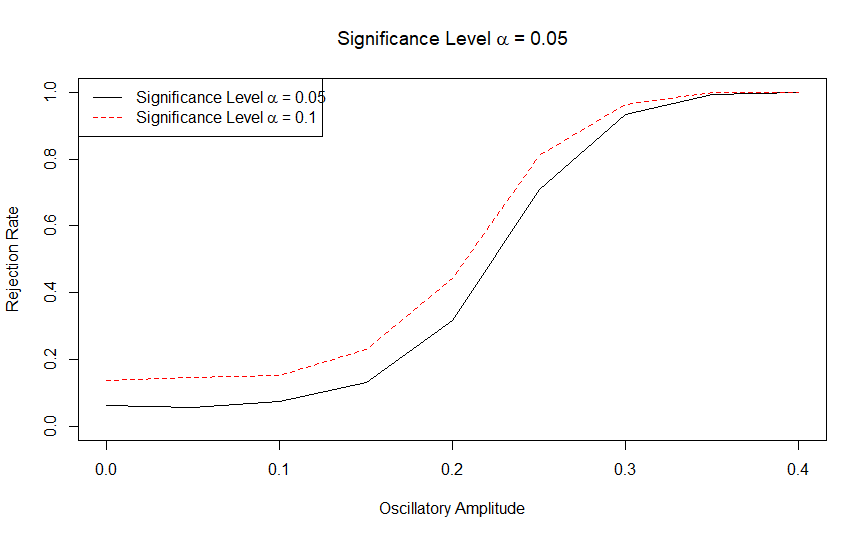}
\vspace{-0.5cm}
\caption{Simulated rejection rates for $\mu_{k,n} =  A \cos(  \omega k)$ for various values of $A$. }\label{fig:Power Plot 1}
\end{figure}

\subsection{Sensitivity analysis on tuning parameter selection}\label{sec:sa}
In this section, we shall perform a sensitivity analysis of the proposed plug-in block size selection method for the bootstrap as well as the penalization method for the local neighborhood selection.  Denote by $m_s$, $m'_s$ and $\mt_s$ as the Stage 1 bootstrap block size, the Stage 2 bootstrap block size and the Stage 2 local neighborhood size selected by our proposed tuning parameter selection methods, respectively. We shall investigate the robustness of our oscillation frequency and change point detection algorithms when the tuning parameters deviate moderately from $m_s$, $m'_s$ and $\mt_s$. Specifically, we investigate the Stage 1 and Stage 2 Type I error rates for Models M1-M4 stated in Section \ref{sec:simu} of the main article when the selected tuning parameters are enlarged or shrank by a factor of 1.2 or 1.4. The results based on $n=500$ with 1000 replications are shown in Tables \ref{tab:sens_s1_null}, \ref{tab:sens_s2_null_1} and \ref{tab:sens_s2_null_2}below.

We can see from the simulation results that our oscillation frequency and change point detection algorithms are reasonably stable when the tuning parameters deviate moderately from those selected by the proposed data-driven methods. In particular, the Type I error rates are reasonably robust to the bootstrap block sizes $m$ and $m'$ selected by our plug-in method. The local neighborhood size parameter $\mt$ is in general robust to moderate deviations from $\mt_s$. But we do observe an increase in Type I error when $\mt$ is increased to $1.4\mt_s$ in which case the energy leak problem occasionally come into effect. Further simulation studies not reported here show that the Stage 2 Type I error rates become reasonably accurate when $\mt=1.4\mt_s$ and the sample size increases to 1000. Additionally, the estimation accuracy of Stage 1 and 2 algorithms have similar sensitivity to the selected tuning parameters as discussed above.

\begin{table}[htb!]
\begin{center}
		\begin{tabular}{ |c| c| c| c| c |c|c|c|c| }
			\hline
			& \multicolumn{8}{|c|}{ Simulated Rejection Rates } \\\hline\hline
			 & \multicolumn{4}{|c|}{ $\alpha$ = 0.05 } & \multicolumn{4}{|c|}{ $\alpha$ = 0.10 }\\ \hline
			  &$1.2m_s$& $1.4m_s$& $m_s/1.2$ &$m_s/1.4$&$1.2m_s$ & $1.4m_s$&$m_s/1.2$&$m_s/1.4$\\ \hline
			 M1 &0.060& 0.062 & 0.058 & 0.061& 0.121& 0.130 & 0.120&0.126\\ \hline
			 M2 &0.057& 0.065 &0.053 &0.061&0.121& 0.130 & 0.120&0.125\\ \hline
			 M3 &0.055& 0.066 & 0.051&0.060 &0.118& 0.125 & 0.111&0.122\\ \hline
			 M4 &0.060& 0.065 & 0.061 &0.066&0.120&0.133 & 0.118&0.124\\ \hline
		\end{tabular}
\end{center}
\caption{\small Simulated rejection rates for various choices of the bootstrap block size parameter $m$ under Stage 1 null conditions when $\mu_{k,n} = 0$ and $n=500$. Here $m_s$ is the Stage 1 bootstrap block size selected by our proposed tuning parameter selection method. \label{tab:sens_s1_null}}
\end{table}

\begin{table}[htb!]
\begin{center}
		\begin{tabular}{ |c| c| c| c| c |c|c|c|c| }
			\hline
			& \multicolumn{8}{|c|}{ Simulated Rejection Rates } \\\hline\hline
			 & \multicolumn{4}{|c|}{ $\beta$ = 0.05 } & \multicolumn{4}{|c|}{ $\beta$ = 0.10 }\\ \hline
			  &$1.2m'_s$& $1.4m'_s$& $m'_s/1.2$ &$m'_s/1.4$&$1.2m'_s$ & $1.4m'_s$&$m'_s/1.2$&$m'_s/1.4$\\ \hline
			 M1 &0.065& 0.065 & 0.061 & 0.069& 0.123& 0.129 & 0.121&0.128\\ \hline
			 M2 &0.053& 0.060 &0.055 &0.061&0.109& 0.120 & 0.111&0.121\\ \hline
			 M3 &0.055& 0.062 & 0.053&0.064 &0.107& 0.116 & 0.105&0.120\\ \hline
			 M4 &0.061& 0.069 & 0.059 &0.068&0.121&0.131 & 0.120&0.128\\ \hline
		\end{tabular}
\end{center}
\caption{\small Simulated rejection rates for various choices of the bootstrap block size parameter $m'$ under Stage 2 null conditions when $\mu_{k,n} = 2\sin(\omega k)$ with $\omega = \pi/15$ and $n=500$. Here the local neighborhood size $\tilde m$ is fixed at $\mt_s$ selected by the penalization method. $m'_s$ is the Stage 2 bootstrap block size selected by our proposed tuning parameter selection method. \label{tab:sens_s2_null_1}}
\end{table}

\begin{table}[htb!]
\begin{center}
		\begin{tabular}{ |c| c| c| c| c |c|c|c|c| }
			\hline
			& \multicolumn{8}{|c|}{ Simulated Rejection Rates } \\\hline\hline
			 & \multicolumn{4}{|c|}{ $\beta$ = 0.05 } & \multicolumn{4}{|c|}{ $\beta$ = 0.10 }\\ \hline
			  &$1.2\mt_s$& $1.4\mt_s$& $\mt_s/1.2$ &$\mt_s/1.4$&$1.2\mt_s$ & $1.4\mt_s$&$\mt_s/1.2$&$\mt_s/1.4$\\ \hline
			 M1 &0.063& 0.085 & 0.061 & 0.07& 0.126& 0.139 & 0.125&0.131\\ \hline
			 M2 &0.055& 0.080 &0.055 &0.063&0.107& 0.140 & 0.110&0.130\\ \hline
			 M3 &0.057& 0.082 & 0.055&0.067 &0.109& 0.142 & 0.111&0.125\\ \hline
			 M4 &0.061& 0.086 & 0.060 &0.070&0.122&0.139 & 0.119&0.131\\ \hline
		\end{tabular}
\end{center}
\caption{\small Simulated rejection rates for various choices of the local neighborhood size parameter $\mt$ under Stage 2 null conditions when $\mu_{k,n} = 2\sin(\omega k)$ with $\omega = \pi/15$ and $n=500$. Here the bootstrap block size $m'$ is fixed at $m'_s$ selected by the plug-in method. $\mt_s$ is the Stage 2 local neighborhood size selected by our proposed tuning parameter selection method. \label{tab:sens_s2_null_2}}
\end{table}

\subsection{Numerical analysis comparing classic frequency detection algorithms}\label{sec:class_freq}
In this section we shall perform a numerical experiment on the performance of two of the classic periodogram-based oscillation frequency detection algorithms under non-stationarity and complex oscillation. The tests investigated are Fisher's maximum periodogram test and the test proposed by Liu and Lin (2009). For a stationary time series $\{X_i\}_{i=1}^n$, Fisher's maximum periodogram test is defined as  
$$F_n=\max_{1\le i\le q}I_n(\omega_i)/(2\pi\hat f(\omega_i)),$$
where $I_n(\omega_i)$ is the periodogram of $\{X_k\}$ at frequency $\omega_i$, $\hat f(\omega_i)$ is the estimated spectral density of $\{X_k\}$ at frequency $\omega_i$, $\omega_i=2\pi i/n$, $i=1,2,\cdots, q$ are the canonical frequencies and $q=\lfloor (n-1)/2\rfloor$. Liu and Lin (2009) proposed using the following test statistic
$$g_n=\frac{\max_{1\le i\le q}I_n(\omega_i)/\hat f(\omega_i)}{q^{-1}\sum_{i=1}^qI_n(\omega_i)/\hat f(\omega_i)}.$$
For stationary time series and under some extra mild conditions, both $F_n-\log q$ and $g_n-\log q$ converge to the standard double exponential distribution (Liu and Lin (2009)). In this section, first we shall investigate the behaviour of the above-mentioned two tests under non-stationarity and the null hypothesis of no oscillation. To this end, we let $X_{i,n}=\epsilon_{i,n}$, where $\epsilon_{i,n}$ is generated from Models M1-M4 stated in Section \ref{sec:simu} of the main article as well as a non-stationary white noise model M5: $\epsilon_{i,n}=(1/4+i/n)e_i$ where $e_i$ are i.i.d. standard normal. We also consider a benchmark stationary Gaussian white noise M6. The simulated Type I error rates with 1000 replicates are summarized in Table \ref{tab:sens_class_freq} below.

The simulation results illustrate that the classic tests perform well under stationary Gaussian white noise. However, both tests are very conservative under all non-stationary noise models considered. It seems that this is due to both spectral dependency and the fact that $f(\omega_i)$ are no longer the relevant quantity to quantify the variability of $I_n(\omega_i)$ for non-stationary time series. Hence $\hat f(\omega_i)$ may not be the appropriate normalizing constants for those classic tests under non-stationary noises. On the other hand, our simulation results in Section \ref{sec:simu} of the main article demonstrate that the dense progress periodogram test performs well in terms of its Type I error rates under piece-wise locally stationary errors. 

One direct consequence of the conservativeness of the classic tests is that they are much less powerful than the DPPT to detect oscillations. Our further simulations confirm this. For instance, recall that in Section \ref{sec:ps1} we illustrated that the DPPT achieves power 1 when $A$ is greater than 0.35 at nominal level 0.05. However, under the same non-stationary data generating mechanism, both $F_n$ and $g_n$ require $A$ to be larger than 1 in order to achieve high power. In terms of estimation accuracy, both $F_n$ and $g_n$'s estimation errors are of the order $n^{-1}$ as they both only consider the canonical frequencies. On the other hand, as we discussed in the main article, the DPPT is able to achieve an estimation error of the order $n^{-3/2}\log(n)$ due to its dense nature. In summary, our simulations illustrate that the DPPT seems to be a more accurate and powerful tool for oscillation frequency detection and estimation under non-stationarity.   

\begin{table}[htb!]
\begin{center}
		\begin{tabular}{ |c| c| c| c| c |c|c| }
			\hline
			& \multicolumn{6}{|c|}{ Simulated Rejection Rates } \\\hline\hline
   & \multicolumn{6}{|c|}{ $\alpha=0.05$ } \\\hline
			  &M1& M2& M3 &M4&M5 &M6\\ \hline
			 $F_n$ &0.022& 0.017 & 0.007 & 0.014& 0.014& 0.046 \\ \hline
			 $g_n$ &0.016& 0.018 &0.006 &0.016&0.008& 0.048\\ \hline
    & \multicolumn{6}{|c|}{ $\alpha=0.1$ } \\\hline
			  &M1& M2& M3 &M4&M5 &M6\\ \hline
			 $F_n$ &0.041& 0.015 & 0.017 & 0.014& 0.014& 0.082 \\ \hline
			 $g_n$ &0.028& 0.014 &0.012 &0.028&0.012& 0.084\\ \hline
		\end{tabular}
\end{center}
\caption{\small Simulated rejection rates for the classic oscillation detection tests $F_n$ and $g_n$ under Stage 1 null conditions (no oscillation) when $\mu_{k,n} = 0$ and $n=500$. Noises are generated according to Models M1 -M4 specified in Section 5 of the main article and Models M5-M6 specified in Section \ref{sec:class_freq}.   \label{tab:sens_class_freq}}
\end{table}

{
\section{Plug-in block size selection: a general setup}\label{tps}

\subsection{Bias and variance in long-run variance estimation}\label{pre}
Let $\{Y_{i,n}\}_{i=1}^n$ be a non-stationary time series with $Y_{i,n}=H_{i,n}({\cal F}_i)$ and $\mathbb E(Y_{i,n})=0$, where $H_{i,n}$ is a measurable function. Recall  the definition of ${\cal F}_i$ in Section \ref{sec:PLS} of the main article. One important problem in the inference of $\{Y_{i,n}\}$ is to estimate the so-called long-run variance
$L_n:=\E(\sum_{i=1}^nY_{i,n})^2/n.$ See for instance \cite{andrews1991heteroskedasticity} among many others. The OBMB algorithm utilized in this paper amounts to estimating $L_n$ by 
\begin{eqnarray}\label{eq:lrv_est}
\hat{L}_{m,n}=\sum_{j=1}^{n-m+1}{\cal S}^2_{j,m}/(n-m+1),
\end{eqnarray}
where ${\cal S}_{i,m}=\sum_{k=i}^{i+m-1}Y_{k,n}/\sqrt{m}$ and $m$ is a block size selected by the user satisfying $m\rightarrow \infty$ with $m/n\rightarrow 0$ as $n\rightarrow\infty$.  Choosing an appropriate tuning parameter $m$ is a difficult problem in practice. If $\{Y_{i,n}\}$ is a stationary AR(1) process, \cite{andrews1991heteroskedasticity} proposed a plug-in method for selecting $m$ which balances the bias and variance of $\hat{L}_{m,n}$. For general classes of non-stationary processes, there exist some data-driven methods in the literature; see for instance \cite{politis2004automatic} and \cite{zhou2013}. However, those data-driven methods are without theoretical justifications and it is unknown whether they possess certain optimality properties (asymptotically). Here we shall propose a plug-in block size selection method for a general class of non-stationary time series. Similar to  \cite{andrews1991heteroskedasticity}, we aim to find $m$ to minimize the mean squared error
$\E(\hat{L}_{m,n}-L_n)^2$. Unlike \cite{andrews1991heteroskedasticity} who calculated $\E(\hat{L}_{m,n}-L_n)^2$ for stationary AR(1) processes, our calculation and estimation are carried out for a general class of non-stationary time series. 

We shall start with calculating the variance of $\hat{L}_{m,n}$. By the same techniques as those in Lemma 1 of \cite{zhou2013}, it can be shown that under appropriate moment and short-range-dependent conditions on $\{Y_{i,n}\}$,
\begin{eqnarray}\label{eq:v_l}
\frac{n}{m}\E(\hat{L}_{m,n}-\E(\hat{L}_{m,n}))^2-C^2_{1,n}=o(1)
\end{eqnarray} 
where $C_{1,n}=O(1)$. The detailed expression for $C^2_{1,n}$ is a complicated function of the fourth cumulants of $\{Y_{i,n}\}$ and will not be discussed in this supplementary material for simplicity. On the other hand, by the same techniques as those in Lemmas 2 - 4 of \cite{zhou2013}, one can derive that under appropriate moment and short-range-dependent conditions on $\{Y_{i,n}\}$ the bias of $\hat{L}_{m,n}$ satisfies that
\begin{eqnarray}\label{eq:b_l}
m\E(\hat{L}_{m,n}-L_n)-C_{2,n}=o(1),
\end{eqnarray} 
where $C_{2,n}=2\sum_{j=1}^{\sqrt{n}} j\Gamma_n(j)$ with 
$$\Gamma_n(j)=\sum_{i=1}^{n-j}\mbox{Cov}(Y_{i,n},Y_{i+j,n})/(n-j).$$
It is clear that $C_{2,n}=O(1)$ under an appropriate short-range dependence condition. As a result, for large $n$ the MSE of $\hat{L}_{m,n}$ is approximately equal to $\frac{m}{n}C^2_{1,n}+\frac{1}{m^2}C^2_{2,n}$. Minimizing the MSE with respect to $m$, the (asymptotically) optimal $m$ should be chosen at 
\begin{eqnarray}\label{eq:optimal_band}    
m^*=\Big(\frac{2C^2_{2,n}}{C^2_{1,n}}\Big)^{1/3}n^{1/3}.
\end{eqnarray}

In the following two subsections, we shall discuss the estimation of $C_{1,n}$ and $C_{2,n}$, respectively.
\subsection{Estimating the asymptotic variance of $\hat{L}_{m,n}$}\label{sec:avl}
Observe that $\hat{L}_{m,n}$ is the average of the non-stationary time series $\{\mathcal{S}^2_{j,m}\}$. As a result, on the high level the asymptotic variance of $\hat{L}_{m,n}$ can be estimated in the same way as that of $\sum_{i=1}^nY_{i,n}/n$. The main difference is that $\E \mathcal{S}^2_{j,m}>0$ while $\E Y_{i,n}=0$. To overcome the non-zero mean, we propose to use the differenced block sum of $\{{\mathcal{S}^2_{j,m}}\}$ to cancel its mean. Specifically, for a given block size $m^o$ (say $m^o=\lfloor n^{1/3}\rfloor$), define 
\begin{eqnarray*}
H_{i,m^o}=\Big[\sum_{j=i-m^o+1}^i\mathcal{S}^2_{j,m^o}-\sum_{j=i+1}^{i+m^o}\mathcal{S}^2_{j,m^o}\Big]/\sqrt{2m^o},    
\end{eqnarray*}
$i=m^o,m^o+1,\cdots,n-2m^o+1$. As mentioned above, the purpose of introducing the differenced block sum $\sum_{j=i-m^o+1}^i\mathcal{S}^2_{j,m^o}-\sum_{j=i+1}^{i+m^o}\mathcal{S}^2_{j,m^o}$ is to cancel out the mean of $\mathcal{S}^2_{j,m^o}$ which is indeed the case for most of $i$ when $\{Y_{i,n}\}$ is a piece-wise locally stationary time series. By \eqref{eq:v_l}, the quantity $C_{1,n}^2$ can then be estimated by
\begin{eqnarray}\label{eq:hatc1}
\hat{C}_{1,n}^2=\sum_{i}H_{i,m^o}^2/[(n-3m^o+2)m^o].
\end{eqnarray}
The consistency of $\hat{C}_{1,n}^2$ is beyond the scope of this paper and will be pursued in a separate research paper.

\subsection{Estimating the asymptotic bias of $\hat{L}_{m,n}$}\label{sec:abl}
We propose to estimate $C_{2,n}$ by
\begin{eqnarray}\label{eq:biasl}
\hat C_{2,n} = \frac{2}{n-3l_n+1}\sum_{j=1}^{n-3l_n+1}\tilde{\cal S}_{j,l_n}(\tilde{\cal S}_{j+l_n,l_n}-\tilde{\cal S}_{j+2l_n,l_n}),    
\end{eqnarray}
where $\tilde{\cal S}_{j,m}=\sqrt{m}{\cal S}_{j,m}=\sum_{k=j}^{j+m-1}Y_{k,n}$ and $l_n$ is a user-chosen parameter such that $l_n\rightarrow \infty$ with $l_n/n\rightarrow 0$. According to our simulation studies and data analysis, we recommend choosing $l_n=\lfloor n^{1/6}\rfloor$. The purpose of using the difference $\tilde{\cal S}_{j+l_n,l_n}-\tilde{\cal S}_{j+2l_n,l_n}$ instead of $\tilde{\cal S}_{j+l_n,l_n}$ in \eqref{eq:biasl} is to cancel out possible non-zero means in $\tilde{\cal S}_{j,m}$ in practical applications with similar arguments as those in Section \ref{sec:avl} under the definition of $H_{i,m^o}$. Combining our discussions in Sections \ref{sec:avl} and \ref{sec:abl}, our plug-in block size selector chooses
\begin{eqnarray}\label{eq:bandwidth} 
\hat m^*=\max\Big(1, \Big\lfloor\Big(\frac{2\hat C^2_{2,n}}{\hat C^2_{1,n}}\Big)^{1/3}n^{1/3}\Big\rfloor\Big).    
\end{eqnarray}

Under the assumption that $\{Y_{i,n}\}$ is sufficiently short-range dependent, elementary calculations yield that
$\mathbb E(\hat C_{2,n})=C_{2,n}+o(1)$. Following similar arguments as those in 
Lemmas 1 - 4 of \cite{zhou2013} and under appropriate moment and short memory assumptions, it can be shown that $\E(\hat{C}_{2,n}-\E\hat{C}_{2,n})^2\rightarrow 0$. 
Again, the detailed proofs are beyond the scope of this paper and will be pursued in a separate research manuscript.

\subsection{Block size selection for differenced long-run variance estimators}\label{sec:dlrv}
In our stage two change point estimation algorithm, the long-run variance $L_n=\E(\sum_{i=1}^nY_{i,n})^2/n$ is approximated by a slightly different method; that is, we use the following
\begin{eqnarray}\label{eq:dlv}
\hat{L}^*_{m,n}=\sum_{j=1}^{n-2m+1}\bar{\cal S}^2_{j,m}/(n-2m+1)
\end{eqnarray}
to estimate $L_n$, where $\bar{\cal S}_{i,m}=(\sum_{k=i}^{i+m-1}Y_{k,n}-\sum_{k=i+m}^{i+2m-1}Y_{k,n})/\sqrt{2m}$. We shall call $\hat{L}^*_{m,n}$ a differenced long-run variance estimator. The plug-in method discussed in Sections \ref{sec:avl} and \ref{sec:abl} above can be easily extended to selecting block size $m$ for the differenced long-run variance estimator. Specifically, following the same discussions as those in Sections \ref{sec:avl} and \ref{sec:abl}, it can be shown that 
$\E(\hat{L}^*_{m,n}-\E(\hat{L}^*_{m,n}))^2=\frac{m}{n}[\bar{C}^2_{1,n}+o(1)]$
and $\E(\hat{L}^*_{m,n}-L_n)=[\bar C_{2,n}+o(1)]/m$, where $\bar C_{2,n}=3\sum_{j=1}^{\sqrt{n}}j\Gamma_n(j)=3C_{2,n}/2$. Now $\bar C^2_{1,n}$ can be estimated the same way as that in Section \ref{sec:avl} with ${\cal S}_{j,m^o}$ therein replaced with $\bar {\cal S}_{j,m^o}$. Meanwhile, $\bar C_{2,n}$ can be estimated by $3\hat C_{2,n}/2$, where $\hat C_{2,n}$ is defined in Section \ref{sec:abl}.

\subsection{Extension to multiple time series}\label{sec:mts}
The oscillating frequency detection stage of our paper involves choosing an appropriate block size for a multiple time series. With a slight abuse of notation, let $\{Y_{i,n}=(Y_{i,1,n},\cdots,Y_{i,h,n})^\top\}$ be a centered $h$-dimensional time series.
Then our OBMB algorithm estimates the so-called long-run covariance matrix
$\E([\sum_{i=1}^nY_{i,n}][\sum_{i=1}^nY_{i,n}]^\top)/n$. Directly assessing the goodness-of-fit for estimators of the latter covariance matrix is computationally intensive when $h$ is large. As a result, we suggest choosing a block size $m$ that minimizes the averaged mean squared error across dimensions. That is, we choose $m$ such that 
\begin{eqnarray}\label{eq:mloss}
\sum_{i=1}^h\E(\hat{L}_{i,m,n}-L_{i,n})^2/h
\end{eqnarray}
is minimized, where $L_{i,n}=\E(\sum_{j=1}^nY_{j,i,n})^2/n$, and $\hat{L}_{i,m,n}=\sum_{k=1}^{n-m+1}{\cal S}^2_{k,i,m}/(n-m+1)
$
with ${\cal S}_{k,i,m}=\sum_{r=k}^{k+m-1}Y_{r,i,n}/\sqrt{m}$,  $i=1,2,\cdots,h$. For each dimension $i$, The bias and variance of $\hat{L}_{i,m,n}$ can be quantified and estimated using the same techniques 
as those discussed in Sections \ref{sec:avl} and \ref{sec:abl}. Therefore the block size $m$ that minimizes \eqref{eq:mloss} can be estimated using essentially the same plug-in method discussed in the aforementioned sections. Details are omitted here.

}

\section{Gaussian approximation without variance lower bounds}\label{sec:A}

	This section of the appendix deals with approximations to the sum of an $h$-dimensional non-stationary time series $\{x_i\}$ (either real or complex) by a centered Gaussian time series $\{y_i\}$ with the same covariance structure without assuming there is a lower bound for the coordinate-wise variances of the normalized sum of $\{x_i\}$.

	First we need the following Lemma \ref{lemma1} which extends Nazarov's Inequality \citep{nazarov2003} to Gaussian random vectors without variance lower bound.
	
	\begin{lemma}[\textbf{An Extended Version of Nazarov's Inequality}]\label{lemma1}
		Let $Y = (Y_1,...,Y_h)^\top$ be a centred Gaussian random vector, where $h>1$. Let $y  \in \mathbb{R}^h$ satisfy $y_i> c$ or $y_i + a < -c$ for each $i$ for some $c > 0$ and $a>0$. Then, for any $b>0$,
		$$
		\pr(Y \leq y +a ) - \pr(Y \leq y)  \leq 4 (a/b)  \sqrt{\log{h}} + \frac{bh}{c\sqrt{2\pi}} \exp(-(c/b)^2/2) .
		$$
		where $Y \leq y$ means $Y_i\leq y_i$ for $i=1\ldots,h$.
	\end{lemma}
	
	\begin{proof}
	For some constant $b >0$, we first separate the coordinates whose variances are greater than $b$ and the coordinates whose variances are smaller than $b$.
		First note that
		$$
		\mathbb P(A\cap B) =\mathbb  P(A) -\mathbb P(A \cap B^c ) \geq \mathbb P(A) -\mathbb P(B^c ).
		$$
Therefore, by taking $A=\{Y_i \leq y_i  , \forall i  \text{ where } \sigma_i >  b\}$ and $B=\{Y_i \leq y_i  , \forall i  \text{ where } \sigma_i \leq b\}$, we have
		\[ \pr(Y_i \leq y_i  ) \geq \pr(Y_i \leq y_i  , \forall i  \text{ where } \sigma_i >  b) -\pr(Y_i > y_i  , \exists i  \text{ where } \sigma_i \leq  b).  
		\]
		We can further expand using the above fact:
		\begin{align*}
		&\pr(Y \leq y +a ) - \pr(Y \leq y) \\
		= \,& \pr(Y_i \leq y_i + a ,\forall i \text{ where } \sigma_i \leq b \text{ and } Y_i \leq y_i + a ,\forall i \text{ where } \sigma_i > b)  \\
		 \,&\quad - \pr(Y_i \leq y_i  , \forall i  \text{ where } \sigma_i \leq  b \text{ and } Y_i \leq y_i  , \forall i  \text{ where } \sigma_i >  b) \\
		 \leq \,&\pr(Y_i \leq y_i +a , \forall i  \text{ where } \sigma_i >  b) -  \pr(Y_i \leq y_i  , \forall i  \text{ where } \sigma_i >  b)\\
		 \,&\quad + \pr(Y_i > y_i  , \exists i  \text{ where } \sigma_i \leq  b).
		\end{align*}
		%
		%
		Then, by Nazarov's inequality \citep{nazarov2003}, we can bound the part where the condition $\sigma_i > b$ holds. Specifically, 
		\begin{align*}
		&\pr(Y_i \leq y_i + a ,\forall i \text{ where } \sigma_i > b) - \pr(Y_i \leq y_i  , \forall i  \text{ where } \sigma_i >  b) \\\leq &\,(a/b)  (\sqrt{2 \log{h}} + 2) \leq 4 (a/b)  \sqrt{\log{h}}  .
		\end{align*}
		To evaluate $\pr(Y_i > y_i  , \exists i  \text{ where } \sigma_i \leq  b)$, we consider the assumption that $y_i> c$ for some $c > 0$ for all $i$ so that $\sigma_i\leq b$ is satisfied. In this case, we need to bound the tail probability under the condition $\sigma_i \leq  b$. Let $n(b)$ be the number of coordinates satisfy $\sigma_i \leq b$. Then
		\begin{align*}
		\pr(Y_i > y_i  , \exists i  \text{ where } \sigma_i \leq  b) &=   \pr(Y_i/\sigma_i > y_i/\sigma_i, \exists i  \text{ where } \sigma_i \leq  b) \\
		& \leq \sum_{k = 1}^{n(b)} \int_{y_i/\sigma_i}^{\infty} \frac{1}{\sqrt{2\pi}} \exp(-t^2/2) dt \quad \text{(Union bound)}  \\
		& \leq \frac{n(b)}{\sqrt{2\pi}} \int_{c/b}^{\infty} \exp(-t^2/2) dt
		\leq  \frac{bh}{c \sqrt{2\pi}}  \exp(-(c/b)^2/2) .
		\end{align*}	
		Thus, we have the claim when the assumption $y_i> c$ for some $c > 0$ is satisfied.
			
		To finish the proof, we consider the final case where there exists an $y_i$ such that $y_i + a <-c$  and $\sigma_i \leq b$ are both satisfies. In this case, by a trivial bound we have
		\[\pr(Y \leq y+a) - \pr(Y \leq y) \leq \pr(Y \leq y+a) \leq  \pr(Y_i \leq y_i+a)  \leq  \frac{bh}{c\sqrt{2\pi}} \exp(-(c/b)^2/2).\]
		By putting the above together, we obtain the claim.
	\end{proof}
	
	The following corollary is a direct application of \textbf{Lemma} \ref{lemma1} by picking proper constants $b$ and $c$.
	
	\begin{corollary}\label{cor1}
		Let $Y = (Y_1,...,Y_h)^\top$ be a centred Gaussian random vector, where $h>1$. Take $\delta > 0$.  Then for $c = 2\log(h)^{-\delta} $ and $a>0$, we have
		\[ 
		\sup_{|x| > c+ a } \pr\left(\Big|\max_{1 \leq j \leq h } Y_j - x\Big| \leq a\right) \leq 8 a (\log{h})^{1+\delta} +  \frac{1}{2 \sqrt{2\pi}(\log h)^{1/2} h} .
			\]
	\end{corollary}

	\begin{proof}
		Pick $b = \log(h)^{-\frac{1}{2} - \delta}$. Then
		$$
		\frac{h}{\sqrt{2\pi}} \frac{b}{c} \exp(-(c/b)^2/2)  =   \frac{h}{2 \sqrt{2\pi} (\log h)^{1/2}}\exp(-2(\log h)) = \frac{1}{2 \sqrt{2\pi}(\log h)^{1/2} h}.
		$$
		Combining the previous results we get that for $y\in \mathbb{R}^h$ so that $y_i > c$ or $y_i<-c- a$, 
		$$
		\pr(Y \leq y +a ) - \pr(Y \leq y)  \leq 4 a (\log{h})^{1+
			\delta} +  \frac{1}{2\sqrt{2\pi}\log^{1/2} (h) h} .
		$$
		We thus get the proof by writing the above quantities entrywisely and noting that 
		\[
		\pr\left(\Big|\max_{1 \leq j \leq h } Y_j - x\Big| \leq a\right)=\pr\left(\max_{1 \leq j \leq h } Y_j \leq x+a\right)-\pr\left(\max_{1 \leq j \leq h } Y_j \leq x-a\right)\,.
		\]
		
	\end{proof}
	
	The next proposition is an extension of Proposition A.1 in \cite{zhang2018}. We shall first introduce some notation used in the latter paper.

Let $\{ x_i\}=$ $\{(x_{i,1},\cdots,x_{i,h})^\top\}$ be a centered $h$-dimensional $M$-dependent times series. Take any truncation levels $M_x>0$. 
For $N\geq M$ and $N$, $M$, $r\rightarrow +\infty$ as $n \rightarrow + \infty$, define block sums:
$$
A_{ij} := \sum_{l = iN+(i-1)M - N +1}^{iN+ (i-1) M} x_{l,j}, \quad B_{ij} := \sum_{l = i(N+M) - M+1}^{i (N+M)} x_{l,j},
$$
the block sum for truncated $\chi_{i,j}:= (x_{i,j} \wedge M_x)\vee(-M_x)$:
$$
\check{A}_{ij} := \sum_{l = iN+(i-1)M - N +1}^{iN+ (i-1) M} \chi_{l,j}, \quad \check{B}_{ij} := \sum_{l = i(N+M) - M+1}^{i (N+M)} \chi_{l,j},
$$
and the block sum for the truncated and centered $\tilde{x}_{i,j} := (x_{i,j} \wedge M_x)\vee(-M_x) - \mathbb{E} [(x_{i,j} \wedge M_x)\vee(-M_x)]$:
$$
\tilde{A}_{ij} := \sum_{l = iN+(i-1)M - N +1}^{iN+ (i-1) M} \tilde{x}_{l,j}, \quad \tilde{B}_{ij} := \sum_{l = i(N+M) - M+1}^{i (N+M)} \tilde{x}_{l,j}.
$$
Let $\varphi(M_x)$ be the smallest finite constant which satisfies
\[ 
\mathbb{E}(A_{ij} - \check{A}_{ij} )^2 \leq N \varphi^2(M_x) ,\quad \mathbb{E}(B_{ij} - \check{B}_{ij} )^2 \leq M \varphi^2(M_x)
\]
uniformly for $i$ and $j$.
Also,
let $\phi(M_x)$ be a constant which satisfies 
\[ 
\max_{1\leq j,k\leq h } \frac{1}{n} \sum_{i =1}^n \left|  \sum_{l=(i-M) \vee 1}^{(i+M)\wedge n } \left( \mathbb{E} x_{ij} x_{lk}  - \mathbb{E}\tilde{x}_{ij} \tilde{x}_{lk} \right)  \right| \leq \phi(M_x). 
\]  

Let $\{y_i\}$ be the Gaussian counterpart of $\{x_i\}$ with the same covariance structure of $\{x_{i}\}$, and take any truncation level $M_y$.
Define $\varphi(M_y)$ and $\phi(M_y)$ similarly based on $\{y_{ij}\}$. Set $ \phi(M_x,M_y) := \phi(M_x) + \phi(M_y)$, and set $\varphi(M_x,M_y) := \varphi(M_x) \vee \varphi(M_y)$.

Recall that $\FF_i=(\cdots,e_{i-1},e_i)$ and $\{e_i\}_{i\in\mathbb{Z}}$ are i.i.d. random variables. If $\{x_i={\cal G}_{i,n}(\FF_i)\}$ is not $M$-dependent, let $\{x_i^{(M)}:=\E (x_i|e_i,e_{i-1},\cdots,e_{i-M})\}$ be the $M$-dependent approximation to $\{x_i\}$. Define $\{ y_i^{(M)}\}$ as the $M$-dependent sequence of Gaussian random variables which preserves the covariance structure of $\{x_i^{(M)}\}$. Similarly we can define $A_{ij}^{(M)}$, $\tilde{A}_{ij}^{(M)}$, $\check{A}_{ij}^{(M)}$, $B_{ij}^{(M)}$, $\tilde{B}_{ij}^{(M)}$ and $\check{B}_{ij}^{(M)}$ based on $\{ x_{i}^{(M)}\}$, and hence 
$\phi^{(M)}(M_x,M_y)$ and $\varphi^{(M)}(M_x,M_y)$ are defined similarly based on $\{ x_i^{(M)}\}$ and $\{ y_i^{(M)}\}$.

	\begin{proposition}\label{prop:A1}
	Let $\{ x_i\}_{i=1}^n=$ $\{(x_{i1},\cdots,x_{ih})^\top\}_{i=1}^n$ be a centered $h$-dimensional $M$-dependent times series. Let $\{y_i\}$ be its Gaussian counterpart. 
	Define $M_{xy} =\max \{M_x, M_y\}$.
	Suppose $2\sqrt{5} \beta (6M+1)  M_{xy}/\sqrt{n} \leq 1 $, where $\beta>0$ is a constant. Also, suppose $M_x > u_x(\gamma) $ and $ M_y > u_y(\gamma)$ for some $\gamma \in (0,1)$, where $u_x(\gamma)$ is the $(1-\gamma)$-quantile of $\max_{i,j}|x_{ij}|$ with $u_y(\gamma)$ is defined similarly. Further, assume that $\max_{1\le i\le h}|\sum_{k,l=1}^{n}\cov(x_{ki},x_{li})/n|\le a_1$ for some finite constant $a_1$. 
	Define $\bar{m}_{x,k}:=\max_{1\le j\le h}[\sum_{i=1}^n\E|x_{ij}|^k/n]^{1/k}$, $k=1,2,\cdots,$ 
	\[
	T_X:=\max_{1\le j\le h}\sum_{i=1}^nx_{ij}/\sqrt{n}
	\] 
	and $\bar{m}_{y,k}$ and $T_Y$ are defined similarly.
	Then, for any $\psi >0$,
		\begin{align}
		&\sup_{|t| > d_{h} } |\mathbb P(T_X \leq t ) -\mathbb P(T_Y\le t)| \nonumber\\
		\lesssim \,&(\psi^2 + \psi \beta)  \phi(M_x,M_y) \\
		&+(\psi^3+\psi^2 \beta + \psi \beta^2) \frac{(2M+1)^2}{\sqrt{n}}(\bar{m}^3_{x,3} +\bar{m}^3_{y,3} ) \nonumber\\
		& + \psi \varphi(M_x,M_y) \sqrt{\log(h/\gamma)} + \gamma \nonumber\\ 
		&	+ (\beta^{-1}\log(h) + \psi^{-1})  (\log h)^{1+ \delta} + h^{-1} (\log h)^{-1/2} ,\nonumber
		\end{align}
		where $d_h := \beta^{-1}\log(h) + \psi^{-1} + 2\log(h)^{-\delta}$.
	\end{proposition}
	
	\begin{proof}
		
		Let
		$$
		g_0(x) = \begin{cases}
		0, & x \geq 1,\\
		30 \int_x^1 s^2 (1-s)^2 ds , & 0 < x < 1,\\
		1, & x \leq 0.
		\end{cases}
		$$
		and pick $g(s) = g_0 (\psi (s - t -e_\beta ))$ with $e_\beta = \beta^{-1} \log h$ and $t$ to be chosen later. Denote 
		\begin{equation}\label{definiteion:G_k}
		G_k :=\sup_{x\in\mathbb{R}}\partial^k g(x)/\partial x^k, 
		\end{equation}
		$k=1,2,\cdots$. We have $G_0 \lesssim 1$, $G_1 \lesssim \psi$, $G_2 \lesssim \psi^2$ and $G_3 \lesssim \psi^3$. Moreover, the function also satisfies
		$$
		\mathbb{I}(x \leq t + e_\beta) \leq g(x) \leq \mathbb{I}(x\leq t + e_\beta + \psi^{-1}), \forall x \in \mathbb{R},
		$$
		where $\mathbb{I}$ is the indicator function.
		Define $m(y)=g\circ F_\beta(y)$, where 
		$$
		F_\beta(y)=\beta^{-1}\sum_{i=1}^h\exp(\beta y_i),\quad y=(y_1,\cdots,y_h)^\top\in\mathbb{R}^h.
		$$
		Denote $X:=(X_1,X_2,\cdots,X_h)^\top=\sum_{i=1}^n x_i/\sqrt{n}$ and $Y:=(Y_1,Y_2,\cdots,Y_h)^\top$ is defined similarly. Based on Equation (39) in the proof of Proposition A.1 of \cite{zhang2018}, we have the following
		\begin{align}
		|\E[m(X) - m(Y)]| \lesssim &(\psi^2 + \psi \beta)  \phi(M_x,M_y)\nonumber\\
		&+(\psi^3+\psi^2 \beta + \psi \beta^2) \frac{(2M+1)^2}{\sqrt{n}}(\bar{m}^3_{x,3} +\bar{m}^3_{y,3} )\nonumber\\
		& + \psi \varphi(M_x,M_y) \sqrt{\log(h/\gamma)} + \gamma   := B.
		\end{align}
		Now assume $|t| >  e_{\beta} + \psi^{-1} +2\sqrt{2} \log^{-\delta} h $, we have
		\begin{align*}
		\pr(\max_{1\leq j \leq h } X_j \leq t ) &\leq \pr( F_\beta (X) \leq t + e_\beta  ) \leq \E [g (F_\beta(X)) ] \\
		& \leq \E[g(F_\beta(Y))] + B\\
		& \leq \pr (F_\beta(Y) \leq t + e_\beta + \psi^{-1}) +B\\
		& \leq \pr \left(\max_{1 \leq j \leq h} Y_j \leq t +e_\beta + \psi^{-1} \right) + 
		B.
		\end{align*}
		Then, by Lemma \ref{lemma1} and assuming $|t| >  e_{\beta} + \psi^{-1}  + 2(\log h)^{-\delta} $, we have
		\begin{align*}
		&\pr \left(\max_{1 \leq j \leq h} Y_j \leq t +e_\beta + \psi^{-1} \right) - \pr \left(\max_{1 \leq j \leq h} Y_j \leq t \right)\\  \lesssim &\,(e_\beta + \psi^{-1}) (\log{h})^{1+\delta} + h^{-1} (\log h)^{-1/2} .
		\end{align*}
		Thus, since $T_X=\max_{1\leq j \leq h } X_j$ and $T_Y=\max_{1 \leq j \leq h} Y_j$, we conclude
		\begin{align*}
		\pr \left(T_X \leq t \right) - \pr\left(T_Y \leq t\right)  \leq B +  (e_\beta + \psi^{-1}) (\log{h})^{1+\delta} + h^{-1} (\log h)^{-1/2}.
		\end{align*}
		The opposite direction can be proved similarly by noting that
		$$
		\pr\left( \max_{1 \leq j \leq h } X_j \leq x\right) \geq \pr\left( \max_{1 \leq j \leq h }Y_j \leq x - e_\beta - \psi^{-1} \right) - B
		$$
		and
		\begin{align*}
		&\pr\left( \max_{1 \leq j \leq h }Y_j \leq x - e_\beta - \psi^{-1} \right) - \pr\left( \max_{1 \leq j \leq h } Y_j \leq x\right)\\  \geq &\, C[ -  (e_\beta + \psi^{-1}) (\log{h})^{1+\delta} - h^{-1} (\log h)^{-1/2}]
		\end{align*}
for some constant $C$.
	\end{proof}

Let $\{x_{i}\}_{i=1}^n$ be a general centered $h$-dim non-stationary time series satisfying 
\begin{itemize}
\item[(A0)] $x_{i}={\cal G}_{i,n}(\cdots,e_{i-1},e_i) \in \R^h$, $ i=1,\cdots, n$ (triangular array), where $e_i$, $i\in\mathbb{Z}$, are i.i.d. random elements and ${\cal G}_{i,n}$ are $h$-dim vector-valued Borel-measurable functions. For an integer $k\ge 0$ and a positive real number $q \ge 1$, define the physical dependence measures of $\{x_{i}\}$
\begin{eqnarray*}
\theta_{j,k,q}=\max_{1\le i\le n}\|{\cal G}_{i,j,n}(\cdots,e_{i-1},e_i)-{\cal G}_{i,j,n}(\cdots,\hat e_{i-k},e_{i-k+1},\cdots, e_{i-1},e_i)\|_q,
\end{eqnarray*}
 where ${\cal G}_{i,j,n}(\cdot)$ is the $j$-th component function of ${\cal G}_{i,n}$, $j=1,2,\cdots, h$, and $\hat e_{i-k}$ is identically distributed as $e_{i-k}$ and is independent of $\{e_i\}_{i\in\mathbb{Z}}$. 
\end{itemize}
We make the following three assumptions for the time series $\{x_i\}$ which corresponds to Assumptions (2.1) to (2.3) of \cite{zhang2018} but without assuming lower variance bounds. 
\begin{itemize}
\item[(A1)] Assume  that $\max_{1\leq i \leq n}\max_{1\le j\le h}  \E x^4_{i,j} < c_1$ for some finite $c_1>0$ and there exists $\mathcal{D}_n >0$ such that one of the following two conditions holds:
\begin{align}\label{eqa12}
	\max_{1 \leq i \leq n} \E \exp( \max_{1\le j\le h}|x_{i,j}| / \mathcal{D}_n ) \leq 1,
	\end{align}	 
	or
	\begin{align}\label{eqa11}
	\max_{1 \leq i \leq n} \E g(\max_{1\le j\le h} |x_{i,j}| / \mathcal{D}_n ) \leq 1
	\end{align}
	for some  strictly increasing and convex function $g$ defined on $[0,\infty)$ satisfying $g(0)=0$.
\item[(A2)] Assume that there exist $M = M(n) >0$ and $\gamma = \gamma(n) \in (0,1)$ such that 
	\begin{align}
& n^{3/8} M^{-1/2} l_n^{-5/8} \geq C_2 \max \{\mathcal{D}_n l_n ,l_n^{1/2} \}  \textit{ under Condition (\ref{eqa12})}\\
	&n^{3/8} M^{-1/2} l_n^{-5/8} \geq C_1 \max \{\mathcal{D}_n g^{-1}(n/\gamma) ,l_n^{1/2} \}  \textit{ under Condition (\ref{eqa11})}
	\end{align}	
	for $C_1,C_2 > 0$, where $l_n = \log (hn/\gamma) \vee 1$. In both cases, suppose $n^{7/4} M^{-2} l_n^{-9/4} \geq C_3 >0$.
\item[(A3)] Assume that $\max_{1\le i\le h}|\sum_{k,l=1}^n\cov(x_{k,i},x_{l,i})/n|\le a_1$ 
	 and $$ \sum_{j = 0}^\infty j\max_{1\leq k \leq h}\theta_{j,k,3}<a_2$$ for some finite constants $a_1$ and $a_2$, where $\theta_{j,k,q}$ is the $j$-th physical dependence measure of the $k$-th coordinate process of $\{x_i\}$ with respect to the ${\cal L}^q$ norm.	
\end{itemize}


Recall that $\{x_i^{(M)}:=\E (x_i|e_i,e_{i-1},\cdots,e_{i-M})\}$ is the $M$-dependent approximation to $\{x_i\}$ and $\{ y_i^{(M)}\}$ is an $M$-dependent sequence of Gaussian random variables which preserves the covariance structure of $\{x_i^{(M)}\}$.	
\begin{lemma} \label{lemma:Bounds}
	Let $\{x_j\}_{j=1}^n$ be a $h$-dim time series satisfy Assumptions (A0)-(A3).  Then  
	$$
	\phi^{(M)}(M_x,M_y)\le C'(1/M_x + 1/M_y^2)
	$$ 
	and 
	$$
	\varphi^{(M)}(M_x,M_y)\le C''(1/M^{5/6}_x + \sqrt{N}/M_x^3 + 1/M_y^2 ) $$
	 for some finite constants $C'$ and $C''$, where we recall that
	 $\phi^{(M)}(M_x,M_y)$ and $\varphi^{(M)}(M_x,M_y)$ are the versions of $\phi(M_x,M_y)$ and $\varphi(M_x,M_y)$ defined based on $\{x_i^{(M)}\}$ and $\{y_i^{(M)}\}$. 
\end{lemma}
\begin{proof}
	The results follow from steps 2 and 3 in proof of Theorem 2.1 in \cite{zhang2018}.
	From  \cite[step 2 in proof of Theorem 2.1]{zhang2018}, we get $\phi^{(M)}(M_x) \le C'/M_x$ and $\varphi^{(M)}(M_x) \le C''( 1/M_x^{5/6} + \sqrt{N}/M_x^3)$.
	From \citep[step 3 in proof of Theorem 2.1]{zhang2018}, we get $ \phi^{(M)}(M_y) \le C'/M_y^2$ and $\varphi^{(M)}(M_y) \le C''/M_y^2 $ for $C',C'' >0$. 
\end{proof}

\begin{theorem}\label{thmA1}
	 Let $\{x_i\}_{i=1}^n$ be a time series satisfy Assumptions (A0)-(A3). Let $\{y_j \in \R^h\}$ be Gaussian random vectors with the same covariance structure as $\{x_j\}$. For any integer $M>0$ and positive real number $q\ge 1$, let $\Xi_M := \max_{1\leq k \leq h} \sum_{j = M}^\infty j \theta_{k,j,2}$ and $\Theta_{M,i,q} :=\sum_{j=M}^\infty \theta_{i,j,q}$, where $i=1,\ldots,h$. Assume that $q\ge 2$ and $\max_{1\le j\le h}\Theta_{0,j,q}<\infty$. Then
	\[
	\sup_{|x|  \gtrsim d_{n,h} }| \pr( T_X \leq x ) - \pr(T_Y \leq x) |\lesssim G(n,h), 
	\]
	where $T_X:= \max_{1\leq j \leq h} \sum_{i =1 }^n x_{ij}/\sqrt{n}$, $T_Y:= \max_{1\leq j \leq h} \sum_{i =1 }^n y_{ij}/\sqrt{n}$,
	\begin{align}\label{Definition:dnh}
	d_{n,h} := n^{-1/8} M^{1/2} l^{11/7}_n +l_n^{-\delta} + \Xi^{1/3}_M l_n^{\delta/3}  
	\end{align}
	and
	\begin{align}
	        G(n,h)  := &\, n^{-1/8} M^{1/2} l^{11/7}_n + \gamma \nonumber\\
	        &+ (n^{1/8} M^{-1/2} l^{-3/8}_n )^{q/(1+q)} \left( \sum_{j = 1}^h \Theta^q_{M,j,q} \right)^{1/(1+q)} + \Xi^{1/3}_M l_n^{\delta /3}.\label{Definition:Gnh}
	\end{align} 
\end{theorem}
\begin{remark}
Note that typically we assume that $\theta_{k,j,q}$ decays at a rate of $j^{-\alpha}$ for some $\alpha>2$ for each $k$. As a result both $\Xi_M$ and $\Theta_{M,i,q}$ decays at a polynomial rate with respect to $M$. In most cases one chooses $M$ that diverges at a small polynomial rate of $n$. Meanwhile $l_n$ typically diverges at the rate $O(\log n)$ according to its definition and the fact that the dimension $h$ diverges at a polynomial rate of $n$ under our assumptions. As a result it is easy to check that the right hand sides of \ref{Definition:dnh} and \ref{Definition:Gnh} converges to 0 in most applications. 
\end{remark}

\begin{proof}
    Follow the same $M$-dependent sequence construction in Lemma \ref{lemma:Bounds} and get $x_i^{(M)}$ and $y_i^{(M)}$.
	By construction, $x_{1,j}$ and $x_{1+l,k}^{(l-1)} $ are independent for any $1 \leq j , k \leq h$. 
Denote $X^{(M)}:=\sum_{i=1}^n x_i^{(M)}/\sqrt{n}$ and $Y^{(M)}:=\sum_{i=1}^n y_i^{(M)}/\sqrt{n}$.
	The triangular inequality and (16) in \cite{zhang2018} imply that
    \[
    |\E[m(X) - m(Y^{(M)})]| \lesssim |\E[m(X^{(M)}) -m(Y^{(M)})]| + (G_0 G_1^q)^{1/(1+q)}\left(  \sum_{j = 1}^h \Theta^q_{M,j,q} \right)\,,
    \]	
	where we recall the definitions of $G_0$ and $G_1$ from \eqref{definiteion:G_k}. Following the arguments in the proof of Proposition A.1 in \cite{zhang2018} and using Proposition \ref{prop:A1} {\em if} all conditions are satisfied, we have
	\begin{align}
	& \sup_{|x| \gtrsim d_{n,h} } |\pr(T_{X^{(M)}} \leq x ) -\pr(T_{Y^{(M)}} \leq x )| \nonumber \\  
	 \lesssim \,& (\psi^2 + \psi \beta)  \phi^{(M)}(M_x,M_y)+(\psi^3 + \psi^2 \beta + \psi \beta^2) \frac{(2M+1)^2}{\sqrt{n}}(\bar{m}^3_{x,3} + \bar{m}^3_{y,3} ) \nonumber \nonumber \\
	&  + G_1 \varphi^{(M)}(M_x,M_y) \sqrt{8\log(h/\gamma)} + G_0 \gamma + (\psi)^{1/(1+q)} \left( \sum_{j=1}^{h} \Theta^q_{M,j,q} \right)^{1/(1+q)}\nonumber\\
	&  +  (\beta^{-1}\log(h) + \psi^{-1})  (\log h)^{1+ \delta} + h^{-1} (\log h)^{-1/2}\,. \label{TheoremA1:key bound} 
	\end{align}
	By Assumption (A2) where $n^{7/4}M^{-1} l^{-9/4 } > C_3 M$, we have
	\begin{align} \label{eqn:Gaussian}
	& (\psi^2 + \psi \beta)  \phi^{(M)}(M_x,M_y) \lesssim n^{-1/8}M^{1/2} l_n^{-9/4} ,\\
	& (\psi^3 + \psi^2 \beta + \psi \beta^2) \frac{(2M+1)^2}{\sqrt{n}} \lesssim n^{-1/8}M^{1/2} l_n^{-9/4} ,\\
	& \psi \varphi(M_x,M_y)\sqrt{8\log(h/\gamma)} \lesssim n^{-1/8} M^{1/2} l_n^{7/8},\\
	& (\beta \log(h) + \psi^{-1}) (\log(h))^{1+\delta} \lesssim \frac{ l^{5/2}_n M u}{\sqrt{n}} \lesssim n^{-1/8} M^{1/2} l_n^{15/8} .\label{eqn:Gaussian2}
	\end{align}
	Finally, by Step 5 in the proof of Theorem 2.1 in \cite{zhang2018} and Theorem \ref{thm:B1}, it follows that \begin{equation}\label{eq:cond1} \sup_{|t| \gtrsim d_h' } \left| \pr(T_Y \leq t) - \pr(T_{Y^{(M)}} \leq t)\right| \lesssim \Xi^{1/3}_M (\log h)^{1+\delta} + \frac{1}{(\log h)^{1/2} h }, \end{equation} where $d_h' = \Xi^{1/3}_M \log(h)^{\delta/3}+  \log(h)^{-\delta}$.
	The  result follows from equations (\ref{eqn:Gaussian}) -(\ref{eq:cond1}).
	
The left is verifying the conditions in Proposition \ref{prop:A1} since \eqref{TheoremA1:key bound} is based on $x^{(M)}_i$ and $y^{(M)}_i$. 
Consider $g$ in Assumption (A1). We have
	\begin{align*}
	    & \pr\Big(\max_{1\leq i \leq n} \max_{1 \leq j \leq h } |x_{ij}^{(M)}| > u \Big) 
	     \leq \sum_{ i = 1 }^n \pr \left( g\left(  \max_{1\leq j  \leq h} |x_{ij}^{(M)}| /\mathcal{D}_n \right) > g(u/\mathcal{D}_n) \right)\\
	     \leq &\,n \max_{1\leq i \leq n}  \E g\left(  \max_{1\leq j  \leq h} |x_{ij}^{(M)}| /\mathcal{D}_n \right)  /g(u / \mathcal{D}_n) \leq n C_1 /g(u / \mathcal{D}_n) .
	\end{align*}
	The last inequality follows from Jensen's inequality and  Assumption (A1). By setting the above equation to $\gamma$, we get $u \leq \mathcal{D}_n g^{-1}(n/\gamma)$. Similarly, we have
	\begin{align*}
	    & \pr\left(\max_{1\leq i \leq n} \max_{1 \leq j \leq h } |y_{ij}^{(M)}| > u \right) \leq \sum_{ i = 1 }^n \sum_{ j = 1 }^h \pr \left( |y_{ij}^{(M)}| > u\right) \leq n h \exp( - u^2 / (2 \bar{\sigma}^2)),
	\end{align*}
	where $\bar{\sigma}^2 =\max_{1\leq i \leq n} \max_{1 \leq j \leq h } \E \left(y^{(M)}_{ij}\right)^2 < \infty$. Then, by setting the above equation to $\gamma$, we get $u \leq \sqrt{2 \bar{\sigma}^2} \sqrt{\log(nh}/\gamma) = C l_n^{1/2}$, where $C=\sqrt{2 \bar{\sigma}^2}$. Therefore, $u_x(\gamma) \lesssim \mathcal{D}_n g^{-1}( n/\gamma)$ and $u_y(\gamma) \lesssim l_n^{1/2}$. Then, by the assumption $n^{3/8} M^{-1/2} l_n^{-5/8} \geq C_1 \max\{ \mathcal{D}_n h^{-1}(n/\gamma) ,l_n^{1/2} \}$, we can choose $u \asymp n^{3/8} M^{-1/2} l_n^{-5/8}$, which leads to
	\begin{align} 
	 &\pr\left(\max_{1\leq i \leq n} \max_{1 \leq j \leq h } |x_{ij}^{(M)}| \leq u \right) \geq 1-\gamma,\nonumber \\
	  &\pr\left(\max_{1\leq i \leq n} \max_{1 \leq j \leq h } |y_{ij}^{(M)}| \leq u \right) \geq 1-\gamma.	\label{eq:con1}
	\end{align}
	Next we will quantify $\varphi^{(M)}(M_x,M_y)$ and $\phi^{(M)}(M_x,M_y)$. From Lemma \ref{lemma:Bounds}, we have $\phi^{(M)}(M_x,M_y) \le C'(1/M_x + 1/M_y^2)$, $\varphi^{(M)}(M_y) \le C''(1/M^{5/6}_x + \sqrt{N}/M_x^3 + 1/M_y^2) $ and  $\bar{m}^3_{x,3} + \bar{m}^3_{y,3} < \infty$. (Observe that $\E(|y_{ij}^{(M)}|^3)\le\E(|x_{ij}|^3)$ and $\max_{i,j}\E|x_{ij}|^3<\infty$ by Assumption (A1). Hence $\bar{m}^3_{x,3} + \bar{m}^3_{y,3} < \infty$.)  We can set 
	$$\psi \asymp n^{1/8} M^{-1/2} l^{-3/8}_n, \quad M_x = M_y  = u \asymp n^{3/8} M^{-1/2} l _n^{-5/8}.$$
	Let $\beta \asymp \sqrt{n}/(uM)$, which implies $2\sqrt{5}\beta(6M+1) M_{xy} / \sqrt{n}\asymp1$. 
	Thus, the conditions in Proposition \ref{prop:A1} are satisfied, and we have finished the proof.

\end{proof}

	\begin{corollary}\label{coro:abs}
	Let $x_i$  and $y_i$ be random vectors satisfying the conditions in Theorem \ref{thmA1}. Then, for $q\geq2$ and 
	\[ 
	T_x := \max_{1\leq j \leq h} \Big| \sum_{i =1 }^n x_{ij} \Big|/\sqrt{n} \quad \text{and} \quad T_y := \max_{1\leq j \leq h} \Big| \sum_{i =1 }^n y_{ij} \Big|/\sqrt{n}\,,
	\]
	we get the following Gaussian approximation bound
	\[\sup_{|x| \gtrsim d_{n,2h} }| \pr( T_x \leq x ) - \pr(T_y \leq x) |\lesssim G(n,2h),
	\]
	where $d_{n,2h}$ is defined in \eqref{Definition:dnh} and $G(n,2h)$ is defined in \eqref{Definition:Gnh}.
\end{corollary}
\begin{proof}	
	Note that for $X \in \mathbb{R}^h$ and let $\tilde{X} := \begin{bmatrix}X\\ -X\end{bmatrix} \in \mathbb{R}^{2h} $ \begin{align*}
	\pr( \max_{1\leq j \leq h} |X_j| \leq x) &= \pr( |X_j| \leq x, \text{for all } j) \\
	&= \pr( X_j  \leq x, \text{for all } j \text{ and } -X_j \leq x, \text{ for all } j )\\
	&= 	\pr\left( \max_{1\leq j \leq 2h} \tilde{X}_j \leq x\right)\,.
	\end{align*}
	The result follows form Theorem \ref{thmA1}.
\end{proof}
The following proposition establishes a Gaussian approximation result for sums of high dimensional complex-valued non-stationary time series under the classic norm of complex numbers without assuming that the variances of the normalized sums have a positive lower bound.  Observe that the results of Theorem \ref{thmA1} are for Gaussian approximations on hypercubes of a high-dimensional Euclidean space. Suppose we take $|z|=\sqrt{\Re(z)^2+\Im(z)^2}$ to be the norm for a complex number $z$, the region $\{\max_{1\le k\le h}|z_k|\le x\}$ for an $h$-dimensional complex vector $\vec z=(z_1,z_2,\cdots,z_h)^\top$ is not a hypercube when $\vec z$ is viewed as a vector on ${\mathbb R}^{2h}$. Hence the results of Theorem \ref{thmA1} cannot be used directly here. To tackle this problem, we adopt the idea of simple convex set approximation used in \cite{chernozhukov2017central}. In particular, we shall approximate a circle on the plane by regular convex polygons from both inside and outside.   

\begin{proposition}\label{prop:complex_approximation}
Let $\{z_i=\Re(z_{i})+\sqrt{-1}\Im(z_{i})\}$ be a centered $h$-dimensional complex-valued time series, where $h>1$, such that the $2h$-dimensional real-valued time series $\{\tilde{z}_i=(\Re(z_{i})^\top,\,\Im(z_{i})^{\top})^\top\}$ satisfies the assumptions of Theorem \ref{thmA1}. Take the centered $h$-dimensional complex Gaussian time series $\{g_i\}$ that has the same covariance and pseudo-covariance structures as those of $\{z_i\}$. Define 
\[ 
\tilde T_z = \max_{1\leq j \leq h} \Big| \sum_{i =1 }^n z_{ij} \Big|/\sqrt{n} \quad \text{and} \quad \tilde T_g = \max_{1\leq j \leq h} \Big| \sum_{i =1 }^n g_{ij} \Big|/\sqrt{n}.
\]
Then, we have 
\begin{align}\label{eq:gbound}
&\sup_{x \gtrsim d^*_{n,h} }| \pr( \tilde{T}_z \leq x ) - \pr(\tilde T_g \leq x) |\\
\lesssim &\,G(n,2nh)+\frac{h}{n^{3q/2}}+n^{-1/2}\log^{1+\delta}(nh):=G^*(n,h)\,,\nonumber
\end{align}
where $d^*_{n,h}=d_{n,2nh}(1+\pi^2/(4n^2))$ and $d_{n,2nh}$ and $G(n,2nh)$ are defined in \eqref{Definition:dnh} and \eqref{Definition:Gnh}, respectively.
\end{proposition}
\begin{proof}
We treat $\{\tilde z_i\}$ as a centered $2h$ dimensional non-stationary time series. Observe that, by the assumption that $\max_{1\le j\le 2h}\Theta_{0,j,q}<\infty$ for $\{\tilde z_i\}$, we have that 
$$
\Big\|\sum_{i=1}^n \Re (z_{ik})\Big\|_q\le C\ \ \mbox{and} \ \ \Big\|\sum_{i=1}^n \Im (z_{ik})\Big\|_q\le C
$$
uniformly for all $k=1,2,\cdots,h$ for some finite constant $C$. Therefore, 
$$
\Big\|\max_{1\le k\le h}\Big|\sum_{i=1}^n z_{ik}\Big|\Big\|_q\le Ch^{1/q}
$$
by a simple maximum inequality. 
Hence, if $x\ge n^{3/2}$,  by Markov's inequality, we have
\begin{eqnarray*}
\pr\left(\max_{1\le k\le h}\Big|\sum_{i=1}^n z_{ik}\Big|> x\right)\lesssim \frac{h}{n^{3q/2}}.
\end{eqnarray*}
Similarly, $\pr(\max_{1\le k\le h}|\sum_{i=1}^n g_{ik}|> x)\lesssim \frac{h}{n^{3q/2}}.$
Therefore, we obtain that
\begin{eqnarray}\label{eq:91}
\left|\pr\left(\max_{1\le k\le h}\Big|\sum_{i=1}^n z_{ik}\Big|\le x\right)-\pr\left(\max_{1\le k\le h}\Big|\sum_{i=1}^n g_{ik}\Big|\le x\right)\right|\lesssim \frac{h}{n^{3q/2}}
\end{eqnarray}
if $x\ge n^{3/2}$. Now if $d^*_{n,h}\le x<n^{3/2}$, we apply the {\em regular polygon approximation} technique. Define
$$
z_{ijl}:=\Re (z_{ij})\cos(\pi l/n)+\Im (z_{ij})\sin(\pi l/n), \quad j=1,2,\cdots, h, l=0,1,\cdots, n-1.
$$
Observe that $\cap_{l=1}^{n-1}\{|z_{ijl}|\le x\cos(\pi/(2n))\}$ is a subset of $\{|z_{ij}|\le x\}$, and $\{|z_{ij}|\le x\}$ is a subset of $\cap_{l=1}^{n-1}\{|z_{ijl}|\le x\}$. Therefore,
\begin{align}\label{eq:92}
&\pr\Big(\max_{1\le j\le h,0\le l\le n-1}|z_{ijl}|
\le x\cos(\pi/(2n))\Big)\le \pr(\tilde T_z\le x)\\
\le \,& \pr\Big(\max_{1\le j\le h,0\le l\le n-1}|z_{ijl}|\le x\Big). \nonumber
\end{align}
The same result holds with $z$ in \eqref{eq:92} replaced by $g$, where $g_{ijl}$ are defined analogously. By the inequality that $1/\cos(x)\le 1+x^2$ for $x\in[0,\pi/4]$, we have that $x\cos(\pi/(2n))\ge d_{n,2nh}$ for $n\ge 2$. Hence, by Corollary \ref{coro:abs},
\begin{align}\label{eq:93}
&\sup_{x\gtrsim d^*_{n,h}}\left|\pr\Big(\max_{1\le j\le h,0\le l\le n-1}|z_{ijl}|\le x\cos(\pi/(2n))\Big)-\pr\Big(\max_{1\le j\le h,0\le l\le n-1}|g_{ijl}|\le x\cos(\pi/(2n))\Big)\right|\nonumber\\
\lesssim &\,G(n,2nh)\,,
\end{align}
and similarly,
\begin{align}\label{eq:94}
&\sup_{x\gtrsim d^*_{n,h}}\left|\pr\Big(\max_{1\le j\le h,0\le l\le n-1}|z_{ijl}|\le x\Big)-\pr\Big(\max_{1\le j\le h,0\le l\le n-1}|g_{ijl}|\le x\Big)\right|\\
\lesssim \,&G(n,2nh).\nonumber
\end{align}
Note that $|x-x\cos(\pi/(2n))|\lesssim xn^{-2}\lesssim n^{-1/2}$ since $x\in [d^*_{n,h},n^{3/2}]$. Therefore by Corollary \ref{cor1}, we have that
\begin{align}
&\sup_{n^{3/2}>x\gtrsim d^*_{n,h}}\left|\pr\Big(\max_{1\le j\le h,0\le l\le n-1}|g_{ijl}|\le x\Big)-\pr\Big(\max_{1\le j\le h,0\le l\le n-1}|g_{ijl}|\le x\cos(\pi/2n)\Big)\right|\nonumber\\
\lesssim\,& n^{-1/2}\log^{1+\delta}(nh)+\frac{1}{nh\log^{1/2}(nh)}.\label{eq:95}
\end{align}
The proposition follows by \eqref{eq:91} to \eqref{eq:95} as we notice that $\frac{1}{nh\log^{1/2}(nh)}$ is dominated by $G(n,2nh)$.
\end{proof}

\section{Gaussian comparison without variance lower bounds}\label{sec:B}
We will be modifying the comparison of distribution theorem in \cite{chernozhukov2015} by dropping the assumption of lower variance bound in the anti-concentration inequality. We will also extend such results to complex-valued high-dimensional Gaussian vectors. The results established in this section are crucial for the theoretical investigation of the multiplier bootstrap proposed in this paper.  
	\begin{theorem}\label{thm:B1}
		(Comparison of distribution without variance lower bound). Let $X =$ $(X_1,\dots,X_h)^\top$ and $Y = (Y_1,\dots,Y_h)^\top$ be centered Gaussian random vectors in $\mathbb{R}^h$. Let $\Sigma_X=(\sigma^X_{i,j})_{i,j=1}^h$ with $\sigma^X_{i,j}=\cov(X_i,X_j)$, and define $\Sigma_Y$ and $\sigma^Y_{i,j}$ similarly. Suppose that $h\geq 2$ and there exists a finite constant $c$ such that $c>\sigma^Y_{jj} > 0$ for all $1\leq j \leq h$. Define
		\[
		\Delta = \max_{i,j \leq h} \big|\sigma^X_{i,j}-\sigma^Y_{i,j}\big| \,.
		\]
		Then,
		\begin{align*}
		&\sup_{|x| > d_h } \left| \pr\Big( \max_{1 \leq j \leq h } X_j \leq x \Big) - \pr\Big(\max_{1 \leq j \leq h } Y_j \leq x\Big)\right| \\
		 =&\,O\left(  \Delta^{1/3}  \log(h)^{1 + 4\delta/3} + \frac{1}{(\log h)^{1/2}h }\right), 
		\end{align*}
		where  $d_h =  2\Delta^{1/3}\log(h)^{\delta/3}+  2\log(h)^{-\delta}$. 
	\end{theorem}

\begin{proof}
	Let
	$$
	g_0(x) = \begin{cases}
	\ 0, & x \geq 1,\\
	\ 30 \int_{x}^1 s^2 (1-s)^2 ds , & 0 < x < 1,\\
	\ 1, & x \leq 0.
	\end{cases}
	$$
	and pick $g(s) = g_0 (\psi (s - t -e_\beta ))$ with $e_\beta = \beta^{-1} \log h$, where $\psi,\beta>0$ and $t$ will be picked later. 
	Denote 
		$G_k :=\sup_{x\in\mathbb{R}}\partial^k g(x)/\partial x^k$, 
		$k=1,2,\cdots$. We have
	 $G_0 \lesssim 1$, $G_1 \lesssim \psi$, $G_2 \lesssim \psi^2$ and $G_3 \lesssim \psi^3$. Moreover, the function also satisfies 
	$$
	\mathbb{I}(x \leq t + e_\beta) \leq g(x) \leq \mathbb{I}(x\leq t + e_\beta + \psi^{-1}),\,\, \forall x \in \mathbb{R}.
	$$
	
	Now assume $|t| >  e_{\beta} + \psi^{-1} + \log(h)^{-\delta} $,
	\begin{align*}
	P(\max_{1\leq j \leq h } X_j \leq t ) &\leq \pr( F_\beta (X) \leq t + e_\beta  )\\ &\leq \E [g (F_\beta(X)) ] \\
	& \leq \E[g(F_\beta(Y))] + C_0(\psi^2 +\beta \psi) \Delta \quad\\
	& \leq \pr (F_\beta(Y) \leq t + e_\beta + \psi^{-1}) + C_0(\psi^2 +\beta \psi) \Delta \quad\\
	& \leq \pr \left(\max_{1 \leq j \leq h} Y_j \leq t +e_\beta + \psi^{-1} \right) + C_0(\psi^2 +\beta \psi) \Delta\,, 
	\end{align*}
	for some absolute constant $C_0>0$, where we utilized Theorem 1 of \cite{chernozhukov2015} in the third inequality above. Then, by Lemma \ref{lemma1} and assume $|x| >  e_{\beta} + \psi^{-1}  + 2(\log h)^{-\delta} $,
	\begin{align*}
	&\pr \left(\max_{1 \leq j \leq h} Y_j \leq t +e_\beta + \psi^{-1} \right) - \pr \left(\max_{1 \leq j \leq h} Y_j \leq t \right) \\
	\leq &\,(e_\beta + \psi^{-1}) (\log{h})^{1+\delta} + h^{-1} (\log h)^{-1/2} .
	\end{align*}
	Thus, we conclude that when $|t| >  e_{\beta} + \psi^{-1}  + 2(\log h)^{-\delta} $,
	\begin{align*}
	&\pr \left(\max_{1\leq j \leq h } X_j \leq t \right) - \pr\left(\max_{1 \leq j \leq h} Y_j \leq t\right) \\
	&\qquad\qquad \leq C_0(\psi^2 +\beta \psi) \Delta  +  (e_\beta + \psi^{-1}) (\log{h})^{1+\delta} + h^{-1} (\log h)^{-1/2}.
	\end{align*}
	The opposite direction can be proven similarly by noting that
	$$
	\pr\left( \max_{1 \leq j \leq h } X_j \leq t\right) \geq \pr\left( \max_{1 \leq j \leq h }Y_j \leq t - e_\beta - \psi^{-1} \right) - C(\psi^2 +\beta \psi) \Delta .
	$$
	and
	\begin{align*}
	&\pr\left( \max_{1 \leq j \leq h }Y_j \leq t - e_\beta - \psi^{-1} \right) - \pr\left( \max_{1 \leq j \leq h } Y_j \leq t\right)\\
	  \geq&\,  -  (e_\beta + \psi^{-1}) (\log{h})^{1+\delta} - h^{-1} (\log h)^{-1/2}
	\end{align*}
	
	
  Finally, pick $\beta =\psi \log(h)$ and $\psi^{-1} = \Delta^{1/3} \log(h)^{\delta/3}$ (note that $e_\beta = \beta^{-1} \log(h) = \psi^{-1}).$ Then, for some constant $C>0$, the following inequality holds 
  \begin{align}
  C_0(\psi^2 +\beta \psi) \Delta  +  (e_\beta + \psi^{-1}) (\log{h})^{1+\delta} &\, = C_0 \psi^2 (1 + \log(h) ) \Delta + 2\psi^{-1} \log(h)^{1+\delta} \nonumber\\
  &\,\leq C \Delta^{1/3} \log(h)^{1+4\delta/3}.
  \end{align}
	
	The lower bound can be derived similar to the previous steps.
\end{proof}
The following proposition establishes a comparison result for complex-valued Gaussian random vectors.
\begin{proposition}\label{prop:complex_comparison}
Let $Z=(z_1,\cdots,z_h)^\top$ and $W=(w_1,\cdots,w_h)^\top$ be centered complex-valued Gaussian random vectors in $\mathbb{C}^h$. Write $z_{i}=\Re (z_{i})+\sqrt{-1}\Im(z_{i})$ and $w_{i}=\Re(w_{i})+\sqrt{-1}\Im(w_{i})$, $i=1,2,\cdots, h$. Let 
\[
\tilde Z=(\Re(z_{1}),\Im(z_{1}),\cdots,\Re(z_{h}),\Im(z_{h}))^\top\in\mathbb{R}^{2h},
\] 
\[
\tilde W=(\Re(w_{1}),\Im(w_{1}),\cdots,\Re(w_{h}),\Im(w_{h}))^\top\in\mathbb{R}^{2h}
\] 
and 
\[
\tilde\Delta=\max_{1\le i,j\le 2h}|\cov(\tilde Z)_{i,j}-\cov(\tilde W)_{i,j}|. 
\]
Suppose that $c\ge  \max(\var(\Re(z_{i})),\var(\Re(w_{i})),\var(\Im(z_{i})),\var(\Im(w_{i})))>0$, $i=1,2,\cdots, h$, for some positive and finite constant $c$. Then we have
\begin{align*}
		\sup_{x \gtrsim d^*_h } &\left| \pr\Big( \max_{1 \leq j \leq h } |z_j| \leq x \Big) - \pr\Big(\max_{1 \leq j \leq h } |w_j| \leq x\Big)\right|  \lesssim \tilde{\Delta}^{1/3}\log^{1+2\delta} (h^2)+h^{-1}\log^{1+\delta}(h^2) , 
		\end{align*}
where $d^*_h=C(\tilde\Delta^{1/3}\log(h^2)^{\delta/3}+  \log(h^2)^{-\delta})(1+\pi^2/(4h^2))$ for some absolute constant $C>0$.		
\end{proposition}
\begin{proof}
The proof of this Proposition is similar to that of Proposition \ref{prop:complex_approximation}. We shall only outline the proofs here.    First of all, if $x\ge h$, then we have $\|\Re(z_{i})\|_{\Psi_2}\le C$, $\|\Im(z_{i})\|_{\Psi_2}\le C$, $\|\Re(w_{i})\|_{\Psi_2}\le C$ and $\|\Im(w_{i})\|_{\Psi_2}\le C$ for some finite constant $C$, where $\|\cdot\|_{\Psi_2}:=\inf\{c>0 : \E\Psi_2(|\cdot|/c)\le 1\}$ with $\Psi_2(x)=e^{x^2}-1$ is the Orcliz norm. A simple maximum inequality of the Orcliz norm yields that 
\begin{eqnarray}\label{eq:p31}
\pr\left(\max_{1\le k\le h}|z_{k}|> x\right)\lesssim h\exp(-h^2).
\end{eqnarray}
Similarly, 
\[
\pr\left(\max_{1\le k\le h}|w_{k}|> x\right)\lesssim h\exp(-h^2)\,.
\]
If $d^*_h<x<h$, then by the same regular-polygon-approximation technique used in Proposition \ref{prop:complex_approximation}, define
$\bar z_{i,l}=\Re(z_{i})\cos(\pi l/h)+\Im(z_{i})\sin(\pi l/h)$, where $i,l=1,2,\cdots,h$. Define $\bar w_{i,l}$ similarly. Then we have
\begin{align}\label{eq:p32}
&\pr\left(\max_{1\le i, l\le h}|\bar z_{i,l}|\le x\cos(\pi/(2h))\right)\\
\le \,&\pr\left(\max_{1\le i\le h}|z_i|\le x\right)\le \pr\left(\max_{1\le i,l\le h}|\bar z_{i,l}|\le x\right)\,. \nonumber
\end{align}
The same inequality holds with $z$ in \eqref{eq:p32} replaced by $w$. Now Theorem \ref{thm:B1} implies that
\begin{align}
&\sup_{x>d^*_h}\left|\pr\left(\max_{1\le i, l\le h}|\bar z_{i,l}|\le x\cos(\pi/(2h))\right)-\pr\left(\max_{1\le i, l\le h}|\bar w_{i,l}|\le x\cos(\pi/(2h))\right)\right| \nonumber\\
\lesssim \,& A(\tilde\Delta,h),\label{eq:p33}
\end{align}
where $A(\tilde\Delta,h)=\tilde{\Delta}^{1/3}\log^{1+2\delta} (2h^2)+(\log 2h^2)^{-1/2}(2h^2)^{-1}$. Similarly,
\begin{eqnarray}\label{eq:p34}
\sup_{x>d^*_h}\left|\pr\left(\max_{1\le i, l\le h}|\bar z_{i,l}|\le x\right)-\pr\left(\max_{1\le i, l\le h}|\bar w_{i,l}|\le x\right)\right|\lesssim A(\tilde\Delta,h).
\end{eqnarray}
Note that $|x-x\cos(\pi/(2h))|\lesssim |x|h^{-2}\lesssim h^{-1}$ if $|x|\le h$. Therefore by Corollary \ref{cor1}, we have that
\begin{align}\label{eq:p35}
\sup_{h>x>d^*_{h}}&\left|\pr\left(\max_{1\le i,l\le h}|z_{i,l}|\le x\right)-\pr\left(\max_{1\le i,l\le h}|z_{i,l}|\le x\cos(\pi/2h)\right)\right|\cr
&\,\qquad\qquad\lesssim h^{-1}\log^{1+\delta}(h^2)+\frac{1}{h^2\log^{1/2}(h^2)}.
\end{align}
Hence the proposition follows by noting that $h^{-1}\log^{1+\delta}(h^2)$ dominates both $\frac{1}{h^2\log^{1/2}(h^2)}$ and $h\exp(-h^2)$.
\end{proof}


\section{Proof of the theoretical results of the paper}\label{Section:AppenC}
In the sequel, we shall omit the subscript $n$ in $X_{i,n}$ and $\epsilon_{i,n}$ in the time series model \eqref{eq:model} to simplify notation. Meanwhile, the symbol $C$ denotes a generic positive and finite constant which may vary from place to place. Recall Assumptions \ref{assumption1}-\ref{assumption9}. We have an immediate comment. 
\begin{remark}
For Assumption \ref{assumption1}, observe that under the assumption that 
\[
\max_{1\le i\le n}\E(g(|X_i|))\le c\mbox{ or }\max_{1\le i\le n}\E(\exp(|X_i|))\le c
\] 
for some finite constant $c$, $\mathcal{D}_n$ can be chosen as a finite constant that does not depend on $n$.
\end{remark}

We would summarize how to implement the Gaussian approximation results in Section \ref{sec:A} to the DPPT.  Let $X_i$ be a centered univariate $PLS(r)$ noise, $i=1,2,\cdots,n$. Define 
	\begin{align*}
	\tilde{x} := [  \tilde{x}_1 , ...,\tilde{x}_n]  = 
	\left( \begin{array}{ccccc}
	X_1 & X_2 & \cdots & X_{n-1} & X_n \\
	X_1 & X_2 & \cdots & X_{n-1} & 0\\ 
	\vdots& \vdots  & \ddots    &       \vdots      & \vdots \\ 
	X_1 &0      & \cdots &   0       & 0\\ 
	\end{array} \right).
	\end{align*}
	Now, set $e_j = [ \exp(\sqrt{-1}  j w_1), \dots , \exp(\sqrt{-1} j w_p)]^\top\in \mathbb{C}^p$ for $j\in \mathbb{N}$, where $w_i \in W$. Then we define the vectorized data as
	\begin{align}\label{eq:defx}
	x_i := \tilde{x}_i \otimes e_i =\begin{pmatrix}\tilde{x}_i  e_{i,1}\\\vdots\\\tilde{x}_i  e_{i,p}\end{pmatrix}\in \mathbb{C}^{np}\,,
	\end{align}
	where $i=1,\ldots,n$.
Here, $x_i$ are complex random vectors satisfying 
\begin{equation}
\max_{1\leq j \leq pn} \Big| \sum_{i =1 }^n x_{ij} \Big|/\sqrt{n} = F(W)\,,
\end{equation}
where
\[
F(W)=\max_{l=1,\ldots,p}\max_{k=1,\ldots,n}\Big|\sum_{i=1}^kX_ie^{\sqrt{-1}i\omega_l}\Big|/\sqrt{n}\,.
\]

Thus, we can write $F(W)$ in a form matching Proposition \ref{prop:complex_approximation}. If $X_i$ satisfies Assumptions 1-3, it implies that $\{x_{i}\}$ satisfies conditions of Proposition \ref{prop:complex_approximation}. By Proposition \ref{prop:complex_approximation}, we get that, for $(np)$-dim centered complex Gaussian random vectors $\{y_{i}\}$ having the same covariance and pseudo-covariance structures of $\{x_i\}$, \begin{align}\label{eq2} 
\sup_{|x| > d^*_{n,pn} } \left| \pr( F(W) \leq x ) - \pr \left(  \max_{1\leq j \leq pn} \Big| \sum_{i =1 }^n y_{ij} \Big|/\sqrt{n}   \leq x \right) \right| &\lesssim G^*(n,pn).
\end{align}
We now briefly discuss when the magnitudes of $d^*_{n,pn}$ as well as the right hand side of \eqref{eq2} will converge to 0. Under Assumption \ref{assumption3}, we have that $\max_{1\le i\le n}|X_i|^4<\infty$. Hence we can choose $q=4$, $g(x)=x^4$ and let $\mathcal{D}_n$ be a constant.   
Observe that Assumption \ref{assumption3} implies that 
\[ 
\Theta_{M,j,q} \le \sum_{k = M}^{\infty} C(k+1)^{-d} \le C\int_{M}^\infty y^{-d} dy \le CM^{-d+1} 
\] 
uniformly in $j$. Choose $M=n^{1/4}\log^{-6}(n)$ and $\gamma=\log^{-1}(n)$. We have that Assumption \ref{assumption2} is satisfied for the above choices of $M$ and $\gamma$. Note that for our test $p=O(n^{3/2}\log(n))$. We obtain that the right hand side of \eqref{eq2} is of the order $O(\log^{-1}(n))$ which converges to 0. Meanwhile, simple calculations yield that $d^*_{n,pn}=O(\log^{-10/7}(n)+\log^{-\delta}(n))$, which goes to 0. As shown in Lemma \ref{lem6} below, the critical values of the DPPT under the null hypothesis of no oscillation is of the order $O(\sqrt{\log n})$ which is asymptotically larger than $d^*_{n,pn}$. 
Hence, under Assumptions \ref{assumption1}-\ref{assumption3}, the Gaussian approximation result established in Proposition \ref{prop:complex_approximation} is sufficient for the DPPT.  

Observe that a faster convergence of the Gaussian approximation error can be obtained under stronger moment and dependence assumptions. For instance, in the best scenario where \eqref{eq12} holds and the dependence of $\{X_i\}$ is of exponential decay, one can choose $M=O(\log n)$ and $\gamma=O(1/n)$. In this case the convergence rate of the Gaussian approximation is of the order $O(n^{-1/8}\log^{29/14}(n))$.

	\subsection{First Stage proof}

	\subsubsection{Consistency}

 Denote by $W$ the set of candidate frequencies and denote $p = |W|$. 
 Define a random vector $\Theta(W)$ by
 \[
 \Theta(W) := [\Theta(w_1)^\top,\ldots, \Theta(w_p)^\top]^\top/\sqrt{n} \in \mathbb{R}^{2np}
 \] 
 and the vector $\Theta(w) \in \R^{2n} $ is defined coordinate-wise 
 \[ 
 \Theta_k(w) :=  \begin{cases} \sum_{j= 1}^k\cos(  w j ) X_j \  \text{    for }\  k \leq n  \\ 
 \sum_{j= 1}^{k -n} \sin(  w j ) X_j \  \text{  for } \ n < k \leq 2n \,. \end{cases} 
 \]
Note we can write
\[
F(W) = \max_{1 \leq i \leq n,1\le j\le p } |\Theta_i(\omega_j)+\sqrt{-1}\Theta_{i+n}(\omega_j)|/\sqrt{n}\,.
\] 
Moreover, if we put 
\[
\bar\Theta_k(W):=[\sum_{l=1}^ke^{\sqrt{-1} \omega_1l}X_l,\ldots,\sum_{l=1}^ke^{\sqrt{-1} \omega_pl}X_l]^\top\in \mathbb{C}^p, 
\]
we have
\[
F(W) = \max_{1 \leq k \leq n} \|\bar\Theta_k(W)\|_\infty/\sqrt{n}\,.
\]

Recall that for the fixed integer bandwidth $m$, we defined 
\[S_{j,m}(w) = \sum_{i = j}^{j+m -1 } \sin(i w) X_i\] and \[C_{j,m}(w) = \sum_{i = j}^{j+m -1 } \cos( i w) X_i.\]

	For an i.i.d. standard Gaussian random variables $G_i$, define $S(w) \in \R^{2n} $ by
	\[
	S_k(w) := \begin{cases}
	S_{1,m}(w)G_1 &\text{ if } k\leq m\\
	\sum_{j= 1}^{k - m+1} S_{j,m}(w) G_j &\text{ if }  m < k\leq n \\
	C_{1,m}(w)G_1 &\text{ if } n <  k\leq n+m\\
	\sum_{j= 1}^{k - n - m+1} C_{j,m}(w) G_j &\text{ if }  n+m < k\leq 2n, \\
	\end{cases}
	\]
	where $S_k(w)$ denotes the $k$-th coordinate of $S(w)$. Define a $n$-dim complex vector $SZ(w)$ such that $SZ_k(w)=\sqrt{-1}S_k(w)+S_{k+n}(w)$, $k=1,2,\cdots,n$. Let 	
	\begin{align}\label{Definition SW in the proof of Gaussian comparison}
	S(W) := [S^\top(w_1),...,S^\top(w_p)]^\top/\sqrt{m(n-m)}\in \mathbb{R}^{2np}.
	\end{align}
	and 
	\begin{align}\label{Definition SZW in the proof of Gaussian comparison}
	SZ(W) := [SZ^\top(w_1),...,SZ^\top(w_p)]^\top/\sqrt{m(n-m)}\in \mathbb{C}^{np}.
	\end{align}
	Then we have
	\begin{align*}
	\tilde{F}(W)&\, = \max_{1\le j\le p}\max_{1\le i\le n} |\sqrt{-1}S_i(w_j)+S_{i+n}(w_j)|/\sqrt{m(n-m)}\\
	&\,=\max_{1\le i\le n} \|SZ_i(W)\|_\infty/\sqrt{m(n-m)}.
	\end{align*}

\begin{lemma}\label{lemma3}
	Assume the same conditions and notation as in Theorem \ref{thm:B1} hold for $h$-dim random vectors $X$ and $Y$. Also, assume that $ \max_{1\leq j \leq h } X_j = \max_{1\leq j \leq h } \epsilon_j + O(a_h)  $ and  $ \max_{1\leq j \leq h } Y_j = \max_{1\leq j \leq h } \xi_j + O(b_h) $ for another pair of $h$-dim Gaussian random vectors $\epsilon$ and $\xi$ and two sequences $\{a_h\}$ and $\{b_h\}$ that converge to $0$ when $h\to \infty$. Then 
	\begin{align*}  
	& \sup_{|x| \gtrsim d_h +a_h  + b_h} \left| \pr\left( \max_{1 \leq j \leq h } X_j \leq x \right) - \pr\left(\max_{1 \leq j \leq h } Y_j \leq x\right)\right|\\ 
 	=&  \sup_{|x| \gtrsim d_h + a_h  + b_h} \left| \pr\left( \max_{1 \leq j \leq h } \epsilon_j \leq x \right) - \pr\left(\max_{1 \leq j \leq h } \xi_j \leq x\right)\right| \\ 
	& \qquad + O(a_h +b_h )\log(h)^{1+\delta} +  \frac{1}{2 \sqrt{2\pi}(\log h)^{1/2} h} . 
	\end{align*}
\end{lemma}

\begin{proof}
	We have
	\begin{align*}
	&\sup_{|x| \gtrsim d_h + a_h  +b_h} \left| \pr\left( \max_{1 \leq j \leq h } X_j \leq x \right) - \pr\left(\max_{1 \leq j \leq h } Y_j \leq x\right)\right|\\
	=&\,\sup_{|x| \gtrsim d_h + a_h  +b_h} \left| \pr\left( \max_{1 \leq j \leq h } \epsilon_j \leq x+O(a_h) \right) - \pr\left(\max_{1 \leq j \leq h } \xi_j \leq x+O(b_h)\right)\right|
	\end{align*}
	and hence by the triangular inequality, 
	\begin{align*}
	&\left|\sup_{|x| \gtrsim d_h + a_h +b_h} \left| \pr\left( \max_{1 \leq j \leq h } \epsilon_j \leq x+O(a_h) \right) - \pr\left(\max_{1 \leq j \leq h } \xi_j \leq x+O(b_h)\right)\right|-\right.\\
	&\qquad
	\left.\sup_{|x| \gtrsim d_h + a_h  + b_h} \left| \pr\left( \max_{1 \leq j \leq h } \epsilon_j \leq x\right) - \pr\left(\max_{1 \leq j \leq h } \xi_j \leq x\right)\right|\right|\\
	\leq&\,\sup_{|x| \gtrsim d_h + a_h  + b_h} \left| \pr\left( \max_{1 \leq j \leq h } \epsilon_j \leq x+O(a_h) \right) - \pr\left(\max_{1 \leq j \leq h } \epsilon_j \leq x\right)\right|\\
	&+\sup_{|x| \gtrsim d_h + a_h  + b_h} \left| \pr\left( \max_{1 \leq j \leq h } \xi_j \leq x+O(b_h) \right) - \pr\left(\max_{1 \leq j \leq h } \xi_j \leq x\right)\right|\,.
	\end{align*}
	Then the result follows directly from Corollary \ref{cor1}. 
\end{proof}

\begin{lemma}\label{lemma4}
	Let $\{x_i \in \R^h\} $ be a centered PLS time series. Suppose that 
	$$
	\max_{1\le j\le h,\,1\le i\le n}\|x_{i,j}\|_{q}<\infty
	$$ 
	for some $q\ge 2$, where $x_{i,j}$ is the $j$-th component of $x_i$. Further assume that $\max_{1\le j\le h}\delta_{j,q}(k)=O((k+1)^{-d})$ for some $d>1$, where $\delta_{j,q}(k)$ is the physical dependence measure of the $j$-th component process of $\{x_i\}$. For $1\leq m<n$, define $S_{k,m} = \sum_{i = k}^{k+m -1} x_i $ and 
	\[ (\Lambda_{r,s})_{k,l} = \frac{1}{m(n-m)}\sum_{j=r}^s (S_{j,m})_k (S_{j,m})_l, 
	\]
	where $1\leq r\leq s\leq n-m$.
    Then for $ 1 \leq k, l \leq h$,
	\[ 
	\max_{1\leq r\leq s\leq n-m+1}\left\|  |(\Lambda_{r,s})_{kl} - \E (\Lambda_{r,s})_{kl} |  \right\|_{q'} = O(\sqrt{m/n}),
	\] 
	where $q'=q/2$.
\end{lemma}

\begin{proof}
Observe that $S_{j,m}$ can be written in a physical representation $S_{j,m}=R_{j,n}(\FF_{j+m})$, where $R_{j,n}$ is some filter function. For any $t\le i+m$, let $S^{(t)}_{j,m}=R_{j,n}(\FF^{(t)}_{j+m})$, where $\FF^{(t)}_{j+m}$ is a coupled version of $\FF_{j+m}$ such that the innovation $e_t$ is replaced by an i.i.d. copy. Let $(S^{(t)}_{j,m})_l$ be the $l$-th coordinate of the vector $S^{(t)}_{j,m}$. Note that \begin{align*}
	&\| (S_{i,m})_l (S_{i,m})_k -  (S^{(t)}_{i,m})_l (S^{(t)}_{i,m})_k \|_{q'} \\
	 =\,& \| ((S_{i,m})_l -  (S^{(t)}_{i,m})_l)( (S_{i,m})_k + (S^{(t)}_{i,m})_k ) \\ &
	+ ((S_{i,m})_l +  (S^{(t)}_{i,m})_l)( (S_{i,m})_k - (S^{(t)}_{i,m})_k )  \|_{q'} /2 \\
	\leq \, & ( \| (S_{i,m})_k\|_{q} + \|(S^{(t)}_{i,m})_k \|_{q}  ) \|(S_{i,m})_l -  (S^{(t)}_{i,m})_l \|_{q}/2 \\
	& +  ( \| (S_{i,m})_l\|_{q} + \|(S^{(t)}_{i,m})_l \|_{q}  ) \|(S_{i,m})_k -  (S^{(t)}_{i,m})_k \|_{q}/2.
	\end{align*}
	Then, the result follows by the proof of \cite[Lemma 1]{zhou2013}. 
\end{proof}

\begin{lemma} \label{lemma5}
    Let $\{\epsilon_i\}$ be a centered PLS($r$) time series satisfying Assumptions \ref{assumption1} - \ref{assumption3}.     Then for a bounded sequence $|a_i| \leq C$, where $C >0$, we have 
    \[ 
    \max_{1 \leq k \leq n} \Big| \sum_{i = 1}^k a_i \epsilon_i\Big| = O_p(\sqrt{n \log(n) }). 
    \]
\end{lemma}

\begin{proof}

	Note that if $X_i \sim N(0,\sigma_i^2)$ and $\sigma_i \leq \sigma$, then $\E[\max_{1\leq i\leq n} |X_i| ]   \leq \sigma \sqrt{ \log(2n)} $ by a simple Orcliz norm maximum inequality. Let $\tilde{y}_i$, $i=1,2,\cdots, n$ be a sequence of centered Gaussian random variables which preserve the covariance structure of $\epsilon_i$, $i=1,2,\cdots, n$. We have
	\[ 
	\max_{1 \leq k \leq n}\Big|   \sum_{i = 1}^{k} a_i \tilde{y}_i \Big|  =O_p\left( \sqrt{ \log(2n)}  \max_{1 \leq i \leq n}  \sqrt{ \var\left( \sum_{j = 1}^{i} a_j \tilde{y}_j  \right) }  \right) \,. 
	\]
	Note that we also have
	\begin{align*}
	\var\left( \sum_{i = 1}^{n} a_i \tilde{y}_i  \right) &\leq 4 \sum_{i = 1}^n \var(\tilde{y}_i) + 4\sum_{i \neq j} Cov(\tilde{y}_i,\tilde{y}_j) \\
	&\leq 4n \sigma^2+ 4C\sum_{k = 1}^{n-1} (n - k) k^{-d}   = O(n), 
	\end{align*}
	where we utilized Lemma 6 of \cite{zhou2014} that guarantees $Cov(\tilde{y}_i,\tilde{y}_j)\le C|i-j|^{-d}$ under Assumption 3. 
	It implies that 
	\[
	\max_{1 \leq k \leq n} \Big|\sum_{i = 1}^k a_i \tilde{y}_i\Big| = O_p(\sqrt{n  \log(n)} ) \,. 
	\]
	
	Combining with the Gaussian approximation results, we get 
	\[
	\max_{1 \leq k \leq n} \Big|\sum_{i = 1}^k a_i \epsilon_i\Big| = O_p(\sqrt{n  \log(n)} )\, . 
	\]
	Indeed, by Gaussian approximation results in Theorem \ref{thmA1}, for $y_i$ normally distributed with the same covariance structure as $\epsilon_i$, we have  
	\[ 
	\sup_{t > d_n } \left| \pr\left(  \max_{1 \leq k \leq n} \Big| \sum_{i = 1}^k a_i \epsilon_i/\sqrt{n}\Big| > t  \right) - \pr\left( \max_{1 \leq k \leq n} \Big| \sum_{i = 1}^k a_i \tilde{y}_i/\sqrt{n} \Big| > t \right) \right| \rightarrow 0, 
	\] 
	where $d_n \rightarrow 0$. 
\end{proof}

The next lemma controls the smooth trend of the first stage statistics under the null; that is, $\mu_i = f(i/n)$.

\begin{lemma}
\label{lemma:smooth}
    Suppose that function $f$ is twice differentiable on $[0,1]$ with Lipschitz continuous second derivatives. Then, for $w\in  [cn^{-\theta},\pi]$ for some constant $c>0$ and $\theta\in[0,1/4)$, we have 
    \[  
    \max_{1\le j\le n} \left| \sum_{i = 1 }^{j} f(i/n)  \exp(\im i w  ) \right| = O(n^{1/4+\theta})\,. 
    \]
\end{lemma}

\begin{proof}
We shall only prove the case where $j=n$ as the other cases follow by the similar arguments. Choose an integer $m \asymp n^{3/4}$. First, group $\sum_{i = 1 }^{n} f(i/n)  \exp(w i \im) $ into
\begin{align*}
    & \sum_{j = 0}^{\lfloor n/m\rfloor-1} \sum_{i = jm+1 }^{ (j+1)m} f(i/n)  \exp( \im i w) +\sum_{i=\lfloor\frac{n}{m}\rfloor m+1}^nf(i/n)\exp(\im i w)\,.
    \end{align*}
    Below we only show the control of $\sum_{j = 0}^{\lfloor n/m\rfloor-1} \sum_{i = jm+1 }^{ (j+1)m} f(i/n)  \exp(w i \im)$, since the term $\sum_{i=\lfloor\frac{n}{m}\rfloor m+1}^nf(i/n)\exp(w i \im)$ can be handled in the same way.
    By the Taylor expansion, we have
    \begin{align*}
    =&\, \sum_{j = 0}^{n/m-1} \sum_{i = jm +1}^{ (j+1)m} \sum_{k = 0}^{2} \big[f^{(k)}(jm) (i/n - jm/n)^k / k! + O((i/n - jm/n)^{3}\big]  \exp(w i \im)\,.  
\end{align*}
 The first term in the expansion can be computed as 
 \[ 
 \left|  \sum_{i = jm+1 }^{ (j+1)m} f(jm) \exp(w i \im)\right| =   |f(jm)| \left| \frac{\exp(\im m w ) - 1 }{\exp(\im w ) - 1}  \right|\leq |\pi f(jm)|/|w|.   
 \]
The second order terms can be simplified as  \begin{align*} 
& \left|\sum_{i = jm}^{(j+1)m -1} f'(jm)  (i/n - jm/n) \exp(\im w_l i) \right| \\
\le&\, C\left| \frac{m \exp(\im (m+2) w )-(m+1) \exp(\im (m+1) w ) + \exp(\im w)  }{n(\exp(\im w ) - 1)^2 }  \right|\\
=&\, O\left( \frac{m}{n} \frac{1}{|w|} + \frac{1}{n} \frac{1}{|w|^2}  \right) = O(1).
\end{align*}  
The third term is tedious to compute but the simplified results show that 
\begin{align*} 
& \left|\sum_{i = jm}^{(j+1)m -1} f''(jm)  (i/n - jm/n)^2 \exp(\im w i) \right|\\
\le&\, C\left| \frac{m^2 e^{\im (m+3) w }+(-2m^2-2m+1) e^{\im (m+2) w } +(m+1)^2 e^{\im (m+1) w } - e^{\im  w }- e^{\im 2 w } }{n^2(\exp(\im w ) - 1)^3 }  \right|\\
=&\, O\left( \frac{m^2}{n^2} \frac{1}{|w|^2} + \frac{1}{n^2} \frac{1}{|w|^3}  \right) = O(1).
\end{align*}  
Using a less precise bound for the remaining term, we have
\begin{align*}
      & \left|\sum_{i = jm}^{(j+1)m -1}  (i/n - jm/n)^3 \exp(\im w_l i) \right| \\
      \le &\, \sum_{i = jm}^{(j+1)m -1}  \left|(i/n - jm/n)^3  \right|  =  O\left( m^4/n^3 \right)
  \end{align*}
By combining the previous results, we get
\begin{align*}
    & \sum_{i = 1 }^{n} f(i/n)  \exp(w i \im) 
    =\sum_{j = 0}^{n/m} O(1/|w| + m^4/n^3)\\
     =&\, O(n/(m |w|) + m^3/n^2) = O(n^{1/4+\theta}).
\end{align*}
\end{proof}

Next we shall prove the first main result of this paper. Define 
\begin{eqnarray}\label{eq:thetamu}
\bar\Theta_k(\omega)=\sum_{i = 1 }^{k} X_i  \exp(\im i\omega),\quad \bar\Theta_k^{(\mu)}(\omega)=\sum_{i = 1 }^{k} \mu_i  \exp(\im i \omega),
\end{eqnarray}
and
\begin{eqnarray}\label{eq:thetaepsilon}
\bar\Theta_k^{(\epsilon)}(\omega)=\sum_{i = 1 }^{k} \epsilon_i  \exp(\im i\omega).
\end{eqnarray}
Clearly, we have 
\[
\bar\Theta_k(\omega)=\bar\Theta_k^{(\mu)}(\omega)+\bar\Theta_k^{(\epsilon)}(\omega).
\]

\begin{proof}[Proof of Theorem \ref{thm:stage1boots}]

	First note that  \begin{align*}
	   & \sup_{x} \left| \pr( F(W) \leq x) -  \pr( \tilde{F}(W) \leq x | X ) \right|\\
	   =&\, \sup_{x} \left| \pr( \max_k \|n^{-1/2}\bar\Theta_k(W) \|_\infty \leq x) -  \pr( \|SZ(W)\|_\infty \leq x | X ) \right|.
	\end{align*}
	By Assumption \ref{assumption4} and Lemma \ref{lemma:smooth}, we have, under the null hypothesis of no oscillation, 
	\begin{align*} 
	\max_{1\leq k\leq n} |\bar\Theta_k^{(\mu)}(W)|  &\, \lesssim \max_{1 \leq k \leq n} \max_{w \in W} \left\{ \left|\sum_{i = 1}^k f(i/n) \cos( w i ) \right|,\left|\sum_{i = 1}^k f(i/n) \sin( w i ) \right|  \right\} \\
	&\, = O(n^{1/4})\,,
	\end{align*}
	which implies that
	\[  
	\Big|\max_k |n^{-1/2}\bar\Theta_k(W)| -  \max_k |n^{-1/2}\bar\Theta_k^{(\epsilon)}(W)|\Big| \leq  \max_k |n^{-1/2}\bar\Theta_k^{(\mu)}(W)| = O(n^{-1/4})\,. 
	\]
	Write it as 
	\begin{eqnarray}\label{eq:thm44}
	\max_k |n^{-1/2}\bar\Theta_k(W)| = \max_k |n^{-1/2}\bar\Theta_k^{(\epsilon)}(W)| + O(n^{-1/4})\, .
	\end{eqnarray}

	Now, let $\{y_i\}$ be centred Gaussian random variables with the same covariance structure as $\{\epsilon_i\}$. Let  $\bar\Theta_k^{(Y)}(W)$ be defined the same way as $\bar\Theta_k^{(\epsilon)}(W)$ by replacing $\epsilon_i$ with $y_i$.  Then by Proposition \ref{prop:complex_approximation} we get   
	\begin{align}
	&\sup_{|x| > d^*_{n,np} } \left| \pr( \max_k \|n^{-1/2}\bar\Theta^{(\epsilon)}_k(W)\|_\infty \leq x) -  \pr( \max_k \|n^{-1/2}\bar\Theta^{(Y)}_k(W)\|_\infty \leq x) \right|\nonumber\\
	 \lesssim  \,& G^*(n,np)\,. \label{eq:thm45}
	\end{align} 
	Note that $G^*(n,np)$ is the Gaussian approximation bound in Proposition \ref{prop:complex_approximation}. 
	By \eqref{eq:thm44}, \eqref{eq:thm45} and similar proofs as those of Lemma \ref{lemma3} and Proposition \ref{prop:complex_comparison}, we have
	\begin{align}
	&\sup_{|x| > d^*_{n,np} + O(n^{-1/4})  } \left| \pr( \max_k \|n^{-1/2}\bar\Theta^{(Y)}_k(W)\|_\infty \leq x ) -  \pr( \max_k \|n^{-1/2}\bar\Theta_k(W)\|_\infty \leq x ) \right|\nonumber\\ 
	\lesssim &\,G^*(n,np)+ \log(np)^{1+ \delta} n^{-1/4}+ 1/(np\log^{1/2}(np)) \lesssim G^*(n,np).\label{eq:thm46} 
	\end{align} 
Note that $d^*_{n,np} + O(n^{-1/4})\le 2d^*_{n,np}$ for a sufficiently large $n$.

By \eqref{eq:thm46}, we thus need to control the distance between
$\pr( \max_k \|n^{-1/2}\bar\Theta^{(Y)}_k(W)\|_\infty \leq x )$ and $\pr(  \|SZ(W)\|_\infty \leq x|X)$, which boils down to controlling $\Delta$ by Proposition \ref{prop:complex_comparison}.
Since $\cov(\Theta^{(Y)}(W))=\cov(\Theta^{(\epsilon)}(W))$, we consider
\[ 
\Delta :=  \sup_{i,j}|\cov(\Theta^{(\epsilon)}(W)) - \cov(S(W))|X)|_{ij}
\]
in Proposition \ref{prop:complex_comparison}.

The rest of the proof is to bound $\Delta$. Recall the definition of $S(W)$ and $SZ(W)$ in \eqref{Definition SW in the proof of Gaussian comparison} and \eqref{Definition SZW in the proof of Gaussian comparison}. Observe that the $(s,t)$-th entry of $\cov(S(W)|X)$, $\cov(S(W)|X)_{st}$, where $1\leq s,t\leq 2np$, can be represented as  
	\begin{eqnarray}\label{eq:41}
	\cov(S(W)|X)_{st}  =\frac{1}{m(n-m)}  \sum_{i = 1}^{ k -m +1 }S_{i,m}(w) S_{i,m}(w'),
	\end{eqnarray}
	or $\frac{1}{m(n-m)}  \sum_{i = 1}^{ k -m +1 }S_{i,m}(w) C_{i,m}(w')$, or $\frac{1}{m(n-m)}  \sum_{i = 1}^{ k -m +1 }C_{i,m}(w) C_{i,m}(w')$ for some $k$, $\omega$ and $\omega'$ depending on $s,t$ when $m<k\leq n$, or other format when $1\leq k\leq m$. Without loss of generality, we will focus on the case \eqref{eq:41} with $m<k\le n$ instead of other combinations of  $S_{i,m}$ and $C_{i,m}$ since it would not affect the rest of proof beside cosmetic reasons.
	
	To study \eqref{eq:41}, we further break down $S_{i,m}(w)$ into 
	\begin{align}
	S_{i,m}(w)=S_{i,m}^{(\epsilon)}(w)+S_{i,m}^{(\mu)}(w),\label{Definition: S_im(omega)}
	\end{align}
	where $S_{i,m}^{(\epsilon)}(w)$ is the stochastic part and $S_{i,m}^{(\mu)}(w)$ is the deterministic part. Observe that, under the null hypothesis of no oscillation, $S^{(\mu)}_{i,m}(w)=O(1)$ uniformly over $i$ and $w\in W$. Therefore
	\begin{align*}
	    &S_{i,m}(w) S_{i,m} (w')\\
	      =&\, 
	(S_{i,m}^{(\epsilon)}(w) + S_{i,m}^{(\mu)}(w) )  (S_{i,m}^{(\epsilon)}(w') + S_{i,m}^{(\mu)}(w')) \\
     = &\,(S_{i,m}^{(\epsilon)}(w) +O(1) )  (S_{i,m}^{(\epsilon)}(w') + O(1) )\\
     =&\, S_{i,m}^{(\epsilon)}(w) S_{i,m}^{(\epsilon)}(w') +  S_{i,m}^{(\mu)}(w')S_{i,m}^{(\epsilon)}(w) + S_{i,m}^{(\mu)}(w) S_{i,m}^{(\epsilon)}(w') + O(1)  .
	\end{align*} 
	Note that $\sum_{i=1}^{k-m+1}S_{i,m}^{(\mu)}(w')S_{i,m}^{(\epsilon)}(w) + S_{i,m}^{(\mu)}(w) S_{i,m}^{(\epsilon)}(w')$ is a linear combination of $\{\epsilon_i\}_{i=1}^n$. Using the result that $S^{(\mu)}_{i,m}(\omega)=O(1)$ and Lemma 6 of \cite{zhou2013}, we have
\begin{eqnarray*}
\left\|\max_k\bigg|\sum_{i=1}^{k-m+1}S_{i,m}^{(\mu)}(w')S_{i,m}^{(\epsilon)}(w) + S_{i,m}^{(\mu)}(w) S_{i,m}^{(\epsilon)}(w')\bigg|\right \|_q=O(m\sqrt{n})\,.
\end{eqnarray*}
Hence
\begin{align}\label{eq:thm47}
&\left\|\max_k\max_{w,w'}
\bigg|\sum_{i=1}^{k-m+1}S_{i,m}^{(\mu)}(w')S_{i,m}^{(\epsilon)}(w) + S_{i,m}^{(\mu)}(w) S_{i,m}^{(\epsilon)}(w')\bigg|\right\|_q\\
=\,& O(m\sqrt{n}p^{2/q})\,.\nonumber
\end{align}
	By Lemma \ref{lemma4}, for fixed $w,w' \in W$ 
	\begin{align}
	&\frac{1}{m(n-m)} \left\| \max_{1\leq k \leq n-m} \bigg| \sum_{i = 1}^{ k -m +1 } S_{i,m}^{(\epsilon)}(w) S_{i,m}^{(\epsilon)} (w')- \sum_{i = 1}^{ k -m +1 } \E S_{i,m}^{(\epsilon)}(w) S_{i,m}^{(\epsilon)} (w') \bigg| \right\|_{q'}\nonumber\\
	 =&\, O(\sqrt{m/n})\,. 
	\end{align}
	Recall that $q'=q/2$. Therefore, 
\begin{align}  
	&\frac{1}{m(n-m)} \left\|  \max_{ w,w' \in W}\max_{1\leq k \leq n-m} \left| \sum_{i = 1}^{ k -m +1 } S_{i,m}^{(\epsilon)}(w) S_{i,m}^{(\epsilon)} (w')- \sum_{i = 1}^{ k -m +1 } \E S_{i,m}^{(\epsilon)}(w) S_{i,m}^{(\epsilon)} (w') \right| \right\|_{q'}\nonumber\\
	=\,&O(p^{4/q} \sqrt{m/n}). \label{eq:thm48}
\end{align}

Observe that, if $\cov(S(W)|X)_{st}$ can be written in the form of \eqref{eq:41}, then the corresponding $(s,t)$ entry of $\mbox{Cov}(\Theta^{(\epsilon)}(W))$ is of the form $n^{-1}\E[S_k^{(\epsilon)}(w)S_l^{(\epsilon)}(w')]$ or $n^{-1}\E[S_k^{(\epsilon)}(w')S_l^{(\epsilon)}(w)]$ for some $l\ge k$, where we recall the definition of $S_k^{(\epsilon)}(w)$ and $C_k^{(\epsilon)}(w)$ in Section \ref{sec:nullboots}. We shall only focus on the first case due to symmetry.

Observe that, for all possible $k,l$ so that $k\le l$ and $k\geq m$, and $w,w'$, we have
	\begin{align*}
	& \left| \frac{1}{m(n-m)} \sum_{i = 1}^{ k -m +1 } \E S_{i,m}^{(\epsilon)}(w) S_{i,m}^{(\epsilon)} (w') - \frac{1}{n} \E S_{k}^{(\epsilon)} (w) S_{l}^{(\epsilon)} (w') \right| \\ 
	\leq &\, 	\left| \frac{1}{m(n-m)} \sum_{i = 1}^{ k - m+1 } \E S_{i,m}^{(\epsilon)}(w) S_{i,m}^{(\epsilon)} (w')- \frac{1}{n}\E S_{k}^{(\epsilon)} (w) S_{k}^{(\epsilon)} (w')  \right| \\
	&\,\,+ \left| \frac{1}{n}\E\left(S_{k}^{(\epsilon)} (w) \sum_{i = k + 1}^{l} \sin( w' i ) \epsilon_i \right) \right| .
	\end{align*}
Note that
	\begin{align*}
	\left| \frac{1}{n}\E\left(S_{k}^{(\epsilon)} (w) \sum_{i = k + 1}^{l} \sin(2\pi  w' i ) \epsilon_i \right) \right| \le &\,   \frac{1}{n} \sum_{i = k + 1}^{l} \sum_{j = 1}^k  \left| \cov(\epsilon_i,\epsilon_j)  \right| \\
	\le&\, C \frac{1}{n} \sum_{i = k + 1}^{l} \sum_{j = 1}^k  \left| i-j  \right|^{-d}= O(1/n),
	\end{align*}	
	where we utilized $|\cov(\epsilon_i,\epsilon_j)|\le C|i-j|^{-d} $ by Lemma 6 of Zhou (2014). Recall $m = n^\theta$ and $\theta<1$. 
	Hence,
	\begin{align*}
	&\left| \frac{1}{m(n-m)} \sum_{i = 1}^{ k - m+1 } \E S_{i,m}^{(\epsilon)}(w) S_{i,m}^{(\epsilon)} (w')-  \frac{1}{n} \E S_{k}^{(\epsilon)} (w) S_{k}^{(\epsilon)} (w') \right|\\
	\leq&\,  \sum_{ |i - j| \leq m}  \left| \frac{m - |i-j| }{m(n-m)} - \frac{1}{n} \right| 	\left| \E \epsilon_i \epsilon_j \right| + \frac{1}{n}  \sum_{ |i - j| >  m} \left| \E \epsilon_i \epsilon_j \right|  \\
	&\qquad\qquad\qquad ( \text{the indices } i,j \text{ satisfy } 1\leq  i,j \leq k ) \\ 
	\lesssim &\,  \sum_{ |i - j| \leq m} \left( \frac{m}{n^2} + \frac{|i-j|}{mn} \right) |i - j|^{-d} + \frac{1}{n}  \sum_{ |i - j| >  m } |i - j|^{-d}  \\ 
	\lesssim &\,  \sum_{ s= 0}^{m - 1} \left( \frac{m}{n^2} + \frac{s}{mn} \right) (k - s) s^{-d} +  \frac{1}{n}  \sum_{t = m}^{k -1} (k - t)  t^{-d}  \\
	= &\,O\left(\frac{m}{n}  \right)+ O\left( \frac{1}{m}  \right) + O(m^{-d + 1})= O(m/n)  + O(1/m). 
	\end{align*}
	Therefore,
	\begin{align*}
	&\left|  \frac{1}{m(n-m +1)}  \sum_{i = 1}^{ k -m +1 } \E S_{i,m}^{(\epsilon)}(w) S_{i,m}^{(\epsilon)} (w')  - \frac{1}{n} \E S_{1,k}^{(\epsilon)} (w) S_{1,l}^{(\epsilon)} (w')  \right| \\
	= \,& O(m/n)  + O(m^{-1})
	\end{align*}
uniformly for all $1\le k\le l\le n-m$ and $w,w'\in W$.	Together with \eqref{eq:thm47} and \eqref{eq:thm48}, we have that	
\begin{align*}
	\Delta  &\,= O(1/m)+O_\pr(p^{2/q}/\sqrt{n})+O_\pr(p^{4/q} \sqrt{m/n}) +O(m/n)  + O(1/m)\\
	&\, =O_\pr(p^{4/q} \sqrt{m/n} +1/m ). \end{align*}
Finally, by Markov's inequality,  \eqref{eq:thm47} and \eqref{eq:thm48}, we have
$$
\pr(A_n)\ge 1-\|\Delta\|^{q'}_{q'}/[h^{q'}_n(p^{4/q}\sqrt{m/n} +1/m )^{q'}]\ge 1-C/h_n^{q'}
$$
with right hand side converging to 1. Combing with Proposition \ref{prop:complex_comparison} with $\delta=0.5$, we have the theorem.
\end{proof}

\subsubsection{Estimation accuracy}

The following lemma establishes that, for any given costant $c\in (0,1)$, the $(1-c)$ quantile of  $F(W)\ge \sqrt{\log (Cn^{1/2+\theta}) }$ for any $\theta\in (0,1/7)$. Observe that the threshold $d^*_{n,np}$ goes to 0 in Theorem \ref{thm:stage1boots}. Hence, Theorem \ref{thm:stage1boots} can be applied to the DPPT.  

\begin{lemma}\label{lem6}
Under Assumptions \ref{assumption1} to \ref{assumption7} and the null hypothesis that $\mu_{i,n}=f(i/n)$, we have that for any $\theta\in(0,1/7)$ and $c\in(0,1)$, there exists a positive constant $C<\infty$ that does not depend on $n$ such that
\begin{eqnarray*}
\pr\Big(F(W)\ge \sqrt{\log (Cn^{1/2+\theta}) }\Big)\ge c
\end{eqnarray*}
for a sufficiently large $n$. 
\end{lemma}

\begin{proof}
Denote  
\[
F^*(W)=\max_{\omega\in W}|L_n(n,\omega)|/\sqrt{n}\,. 
\]
Obviously $F(W)\ge F^*(W)$.  Let $\{y_i^*\}$ be a centered Gaussian process that preserves the covariance structure of $\{x_i\}$. Define 
\begin{align}
T^*_y=\max_{\omega\in W}|L_{n,y}(n,\omega)|/\sqrt{n}\,, 
\end{align}
where 
\begin{align}
    L_{n,y}(n,\omega)=\sum_{k=1}^ny_k^*e^{\sqrt{-1}\omega k}. \label{Definition:Lny}
\end{align}
By Proposition \ref{prop:complex_approximation} and Lemma \ref{lemma:smooth}, we have that    
\begin{eqnarray}\label{eq:51}
\sup_{|x|>d^*_{n,p}}|\pr(F^*(W)\le x)-\pr(T_y^*\le x)|\lesssim G^*(n,p)\rightarrow 0\,. 
\end{eqnarray}
Observe that $d^*_{n,p}\rightarrow 0$. Consider an equally-spaced subset $W^*$ of $W$ such that the mesh size of $W^*$ is proportional to $ n^{-1/2-\theta}$ for some $\theta\in (0,1/7)$. Let 
\[
CL_{y}(\omega):=\sum_{k=1}^ny_k^*\cos(\omega k)/\sqrt{n}\ \mbox{ and }\ T^{**}_y:=\max_{\omega\in W^*}|CL_{y}(\omega)|\,. 
\]
Then clearly $T^*_y\ge T^{**}_y$. By \eqref{eq:51} and the fact that $\sqrt{\log (Cn^{1/2+\theta})}\gg 1$ when $n$ is sufficiently large, it suffices to show that
\begin{eqnarray}\label{eq:52}
\pr\Big(T^{**}_y\ge \sqrt{\log (Cn^{1/2+\theta})}\Big)\ge c
\end{eqnarray}
for sufficiently large $n$. Define $\Gamma(i,j)=\cov(y^*_i,y^*_j)$. Then, by Lemma 6 of \cite{zhou2014} and Assumption \ref{assumption3}, we have $|\Gamma(i,j)|\le C(|i-j|+1)^{-d}$. Hence, for any $\omega$, $\omega'\in W^*$, we have
\begin{eqnarray*}
\cov(CL_y(\omega),CL_y(\omega'))=\frac{1}{n}\sum_{|k|<n}\sum_{1\le i,i+k\le n}\Gamma(i,i+k)\cos(\omega i)\cos(\omega'(i+k))\,.
\end{eqnarray*}
By Assumption 7 and the proof of Lemma \ref{lemma:smooth}, we have that, for each $|k|\leq n^{1/4}$,
\begin{eqnarray*}
\Big|\sum_{1\le i,i+k\le n}\Gamma(i,i+k)\cos(\omega i)\cos(\omega'(i+k))\Big|=O(n^{5/8+3\theta/4})\,;
\end{eqnarray*}
for $|k|>n^{1/4}$, we have 
\begin{eqnarray*}
\Big|\sum_{1\le i,i+k\le n}\Gamma(i,i+k)\cos(\omega i)\cos(\omega'(i+k))\Big|\le n\max_i|\Gamma(i,i+k)|\le nC|k|^{-d}.
\end{eqnarray*}
Hence
\begin{eqnarray}\label{eq:53}
\cov(CL_y(\omega),CL_y(\omega'))\le Cn^{-\alpha}
\end{eqnarray}
for some $\alpha>0$.
For any $\omega\in W^*$, let $CL'_y(\omega)$ be a centered Gaussian random variable with the same {\em variance} as $CL_y(\omega)$, and $CL'_y(\omega)$ and $CL'_y(\omega')$ are independent for $\omega\neq\omega'$. Let 
$T^{***}_y=\max_{\omega\in W^*}|CL'_{y}(\omega)|$. By Assumption \ref{assumption6}, we have that $\text{Var}(CL'_{y}(\omega))\ge \delta_1 /2$ for sufficiently large $n$. Hence by Theorem 2 of \cite{chernozhukov2015}, we have
\begin{eqnarray}\label{eq:54}
\sup_x|\pr(T^{**}_y\le x)-\pr(T^{***}_y\le x)|\le Cn^{-\alpha/3}\log^{2/3}(n).
\end{eqnarray}
Since $T_y^{***}$ is the maximum of an independent Gaussian vector of length $O(n^{1/2+\theta})$ and each component's variance is bounded from below. It is straightforward to derive that
$$\pr\Big(T^{***}_y\ge \sqrt{\log (Cn^{1/2+\theta})}\Big)\ge c
$$
for some finite constant $C$ that does not depend on $n$. Therefore \eqref{eq:52} holds by \eqref{eq:54}. The lemma follows. 
\end{proof}

\begin{lemma}\label{lem:9}
Suppose that the assumptions of Theorem \ref{thm:s1_accuracy} hold true. Then on a sequence of events $D_n$ with $P(D_n)\ge 1-2h^{-q'}_n$, where $q'=q/2$ and $h_n$ is a positive sequence of real numbers which diverges to infinity at an arbitrarily slow rate, we have that, conditional on $X=\{X_{i}\}_{i=1}^n$, 
\begin{enumerate}
\item[(a)] if $\Omega=\emptyset$, then we have
$\pr(\tilde{F}_{m,l}(W)\ge \log(pn))\lesssim n^{-c}$; 

\item[(b)] if $\Omega\neq \emptyset$, then $\pr(\tilde{F}_{m,l}(W)\ge \sqrt{m}\log (pn))\lesssim n^{-c}$,
\end{enumerate}
where $c$ is any positive and finite constant.
\end{lemma}
Recall that $\tilde{F}_{m,l}(W)$ is the multiplier bootstrap statistic. Hence Lemma \ref{lem:9} implies that the critical values of the bootstrap are bounded by $\log (pn)$ with high probability if $\Omega=\emptyset$. If $\Omega$ is not empty, then the critical values are bounded by $\sqrt{m}\log (pn)$ with high probability.  

\begin{proof}
Recall the definition of $S_{j,n}(\omega)$, $S^{(\mu)}_{j,n}(\omega)$ and $S^{(\epsilon)}_{j,n}(\omega)$ in \eqref{Definition: S_im(omega)}. Decompose $C_{j,n}(\omega)$ into $C^{(\mu)}_{j,n}(\omega)$ and $C^{(\epsilon)}_{j,n}(\omega)$ in a similar way. Since $G_i$ are i.i.d standard Gaussian, we observe that, conditional on $X$, 
\begin{equation*}
	    \begin{split}
	   &\frac{1}{m(n-m)} \left[\var \left( \sum_{j = 1}^{n-m}  \sum_{l = j }^{j + m}  \cos(wl) X_l G_{j} \right) +  \var \left( \sum_{j = 1}^{n-m}  \sum_{l = j }^{j + m} \sin(wl) X_l G_{j} \right) \right]\\
	         =&\,  \frac{1}{m(n-m)} \left[\sum_{j = 1}^{n-m} \left(  \sum_{l = j }^{j + m}\cos(wl) X_l \right)^2  +   \sum_{j = 1}^{n-m} \left(  \sum_{l = j }^{j + m}\sin(wl) X_l \right)^2\right] \\
	  =&\,\frac{1}{m(n-m)}\left[\sum_{j=1}^{n-m}(C_{j,n}(\omega))^2+\sum_{j=1}^{n-m}(S_{j,n}(\omega))^2\right].
	    \end{split}
	\end{equation*}
	For the stochastic part, based on similar arguments as those for the proof of \eqref{eq:thm48}, we have that
\begin{eqnarray}\label{eq:lem61}
\Big\|\frac{1}{m(n-m)}\max_{1\le k\le n-m}\max_{\omega\in W}  \sum_{j = 1}^{k}\Big[  (S^{(\epsilon)}_{j,m}(w))^2-\E(S^{(\epsilon)}_{j,m}(w))^2\Big]  \Big\|_{q'}\lesssim p^{2/q}\sqrt{\frac{m}{n}}, 
\end{eqnarray}
and
\begin{eqnarray}\label{eq:lem62}
\frac{1}{m(n-m)}\Big\|\max_k\max_{w}
\Big|\sum_{i=1}^{k}S_{i,m}^{(\mu)}(w)S_{i,m}^{(\epsilon)}(w)\Big|\Big\|_q\lesssim p^{1/q}/\sqrt{n}.
\end{eqnarray}
Note that $p^{1/q}/\sqrt{n}\ll p^{2/q}\sqrt{\frac{m}{n}}$ and $p^{2/q}\sqrt{\frac{m}{n}}$ converges to zero polynomially fast. Therefore, if we define a sequence of events 
$D_n:= D_{n,1}\cap D_{n,2}$, where
$$
D_{n,1}:=\left\{\frac{1}{m(n-m)}\max_{1\le k\le n-m}\max_{\omega\in W} \sum_{j = 1}^{n-m}\Big[  (S^{(\epsilon)}_{j,m}(w))^2-\E(S^{(\epsilon)}_{j,m}(w))^2\Big] \le p^{2/q}\sqrt{m/n}h_n\right\}, 
$$
$$
D_{n,2}:=\left\{\frac{1}{m(n-m)}\max_{1\le k\le n-m}\max_{w\in W}
\Big|\sum_{i=1}^{k-m+1}S_{i,m}^{(\mu)}(w)S_{i,m}^{(\epsilon)}(w)\Big|\le p^{1/q}/\sqrt{n}h_n\right\},
$$
then $P(D_n)\ge 1-2/h^{q'}_n$ by Markov's inequality.  
Meanwhile, it is easy to show that
\begin{eqnarray}\label{eq:lem63}
\frac{1}{m(n-m)}\max_{1\le k\le n-m}\max_{\omega\in W}  \sum_{j = 1}^{n-m}\E (S^{(\epsilon)}_{j,m}(w))^2  \lesssim 1.
\end{eqnarray}

For the deterministic part, due to Assumption \ref{assumption4}, it is easy to see that
	 \begin{eqnarray}\label{eq:lem64}
	 \max_{w \in W} \max_{1 \leq j \leq n-m} \left| \sum_{l = j }^{j + m} \sin(wl) \mu_l \right| = O({m})
	 \end{eqnarray}
	if $\Omega\neq\emptyset$.   Otherwise, by Assumption \ref{assumption4}, the bound in \eqref{eq:lem64} becomes $O(1)$.

	Similar result holds for the cosine terms. Hence, by a simple Orcliz norm maximum inequality, we have that, conditional on $X$ and on the event $D_n$, 
\begin{eqnarray}\label{eq:l92}	
\pr\Big(	 \max_{w \in W} \max_{1 \leq k \leq n-m} \left| \sum_{j = 1}^{k}  \sum_{l = j }^{j + m}e^{ \im w l}X_l G_{j} /\sqrt{m (n-m) } \right|\ge \sqrt{m}\log(pn)\Big)\lesssim n^{-c}
\end{eqnarray}	 
for any positive and finite constant $c$ if $\Omega\neq \emptyset$. If $\Omega=\emptyset$, then by Lemma \ref{lemma5}, we have that, conditional on $X$ and on the event $D_n$, 	
\begin{eqnarray}\label{eq:l93}	
\pr\Big(\max_{w \in W} \max_{1 \leq k \leq n-m} \left| \sum_{j = 1}^{k}  \sum_{l = j }^{j + m}e^{ \im w l}X_l G_{j} /\sqrt{m (n-m) } \right|\ge \log(pn)\Big)\lesssim n^{-c}
\end{eqnarray}
for any positive and finite $c$. The lemma follows.

\end{proof}

\begin{lemma}\label{lemb6}
Assume $\Omega\neq \emptyset$. Let $\Omega=\{\omega_1,\ldots,\omega_K\}$. Then under the assumptions of Proposition \ref{prop:nophasechange} and \eqref{eq:thmcon} held true, we have
$\max_{1\le k\le K} |\hat{\omega}_k - \omega_k|=o_\pr(n^{-3/2}\log n)$ and
$$\pr\left(\max_{1\le k\le K} |\hat{\omega}_k - \omega_k|\ge n^{-3/2}\log n\right)\lesssim (\log n)^qn^{(\theta-1)q/8}+G^*(n,pn)\,.
$$
\end{lemma}
\begin{proof} For the simplicity of presentation, we shall only prove the case where $\Omega=\{\omega_0\}$; that is, there is only one oscillatory frequency. The general case follows by similar arguments since the number of oscillatory frequencies is assumed to be bounded and the frequencies are well separated by Assumption \ref{assumption8}. 

Under the above assumption, the mean function is
$$\mu_j=\sum_{0\le r\le M_0}(A_r\cos(\omega_0 j)+B_r\sin(\omega_0 j))\mathbb{I}({b_{r} \leq j \leq b_{r+1} })+f(j/n).$$
By Lemma \ref{lemma:smooth} and Assumption \ref{assumption4}, the contribution of  $f(\cdot)$ is negligible. Hence, without loss of generality, we set $f(\cdot)=0$ in the sequel.
Recall the definitions of $\bar\Theta^{(\mu)}_k(\omega)$, $\bar\Theta^{(\epsilon)}_k(\omega)$ and $\bar\Theta_k(\omega)$ in \eqref{eq:thetamu}. Define
\begin{eqnarray*}
\bar\Theta^{(\mu)}(\omega):=\max_{1\le k\le n}|\bar\Theta^{(\mu)}_k(\omega)|\quad\text{and}\quad \bar\Theta(\omega):=\max_{1\le k\le n}|\bar\Theta_k(\omega)|\,.
\end{eqnarray*}
Note that $F(W)=\max_{\omega\in W}\bar\Theta(\omega)/\sqrt{n}$.
It is easy to see that for any $\omega\in W$,
$$
\bar\Theta^{(\mu)}(\omega)=\max_{1\le i\le M_0+1}|\bar\Theta^{(\mu)}_{b_i}(\omega)| \,.
$$
Elementary but tedious calculations using sums of trigonometric series yield that there exist finite and positive constants $c$ and $C$ such that 
\begin{align}
cn\le\bar\Theta^{(\mu)}(\omega_0)\le Cn    
\end{align}

Note that by \eqref{eq:lem61} to \eqref{eq:lem64} in Lemma \ref{lem:9} and the Markov's inequality, the critical value of the DPPT is no larger than $n/\log n$ with probability at least $1-C(\log n)^q\exp(\log p-q(1-\theta)\log n/4)\ge 1-C(\log n)^q n^{(\theta-1)q/8}$, where we used assumption \eqref{eq:thmcon} in the above inequality. Therefore, the critical values will be surpassed with probability at least $1-C (\log n)^q n^{(\theta-1)q/8}$. If $|\omega-\omega_0|\ge 1/n$, then 
\begin{eqnarray}\label{eq:66}
\bar\Theta^{(\mu)}(\omega)/\bar\Theta^{(\mu)}(\omega_0)\le 1-c_0
\end{eqnarray}
for some $c_0>0$. If $|\omega-\omega_0|<1/n$, then using the formula for sums of trigonometric series, the assumption that there is no phase jump, as well as Taylor expansion, we have 
\begin{eqnarray}\label{eq:61}
\bar\Theta^{(\mu)}(\omega_0)-\bar\Theta^{(\mu)}(\omega)\ge C_0n^3|\omega-\omega_0|^2
\end{eqnarray}
for some finite and positive constant $C_0$. 

Recall that  $\bar\Theta^{(Y)}_k(\omega)=\sum_{i=1}^ky_i\exp(\sqrt{-1}i\omega)$, which comes from replacing $\epsilon_i$ in $\bar\Theta^{(\epsilon)}_k(\omega)$ by $y_i$, where $\{y_i\}$ is a centered Gaussian sequence having the same auto-covariance structure as that of $\{\epsilon_i\}$. By the proof of Lemma \ref{lemma5}, we have that
\begin{align}
&\pr\left(\max_{1\le k\le n}\max_{\omega\in\Omega}|\bar\Theta^{(\epsilon)}_k(\omega)|\ge \sqrt{n}\log n\right)\nonumber\\
=\,&\pr\left(\max_{1\le k\le n}\max_{\omega\in\Omega}|\bar\Theta^{(Y)}_k(\omega)|\ge \sqrt{n}\log n\right)+G^*(n,pn),\label{eq:69}
\end{align}
 Using an Orcliz norm maximum inequality (note that $p\asymp n^{3/2}\log n$), we have 
\begin{eqnarray}\label{eq:64}
\pr\left(\max_{1\le k\le n}\max_{\omega\in\Omega}|\bar\Theta^{(Y)}_k(\omega)|\ge \sqrt{n}\log n\right)\le n^{-c}
\end{eqnarray}
for any positive and finite constant $c$ if $n$ is sufficiently large. Define a sequence of events 
$$ 
A_n=\Big\{\max_{1\le k\le n}\max_{\omega\in\Omega}|\bar\Theta^{(\epsilon)}_k(\omega)|\ge \sqrt{n}\log n\Big\}\,.
$$ 
Since $G^*(n,pn)$ dominates $n^{-c}$, we have $\pr(A_n)\lesssim G^*(n,pn)$, which converges to 0. On the event $A^c_n$, we have by \eqref{eq:66} and \eqref{eq:61} that, if $|\omega-\omega_0|\ge n^{-1}/\log n$, then
\begin{eqnarray}\label{eq:62}
\bar\Theta(\omega)-\bar\Theta(\omega_0)\le -C\min(n, n^3|\omega-\omega_0|^2)+2\max_{1\le k\le n}\max_{\omega\in\Omega}|\bar\Theta^{(\epsilon)}_k(\omega)|<0.
\end{eqnarray}
On the other hand, if $|\omega-\omega_0|< n^{-1}/\log n$, then note that
$$
|\bar\Theta_k(\omega)-\bar\Theta_k(\omega_0)|\le \int_{\omega_0}^{\omega}\Big|\sum_{j=1}^k j\epsilon_j\exp(tj\sqrt{-1})\Big|\,dt:=\int_{\omega_0}^{\omega}|\tilde{\Theta}_k^{(\epsilon)}(t)|\,dt.$$
Similarly to the arguments above and by a simple chaining technique, we have that
\begin{align}
&\pr\left(\max_{1\le k\le n}\sup_{|t-\omega_0|< n^{-1}/\log n}|\tilde{\Theta}_k^{(\epsilon)}(t)|\ge C_0n^{3/2}\log n\right)\nonumber\\
\lesssim &\,G^*(n,pn)+n^{-1}\lesssim G^*(n,pn)\label{eq:63}\,.
\end{align}
Define the sequence of events $$
B_n=\left\{\max_{1\le k\le n}\sup_{|\omega-\omega_0|< n^{-1}/\log n}|\tilde\Theta^{(\epsilon)}_k(\omega)|\ge C_0n^{3/2}\log n\right\}.
$$ 
Then we have, on the event $B^c_n$, if $n^{-1}/\log n>|\omega-\omega_0|> n^{-3/2}\log n$, 
\begin{align}
\bar\Theta(\omega)-\bar\Theta(\omega_0)\le&\, -C_0n^3|\omega-\omega_0|^2+\max_{1\le k\le n}\max_{|\omega-\omega_0|< n^{-1}/\log n}|\bar\Theta^{(\epsilon)}_k(\omega_0)-\bar\Theta^{(\epsilon)}_k(\omega)|\nonumber\\
<&\, -C_0n^3|\omega-\omega_0|^2+C_0n^{3/2}\log n|\omega-\omega_0|<0\,.\label{eq:65}
\end{align}
By \eqref{eq:62} and \eqref{eq:65}, we have that, on the event $ A^c_n\cap B^c_n$, $|\hat{\omega}_0-\omega_0|\le n^{-3/2}\log n$. Note that $1-\pr( A^c_n\cap B^c_n)= O(G^*(n,pn))$. Similarly, for any fixed $\delta>0$, we can derive 
$$
\pr(|\hat{\omega}_0-\omega_0|> \delta n^{-3/2}\log n )=O(G^*(n,pn))\,.
$$ 
Hence the lemma follows.

\end{proof}

\begin{lemma}\label{lem:8}
Suppose $\Omega\neq \emptyset$. Let $\Omega=\{\omega_1, \ldots, \omega_K\}$. Then under the assumptions of Theorem \ref{thm:s1_accuracy}, we have that
$\max_{1\le k\le K} |\hat{\omega}_k - \omega_k|=O_\pr(n^{-1})$ and
$$\pr\left(\max_{1\le k\le K} |\hat{\omega}_k - \omega_k|\ge n^{-1}h_n\right)\lesssim (\log n)^qn^{(\theta-1)q/8}+G^*(n,pn)$$
for any sequence $h_n>0$ that diverges to infinity arbitrarily slowly.
\end{lemma}
\begin{proof}
This lemma follows from similar and simpler arguments as those of Lemma \ref{lemb6}. The only difference between the condition of Lemma \ref{lemb6} and that of the current lemma is \eqref{eq:nophasechange}.
We shall briefly outline the proof here. First of all, elementary calculations using sums of trigonometric series yield that there exist finite and positive constants $c$ and $C$ such that $cn\le \bar\Theta^{(\mu)}(\omega_k)\le Cn $ for any $\omega_k\in\Omega$. Note that by Lemma \ref{lemb6}, the critical values of the DPPT is no larger than $n/\log n$ with probability at least  $1-Cn^{(\theta-1)q/8}$. Therefore the critical values will be surpassed with probability as least $1-Cn^{(\theta-1)q/8}$ for the first $K$ steps of the DPPT. Furthermore, elementary calculations yield that
$\bar\Theta^{(\mu)}(\omega')=O(n/h(n))$ if $h(n)/n\le |\omega'-\omega_k|$ for all $\omega_k\in\Omega$. Note that by Assumption \ref{assumption8}, two oscillatory frequencies are at least $O(1/\log n)$ away. On the other hand, we have by \eqref{eq:69} and \eqref{eq:64} that
\begin{eqnarray}\label{eq:71}
\pr(\tilde A_n)\lesssim G^*(n,np),
\end{eqnarray}
where 
\begin{eqnarray}\label{eq:72}
\tilde A_n=\left\{\max_{1\le k\le n}\max_{\omega\in\Omega}|\bar\Theta^{(\epsilon)}_k(\omega)|\ge \sqrt{n}\log n\right\}\,.
\end{eqnarray}
On the event $\tilde A^c_n$, we have $\bar\Theta(\omega')<\bar\Theta(\omega_k)$ for all $\omega_k\in\Omega$ if $h(n)/n\le |\omega'-\omega_k|$. Hence the lemma follows.
\end{proof}

\begin{proof}[Proof of Theorem \ref{thm:s1_accuracy}]
Part 1 of Theorem \ref{thm:s1_accuracy} follows directly from Theorem \ref{thm:stage1boots}. Hence, we only need to prove Part 2. By  Lemma \ref{lem:8}, we have that
\begin{eqnarray*}
\pr\left(\max_{1\le k\le K} |\hat{\omega}_k - \omega_k|\ge n^{-1}h_n\right)\lesssim n^{(\theta-1)q/8}+G^*(n,pn).
\end{eqnarray*}
Observe that $G^*(n,pn)$ converges to $0$. 
Hence, to prove Part 2, we only need to prove that 
\begin{eqnarray}\label{eq:thm31}
\pr(|\hat{\Omega}|=|\Omega|)\rightarrow 1-\alpha.
\end{eqnarray}
By the proof of Lemma \ref{lem:8}, we have that with probability at least  $1-Cn^{(\theta-1)q/8}$, the critical values will be surpassed for the first $K$ steps of the DPPT. That is, $\pr(|\hat{\Omega}|<|\Omega|)\lesssim n^{(\theta-1)q/8}$ %
Thus, to prove \eqref{eq:thm31}, it suffices to prove
\begin{eqnarray}\label{eq:thm32}
\pr(|\hat{\Omega}|>|\Omega|)\rightarrow \alpha.
\end{eqnarray}
Define the sets 
$$
V=\bigcup_{i = 1}^{K} \left[w_i - (h_n/n+\log(m)/(4m^{1/2})), w_i + (h_n/n+\log(m)/(4m^{1/2}))\right]\bigcap W
$$ 
and 
$$
V'=\bigcup_{i = 1}^{K} \left[\hat w_i - \log(m)/(4m^{1/2}), \hat w_i + \log(m)/(4m^{1/2})\right]\bigcap W\,.
$$ 
Observe that $V$ is a fixed set and $V'$ is a random set. Recall the definition of $\tilde A_n$ in \eqref{eq:72}. On the event $\tilde A_n^c$, by Lemma \ref{lem:8}, $V'$ is a subset of $V$ and the cardinality of the set difference $|V\backslash V'|=O(\sqrt{n}{h_n}\log n)$, where note that $p\asymp n^{3/2}\log n$. 

Recall the definition of $L_{n,y}$ in \eqref{Definition:Lny}. Note that, for any $\omega\in V\backslash V'$, we have that 
$\E(\max_{1\le i\le n}|L_{n,y}(i,\omega)|^2)= O(n)$.
Hence by the proof of Lemma \ref{lem6}, it follows that
$$
\pr\Big(F(V\backslash V')\ge \sqrt{\log( Cn^{1/2+\theta_1})}\Big)\le G^*(n,n^{3/2}h_n\log n)+n^{-c}\lesssim G^*(p,pn)
$$
for any positive and finite constant $c\in(1/8,\theta_1)$, where $\theta_1$ is any constant in $(1/8,1/7)$. Note that $n^{-c}$ is bounded by $G^*(p,pn)$ when $c$ is larger than $1/8$. Furthermore, $G^*(n,n^{3/2}{h_n}\log n)$ is bounded by $G^*(n,pn)$. Define the events 
$$
E_n=\Big\{F(V\backslash V')\ge \sqrt{\log( Cn^{1/2+\theta_1})}\Big\},
$$
which by the above argument satisfies $\pr(E_n)\to 0$ when $n\to \infty$.
By Lemma \ref{lem6}, we have that there exists constants $C_{\alpha}>0$ and $\theta_2\in(\theta_1,1/7)$ such 
that 
$$
\pr\Big(F(W\backslash V)\ge \sqrt{\log( C_{\alpha}n^{1/2+\theta_2})}\Big)>\alpha.
$$
 Hence, on the event $E_n^c$, we have that the $(1-\alpha)$-quantile of 
$F(W\backslash V)$ and $F(W\backslash V')$ are asymptotically equivalent. In particular, asymptotically the $(1-\alpha)$-quantile of $F(W\backslash V)$ is at least as large as $O(\sqrt{\log n})$.

Next, we discuss the multiplier bootstrap statistics $\tilde{F}_{m,l}$. We have that, on the event $\tilde A^c_n$ and for any $\omega\in W\backslash V'$,
\begin{eqnarray*}
\sum_{j=1}^k\Big|\sum_{i=j}^{j+m-1}\exp(\sqrt{-1}i\omega)\mu_{i,n}\Big|^2/(m(n-m))\lesssim \log^{-1} m
\end{eqnarray*}
uniformly in $k$ since $\omega$ is a least $O(\log m/m^{1/2})$ away from an oscillatory frequency. Therefore, with similar arguments as above, we have that the $(1-\alpha)$-quantile of 
$\tilde F(W\backslash V)$ and $\tilde F(W\backslash V')$ are asymptotically equivalent with probability at least $1-O(h_n^{-q'})$.
Now, by the construction of $V$, there is no oscillation left, so by the proof of Theorem \ref{thm:stage1boots}, we have that, with probability as least $1-O(h_n^{-q'})$,
\begin{eqnarray}\label{eq:thm33}
|\pr(F(W\backslash V)\le x)-\pr(\tilde F(W \backslash V)\le x|X)|\rightarrow 0
\end{eqnarray}
uniformly for all $x>d^\circ_{n,|W\backslash V|}$. Note that $d^\circ_{n,|W \backslash V|}$ converges to 0 and the $(1-\alpha)$-quantile of $F(W\backslash V)$ is at least as large as $O(\sqrt{\log n})$. Hence by \eqref{eq:thm33}, we have $\pr(|\hat \Omega|>|\Omega|)\rightarrow \alpha$. 
The theorem follows.
\end{proof}
\begin{proof}[Proof of Proposition \ref{prop:nophasechange}]
The proof of this proposition follows from Lemma \ref{lemb6} and the essentially the same arguments as those in the proof of Theorem \ref{thm:s1_accuracy}. Details are omitted.    
\end{proof}

\subsection{Proof of results in stage 2}\label{sec:E2}

\subsubsection{Consistency}

The following Lemmas \ref{lem:b8}, \ref{lem:b9} and \ref{lem:12} show that the error in estimating the oscillatory frequency has asymptotically negligible impact on the local change point detection algorithm.

\begin{lemma} \label{lem:b8}
	Let $\{\epsilon_k\}$ be a zero mean PLS time series satisfying Assumption \ref{assumption3} and $m^\circ$ satisfy $m^\circ\asymp n^{\gamma}$, where $\gamma\in(0,1)$. Define 
	\[ 
	T_n^{(\epsilon)}(w) := \sup_{ m^\circ < l < n-m^\circ} \left|\sum_{k = l - m^\circ}^{l}  e^{ \im  w (k-l) } \epsilon_k  -  \sum_{k = l +1}^{l+m^\circ+1}  e^{ \im  w (k-l) } \epsilon_k \right|/\sqrt{2m^\circ} .
	\] 
Suppose $\hat{w}$ is an estimator of the oscillatory frequency $w$. If $\pr(|w - \hat{w}|\ge a_n) = O(b_n) $ for some sequences $a_n$ and $b_n$, then 
	\[ 
	\pr\Big(|T_n^{(\epsilon)}(w) - T_n^{(\epsilon)}(\hat{w} )|\ge a_n m^\circ n^{1/q}\log(n)\Big)\lesssim b_n+1/\log^q n.
	\]	
\end{lemma}

\begin{proof}
Define  
$$
H^{(\epsilon)}_{l,m^\circ}(\omega)=\sum_{k = l - m^\circ}^{l}  e^{ \im  w (k-l) } \epsilon_k  -  \sum_{k = l +1}^{l+m^\circ+1}  e^{ \im  w (k-l) } \epsilon_k\,.
$$
For any $w'$ such that $|w'-w|<a_n$, we have
\begin{align*}
H^{(\epsilon)}_{l,m^\circ}(\omega)-&H^{(\epsilon)}_{l,m^\circ}(\omega')
=\int^\omega_{\omega'}(H^{(\epsilon)}_{l,m^\circ})'(\theta)\,d\theta\\
=&\,\int^\omega_{\omega'}\Big[\sum_{k = l - m^\circ}^{l}  \im(k-l)e^{ \im  \theta (k-l) } \epsilon_k  -  \sum_{k = l +1}^{l+m^\circ+1}  \im(k-l)e^{ \im  \theta (k-l) } \epsilon_k\Big]\,d\theta\,.
\end{align*}
By Lemma 6 of \cite{zhou2013}, we have that
\begin{eqnarray}\label{eq:b81}
\max_{1\le l\le n}\sup_{\theta\in[0,\pi]}\|(H^{(\epsilon)}_{l,m^\circ})'(\theta)\|_q=O((m^\circ)^{3/2}).
\end{eqnarray}
Therefore, by \eqref{eq:b81} and a simple ${\cal L}^q$ maximum inequality, we have
\begin{align}\label{eq:b82}
\Big\|\max_{1\le l\le n}\sup_{|\omega-\omega'|< a_n}|H^{(\epsilon)}_{l,m^\circ}(\omega)-H^{(\epsilon)}_{l,m^\circ}(\omega')|\Big\|_q \le&\, \int_{\omega-a_n}^{\omega+a_n}\Big\|\max_{1\le l\le n}|(H^{(\epsilon)}_{l,m^\circ})'(\theta)|\Big\|_q\,d\theta\nonumber\\
\lesssim&\, a_n(m^{\circ})^{3/2}n^{1/q}.
\end{align}
The above inequality implies that, on the event $\bar B_n^c$, where the events $\bar{B}_n := \{|w - \hat{w}|\ge a_n\}$, we have
\begin{equation}\label{eq:b83}
\Big\|\max_l|H^{(\epsilon)}_{l,m^\circ}(\omega)-H^{(\epsilon)}_{l,m^\circ}(\hat\omega)|\Big\|_q\lesssim a_n(m^{\circ})^{3/2}n^{1/q}.   
\end{equation}
Hence, the lemma follows by \eqref{eq:b83} and the Markov's inequality.
\end{proof}

Note that the probability bound can be decreased to $b_n+n^{-\theta q}$ for some $\theta>0$ if we increase the threshold to $a_nm^\circ n^{1/q+\theta}$. Lemma \ref{lem:b8} implies that if $a_n = n^{-3/2} \log(n)$, then 
\[
|T^{(\epsilon)}_n(w) - T^{(\epsilon)}_n(\hat{w} )|  = O_p(n^{1/q-3/2}m^{\circ}\log n)\,.
\] 
And if $a_n = n^{-1} h_n$, then 
\[
|T^{(\epsilon)}_n(w) - T^{(\epsilon)}_n(\hat{w} )|  = O_p(n^{1/q-1}m^{\circ}h(n))\,.
\] 
Both bounds can converge to 0 for sufficiently large $q$ when $m^{\circ}\asymp n^{\gamma}$, $\gamma\in(0,1)$. 
 
\begin{lemma} \label{lem:b9}
	Let $\mu_k=\mu_{k,n}$ be the mean function defined in \eqref{eq:mean} and $m^\circ$ satisfy $m^\circ\asymp n^{\gamma}$, where $\gamma\in(0,1)$. Assume that $\Omega\neq\emptyset$, $\omega\in\Omega$ and there is no change point at frequency $\omega$. Define 
	\[  
	T_n^{(\mu)}(w) = \sup_{ m^\circ < l < n-m^\circ}  \left|\sum_{k = l - m^\circ}^{l}  e^{ \im  w (k-l) } \mu_k  -  \sum_{k = l +1}^{l+m^\circ+1}  e^{ \im  w (k-l) } \mu_k \right|/\sqrt{2m^\circ} .
	\] 
	Suppose $\hat{w}$ is an estimator of the oscillatory frequency $w$.
	If $\pr(|w - \hat{w}|\ge a_n) = O(b_n) $ for some sequences $a_n$ and $b_n$ with $a_n m^\circ\rightarrow 0$, then \[ \pr\Big(|T_n^{(\mu)}(w) - T_n^{(\mu)}(\hat{w} )|> C(m^\circ)^{3/2}a_n\Big)\lesssim b_n \]
	for some finite constant $C>0$, which does not depend on $n$.
\end{lemma}
\begin{proof}
Recall the events $\bar{B}_n=\{|w - \hat{w}|\ge a_n\}$ defined in the proof of Lemma \ref{lem:b8}. Define
$$
\ H^{(\mu)}_{l,m^\circ}(\omega)=\sum_{k = l - m^\circ}^{l}  e^{ \im  w (k-l) } \mu_k  -  \sum_{k = l +1}^{l+m^\circ+1}  e^{ \im  w (k-l) } \mu_k.
$$
For any $w'$ such that $|w'-w|<a_n$, we have
\begin{align}
 H^{(\mu)}_{l,m^\circ}(\omega)- &H^{(\mu)}_{l,m^\circ}(\omega')=\int^\omega_{\omega'} H^{(\mu)'}_{l,m^\circ}(\theta)\,d\theta. \label{eq:l101}\\
 =&\,\int^\omega_{\omega'}\Big[\sum_{k = l - m^\circ}^{l}  \im(k-l)e^{ \im  \theta (k-l) } \mu_k  -  \sum_{k = l +1}^{l+m^\circ+1}  \im(k-l)e^{ \im  \theta (k-l) } \mu_k\Big]\,d\theta \nonumber.
\end{align}
If $|\omega-\omega'|<a_n$ and for any oscillatory part in $\mu_k$ with frequency different from $\omega$, we have by Assumption \ref{assumption8} that the contribution of such oscillation to $H^{(\mu)'}_{l,m^\circ}(\theta)$ is $O(m^\circ\log n)$ uniformly in $l$ when $n$ is sufficiently large. 

Similarly, by similar arguments as those in Lemma \ref{lemma:smooth}, we have that the contribution of the smooth function $f(\cdot)$ to  $H^{(\mu)'}_{l,m^\circ}(\theta)$ is $O((m^\circ)^{5/4})$ uniformly in $l$. 
Hence by \eqref{eq:l101}, we have the contributions of the oscillatory frequencies other than $\omega$ and the smooth function $f(\cdot)$ to $| H^{(\mu)}_{l,m^\circ}(\omega)- H^{(\mu)}_{l,m^\circ}(\omega')|$ is of the order $a_n(m^\circ)^{5/4}$ uniformly over $l$ if $|\omega-\omega'|<a_n$. 
On the other hand, elementary but tedious calculations yield that, when $\mu_k=A\cos(\omega k+\theta)$, $A>0,0\le \theta<2\pi$, and $|\omega'-\omega|\le a_n$,
\begin{eqnarray}\label{eq:l102}
|H^{(\mu)}_{l,m^\circ}(\omega)- H^{(\mu)}_{l,m^\circ}(\omega')|=A\frac{\sin^2(m^\circ (\omega'-\omega)/2)}{|\sin [(\omega'-\omega)/2]|}+O(m^{\circ}a_n),
\end{eqnarray}
where the $O(m^{\circ}a_n)$ term is uniformly in $l$. Therefore, 
\[
| H^{(\mu)}_{l,m^\circ}(\omega)- H^{(\mu)}_{l,m^\circ}(\omega')|\lesssim (m^{\circ})^2|w-w'|+m^{\circ}a_n\lesssim (m^{\circ})^2 a_n
\]
uniformly in $l$
if $|\omega'-\omega|\le a_n.$
Since $\pr(\bar B_n^c)=O(b_n)$, the lemma follows.
\end{proof}
\begin{remark}
In order for the bound $(m^{\circ})^{3/2}a_n$ to converge to 0, we need $m^{\circ}\ll n/\log ^{2/3}n$ if $a_n=n^{-3/2}\log n$ and $m^{\circ}\ll n^{2/3}$ if $a_n=n^{-1}$, where $n^{-3/2}\log(n)$ and $n^{-1}$ come from Lemmas \ref{lemb6} and \ref{lem:8}. In fact, Lemma \ref{lem:b9} demonstrates the effectiveness of the local change point test in relieving the energy leak problem. Observe that in classic global change point detection algorithms such as those based on the CUSUM statistic, binary segmentation or dynamic programming, typically the detection window length $m^\circ$ in some steps are proportional to $n$. In this case  $(m^{\circ})^{3/2}a_n\ge C$ for some positive constant $C$ even if $a_n=n^{-3/2}$. Therefore, one cannot plug-in the estimated oscillatory frequency $\hat\omega$ for change point detection without appropriate adjustments for the estimation error. Unfortunately, estimating the error in frequency estimation under complex oscillation is a very difficult problem.  
\end{remark}

\begin{lemma}\label{lem:12}
Let $\omega\in\Omega$ and assume that there is no change point at frequency $\omega$.
Take $B$ to be the set of potential change points. 
Define 
\begin{align}
\Phi_{i}(j,\hat w) &\,= \left(\sum_{l = i -m'  }^{i }  \cos(   \hat{w} (l -j )) X_{l,n} -\sum_{l = i +1 }^{i + m' +1} \cos(  \hat{w} (l -j ) ) X_{l,n} \right) /\sqrt{2m'},\label{eq:stage2boots1}\\
\Psi_{i}(j,\hat w)&\,= \left(\sum_{l =i -m'  }^{i }  \sin(   \hat{w} (l -j ) ) X_{l,n} -\sum_{l =i +1 }^{i + m' +1} \sin(  \hat{w} (l -j ) ) X_{l,n} \right) /\sqrt{2m'}, \label{eq:stage2boots2}
\end{align}
and recall that
\[ 
\Upsilon_{k}(i,w) := \left(\sum_{l = k -m'  }^{k }  \exp(\im{   w (l - i ) } )X_l -\sum_{l = k +1 }^{k + m' +1} \exp(\im{ w (l - i ) } ) X_l \right) /\sqrt{2m'}\,. 
\]
For $ \mt+m' \leq i  \leq n - \mt - m' $, define 
\[ 
\hat{T}(i,w) =   \left( \sum_{k = i -\mt  }^{i }  \Upsilon_{k}(i,w) G_k -\sum_{k = i +1 }^{i + \mt +1}  \Upsilon_{k}(i,w) G_k \right)/\sqrt{2\mt} 
\]
and $\hat T(B,w)=\max_{i}|\hat{T}(i,w)|$, where $G_k$ are i.i.d. standard normal random variables independent of the data. If $\pr(|w - \hat{w}|\ge a_n) = O(b_n) $ for some sequences $a_n$ and $b_n$, $\tilde m\asymp n^{\gamma_1}$ with $1/6<\gamma_1<1$, $m'\asymp n^{\eta}$ with $0<\eta<1$, $a_n m'\rightarrow 0$, and Assumptions \ref{assumption3} to \ref{assumption8} hold true, then one can construct a sequence of events $H_n$ with $\pr(H_n)\ge 1-C(b_n+1/\log^q n)$, such that on the events $H_n$ and conditional on the data,
\begin{equation*}
\sup_{x\in\mathbb{R}}\Big|\pr(\hat T(B,w)\le x)-\pr (\hat T(B,\hat w)\le x)\Big|\lesssim [a_nm'(\sqrt{m'}+n^{1/q}\log n)]^{1/3}\log^2 n+n^{-1}\log^{1.5} n. 
\end{equation*}
\end{lemma}
\begin{proof}
Take 
$$
H_n:=\left\{|T_n^{(\epsilon)}(w) - T_n^{(\epsilon)}(\hat{w} )|\le a_n m' n^{1/q}\log(n)\right\}\bigcap\left\{|T_n^{(\mu)}(w) - T_n^{(\mu)}(\hat{w} )|\le  C(m')^{3/2}a_n\right\}\,.
$$ 
This lemma follows from Lemmas \ref{lem:b8} and \ref{lem:b9} and Proposition \ref{prop:complex_comparison}. Note that there is no need to put constraint on the range of $x$ here since the conditional variances of the real and imaginary parts of $\hat{T}(i,w)$ have positive lower bounds uniformly in $i$ for sufficiently large $n$ by Assumption \ref{assumption6} (see also \eqref{eq:10_8} below with $i=k$ therein).    
\end{proof}

The following Lemma \ref{lem:10} establishes a uniform Gaussian approximation result for the local change point detection algorithm.

\begin{lemma}\label{lem:10}
Assume Assumptions \ref{assumption3} to \ref{assumption8} hold, and $\tilde m\asymp n^{\gamma_1}$ with $\gamma_1>16/29$. Take $B$ to be the set of potential change points.
Define
$$
T^{(\epsilon)}(B, \omega):=\max_{i\in B} \left| \sum_{k = i -\mt  }^{i }  \exp(\im{ \omega (k - i ) }) \epsilon_k -\sum_{k = i +1 }^{i + \mt +1} \exp(\im{ \omega (k - i ) }) \epsilon_k \right|/\sqrt{2\mt}.
$$
And let $T^{(y)}(B, \omega)$ be the version of $T^{(\epsilon)}$ with $\epsilon_k$ therein replaced by $y_k$, where $\{y_k\}$ is a centered Gaussian time series preserving the covariance structure of $\{\epsilon_k\}$.  
Then, for any $\omega\in \Omega$, we have that
\begin{eqnarray}
\sup_{x\in\mathbb{R}}\Big|\pr(T^{(\epsilon)}(B,\omega)\le x)-\pr(T^{(y)}(B,\omega)\le x)\Big|\rightarrow 0.
\end{eqnarray}
\end{lemma}
\begin{proof}
Let $U_k=[\cos(\omega k)\epsilon_k,\,\sin(\omega k)\epsilon_k]^\top$, where $k=1,2,\cdots, n$, and we let $U_j=0$ if $j>n$. For any $j\in \mathbb{Z}$, define the projection operator ${\cal P}_j(\cdot)=\E(\cdot|\FF_j)-\E(\cdot|\FF_{j-1})$. Let 
\[
D_i=\sum_{j=i}^\infty{\cal P}_i(U_j)\,, 
\]
where $i=1,2,\cdots,n$. Let 
\[
\Sigma_k=\E(D_kD^\top_k)\,,
\]
where $k=1,2,\cdots, n$. 
Now, by the proof of Lemma 5 in \cite{zhou2014}, we have that
\begin{align}
&\max_{1\le k\le \mt}\max_i\left|\sum_{j=i-k}^i\left[\Sigma_j-\mbox{Diag}(v(j/n,\omega),v(j/n,\omega))\right]/\tilde m\right|\nonumber\\
=\,&O(n^{3/8}/\tilde m)=O(n^{-\alpha_4})\label{eq:10_6}
\end{align}
for some $\alpha_4>0$, where recall that $v$ is defined in Assumption \ref{assumption6}.
By Assumptions \ref{assumption3} to \ref{assumption8} 
 and the proofs of Theorem 1 and Corollary 2 of \cite{wu2011gaussian}, we have that, on a possibly richer probability space from a proper construction procedure, there exist i.i.d. standard normal random vectors $Y_i$, where  $i=1,2,\cdots,$ such that
\begin{eqnarray}\label{eq:10_1}
\max_{1\le i\le n}|S_{U,i}-S_{Y,i}|=O_\pr(n^{8/29}(\log(n))^{35/29})\,,
\end{eqnarray}
where 
\[
S_{U,i}:=\sum_{j=1}^iU_j\quad\text{and}\quad S_{Y,i}:=\sum_{j=1}^i\Sigma_j^{1/2}Y_j\,.
\]
Let $\tilde U_i=(U_{i})_1+\im (U_{i})_2$ and $\tilde{Y}_i=(\Sigma_i^{1/2}Y_i)_1+\im (\Sigma_i^{1/2}Y_i)_2$, where $x_j$ denotes the $j$-th entry of a vector $x$. Then, \eqref{eq:10_1} implies that
\begin{eqnarray}\label{eq:10_2}
\max_{1\le i\le n}|\tilde S_{U,i}-\tilde S_{Y,i}|=O_\pr(n^{8/29}(\log(n))^{35/29})\,,
\end{eqnarray}
where $\tilde S_{U,i}=\sum_{j=1}^i\tilde U_j$, and $\tilde S_{Y,i}$ is defined similarly. Note that 
\begin{align}
&\left| \sum_{k = i -\mt  }^{i }  \exp(\im{ \omega (k - i ) }) \epsilon_k -\sum_{k = i +1 }^{i + \mt +1} \exp(\im{ \omega (k - i ) }) \epsilon_k \right|\\
=\,&|\tilde S_{U,i+\tilde m+1}+\tilde S_{U,i-\tilde m-1}-2\tilde S_{U,i}|.\nonumber
\end{align}
Therefore, by combining this with \eqref{eq:10_2}, we have that
\begin{eqnarray}\label{eq:10_3}
T^{(\epsilon)}(B,\omega)= T^{(\tilde Y)}(B,\omega)+O_\pr(n^{8/29}(\log(n))^{35/29}/\sqrt{\tilde m})=o_{\pr} (n^{-\alpha_3})
\end{eqnarray}
for some $\alpha_3>0$, where 
\[
T^{(\tilde Y)}(B,\omega):=\max_i\left|\sum_{k=i-\tilde m}^i \tilde Y_k-\sum_{k=i+1}^{i+\tilde m+1} \tilde Y_k\right|/\sqrt{2\tilde m}\,. 
\]
By \eqref{eq:10_3}, a similar regular convex  polygon approximation argument as that in the proof of Proposition \ref{prop:complex_comparison} and Nazarov's inequality \citep{nazarov2003}, we have
\begin{eqnarray}\label{eq:10_4}
\sup_{x\in\mathbb{R}}\Big|\pr(T^{(\epsilon)}(B,\omega)\le x)-\pr(T^{(\tilde Y)}(B,\omega)\le x)\Big|\rightarrow 0\,.
\end{eqnarray}
Note that the range of $x$ does not have to be constrained since the eigenvalues of $\sum_{j={i-\tilde m}}^{i}\Sigma_j/\tilde{m}$ are uniformly bounded away from 0 for sufficiently large $n$ by Assumption \ref{assumption6} and \eqref{eq:10_6} above. 

Define 
$$
\hat Y_k=\mbox{Diag}(v^{1/2}(k/n,\omega),v^{1/2}(k/n,\omega))Y_k
$$
and
$$
T^{(D)}(\omega):=\max_i\left|\sum_{k=i-\tilde m}^i \hat{Y}_k-\sum_{k=i+1}^{i+\tilde m+1} \hat Y_k\right|/\sqrt{2\tilde m}\,.
$$
By \eqref{eq:10_6} and Proposition \ref{prop:complex_comparison}, we have that \begin{eqnarray}\label{eq:10_7}
\sup_{x\in\mathbb{R}}\Big|\pr(T^{(\tilde Y)}(B,\omega)\le x)-\pr(T^{(D)}(B,\omega)\le x)\Big|\rightarrow 0.
\end{eqnarray}
Note that again the range of $x$ does not need to be constrained in \eqref{eq:10_7} since all $v(j/n,\omega)$ are bounded away from 0. Finally, again by the proof of Lemma 5 in \cite{zhou2014}, we obtain that, 
\begin{align}
&\max_{i\le k}\left|\E[ \mathsf S_{y,i}\mathsf S^\top_{y,k}]-\sum_{j=k-\tilde m}^{i+\tilde m +1}\mbox{Diag}(v(j/n,\omega),v(j/n,\omega))/(2\mt)\right|\nonumber\\
=\,&O(n^{-\alpha_5}),\label{eq:10_8}
\end{align}
where 
$$
\mathsf S_{y,i}=\left(\sum_{j=i-\tilde m}^i[\cos(\omega j)y_j,\sin(\omega j)y_j]^\top-\sum_{j=i+1}^{i+\tilde m+1}[\cos(\omega j)y_j,\sin(\omega j)y_j]^\top\right)/\sqrt{2\tilde m}
$$ 
and $\alpha_5>0$ is a constant.
Hence, we obtain by Proposition \ref{prop:complex_comparison} that
\begin{eqnarray*}
\sup_{x\in\mathbb{R}}\Big|\pr(T^{(y)}(B,\omega)\le x)-\pr(T^{(D)}(B,\omega)\le x)\Big|\rightarrow 0\,.
\end{eqnarray*}
By combining the above bounds, the lemma follows.
\end{proof}

The following lemma establishes that the covariance structure of the real part of the phase-adjusted OBMB is asymptotically close to that of the target Gaussian process. The same result can be established for the covariance structure of the imaginary part and the covariance between the real and imaginary parts using the same arguments.

\begin{lemma}\label{lem:11}
Assume that assumptions \ref{assumption3} to \ref{assumption8} hold true. 
Take $\omega\in\Omega$. Further, assume that there is no change point at frequency $\omega$. For any $k\in [\mt+m',n-\mt-m']$, set $k^*=k-\mt-m'+1$. Define
$$
\Delta'' := \max_{\mt+m' \leq i,j  \leq n - \mt - m'} \left| \cov(\tilde S^{(2)}_{i^*}, \tilde S^{(2)}_{j^*}|X) -  \cov(\Theta^{(2)}_{i^*}, \Theta^{(2)}_{j^*})\right|\,,
$$
where recall the definitions of $\Theta^{(2)}_{i^*}$ and $\tilde S^{(2)}_{i^*}$ in Section \ref{sec:nullboots}.
Under the assumption that $m'\rightarrow\infty$ and $m'/\mt\rightarrow 0$, we have
$$
\pr(\Delta''\lesssim 1/m' +m'/\mt+ n^{1/q'} \sqrt{m'/\mt}\log n)\ge 1-C/(\log(n))^{q'} 
$$
for some finite positive constant $C$, where recall $q'=q/2$.
\end{lemma}

Note that $\Delta''$ defined in this Lemma is different from $\Delta'$ defined in Section \ref{sec:nullboots}.

\begin{proof}
We first separate $\Theta^{(2)}_{i^*}$ into individual summations. Denote 
$$
E_{i,\mt }^+ = \sum_{k = i - \mt  }^{i}  \cos( w (k - i + \mt) ) \epsilon_k\ \mbox{ and }\ E_{i,\mt}^- = \sum_{k = i + 1   }^{i+\mt +1}  \cos( w (k - i + \mt )) \epsilon_k\,.
$$ 
Then 
\[ 
\Theta^{(2)}_{i^*} = ( E_{i,\mt }^+  - E_{i,\mt }^- )/\sqrt{2\mt } 
\] 
and
	\begin{align*}
	2\mt \cov(&\Theta^{(2)}_{i^*}, \Theta^{(2)}_{j^*}) = \cov \left(   E_{i,\mt}^+  - E_{i,\mt}^- , E_{j,\mt}^+  - E_{j,\mt}^- \right) \\
	& = \cov\left(E_{i,\mt}^+,E_{j,\mt}^+ \right) -  \cov\left(E_{i,\mt}^-,E_{j,\mt}^+ \right) - \cov\left(E_{i,\mt}^+,E_{j,\mt}^- \right) +  \cov\left(E_{i,\mt}^-,E_{j,\mt}^- \right)\,.
	\end{align*}
If $  i \leq j $, $E_{i,\mt}^+$ and $\,E^-_{j,\mt}$ do not overlap, and we have 
	\begin{align*}
	\cov\left(E_{i,\mt}^+,E_{j,\mt}^- \right) &\leq \sum_{ i - \mt  \leq k \leq i} \sum_{ j \leq l \leq j + \mt} | \E \epsilon_k \epsilon_l | \\
	& \leq  \sum_{k = 0}^{\mt} (\mt - k) |j-i+\mt+k|^{-q} + \sum_{k = 1}^{\mt} k |j-i + k|^{-q} \\ &  \leq   \sum_{k = 0}^{\mt} (\mt - k) k^{-q} + \sum_{k = 1}^{\mt} k^{-q+1}\leq \mt^{-q+2} = O(1)
	\end{align*}
since $q\geq 4$ by Assumption \ref{assumption3}. For other terms, we will control them one by one below.
	
	For the multiplier bootstrap part we have
	\begin{align*}
	&2\tilde{m}\cov(\tilde S^{(2)}_{i^*}, \tilde S^{(2)}_{j^*}|X) \\
	 =&\, \cov \left(  \sum_{k = i - \mt  }^{i }  \Phi_{k}(i,w) G_k -\sum_{k = i +1 }^{i + \mt +1} \Phi_{k}(i,w) G_k  , \sum_{k = j -\mt  }^{j} \Phi_{k}(j,w) G_k -\sum_{k = j +1 }^{j + \mt +1} \Phi_{k}(j,w) G_k  \right) \\
	=&\, \begin{cases}
	0 & \text{ if } |i - j| > 2 \mt + 2\\ 
	- \sum_{k = i+ \delta_1}^{j - \delta_1}\Phi_{k}(i,w)\Phi_{k}(j,w)  & \text{ if }  \mt + 1< |i - j| \leq 2 \mt + 2\\
	\left[\sum_{k = i+ \delta_1}^{i}-  \sum_{k = i}^{j} + \sum_{k = j}^{j-\delta_1}\right]\Phi_{k}(i,w)\Phi_{k}(j,w) 
	&  \text{ if } 0  \leq |i - j| \leq \mt + 1\,,
	\end{cases}
	\end{align*}
	where $\delta_1 = |i -j| -\mt$. 
	
	\underline{Case 1, $|i - j| > 2\mt+2$}. By the above analysis, since both summations do not overlap, we have  
	\[ 	
	\cov(\Theta^{(2)}_{i^*}, \Theta^{(2)}_{j^*}) = O(1/\mt)\ \mbox {and }\ \cov(\tilde S^{(2)}_{i^*}, \tilde S^{(2)}_{j^*}|X) = 0\,.
	\]
	Thus, 
	\[ 
	| \cov(\Theta^{(2)}_{i^*}, \Theta^{(2)}_{j^*}) - \cov(\tilde S^{(2)}_{i^*}, \tilde S^{(2)}_{j^*}|X)| = O(1/\mt).  
	\]
	
	\underline{Case 2, $\mt+1< |i - j| \leq 2\mt+2$.} Without loss of generality, assume $i \leq j$. Note that $ \cov\left(E_{i,\mt}^-,E_{j  ,\mt }^+ \right)  $ is the only term with overlapping entries. Thus, 
	\[ 	
	\cov(\Theta^{(2)}_{i^*}, \Theta^{(2)}_{j^*})  = -\cov\left(E_{i,\mt}^-,E_{j ,\mt}^+ \right)/(2\mt) +  O(1/\mt) \,.
	\]
	We can further expand overlapping and non-overlapping parts and note the covariance between the non-overlapping part is well controlled: 
	\begin{align}
	\cov\left(E_{i,\mt}^-,E_{j ,\mt }^+ \right)  =&\,  \cov\left(\sum_{k = i +1  }^{i+\mt +1}  \cos({w (k - i ) }) \epsilon_k ,\sum_{k = j-\mt }^{j}  \cos({w (k - j) }) \epsilon_k \right) \label{Expansion of CovE-E+}\\
	 =&\, \cov\left(\sum_{k = i +1  }^{i+ \delta_1}  \cos({w (k - i) }) \epsilon_k ,\sum_{k = j-\mt}^{i+\mt+1}  \cos({w (k - j) }) \epsilon_k \right)\nonumber \\
	& + \cov\left(\sum_{k = i+1  }^{i+\mt}  \cos({ w (k - i ) }) \epsilon_k ,\sum_{k = j-\delta_1 }^{j}  \cos({w (k - j) }) \epsilon_k \right)\nonumber \\
	& + \cov\left(\sum_{k = i + \delta_1  }^{j - \delta_1}  \cos({ w (k - i  ) }) \epsilon_k ,\sum_{k =i + \delta_1}^{j- \delta_1}  \cos({ w (k - j ) }) \epsilon_k \right)\nonumber \\
	=&\, \cov\left(\sum_{k = i + \delta_1  }^{j - \delta_1}  \cos({ w (k - i ) }) \epsilon_k ,\sum_{k =i +\delta_1}^{j- \delta_1}  \cos({w (k - j ) }) \epsilon_k \right) +  O(1)\,.\nonumber
	\end{align}
	Next we move on to the  multiplier bootstrap part in this case.
	Note that under our assumption  about the trend,
	\[  
	\left(\sum_{l = k -m'  }^{k }  \cos({   w (l - i) } ) \mu_l -\sum_{l = k +1 }^{k + m' +1} \cos({   w (l - i ) } ) \mu_l \right) /\sqrt{2m'} = O(1).
	\] 
	Thus, we can break down 
	\begin{align*}
	\Phi_{k}(i,w) & =  \left(\sum_{l = k -m'  }^{k }  \cos({  w (l - i ) } ) \epsilon_l -\sum_{l = k +1 }^{k + m' +1} \cos({   w (l - i ) } ) \epsilon_l \right) /\sqrt{2m'} + O(1) \\
	& = \Phi_{k}^{(\epsilon)}(i,w)  + O(1)\,.
	\end{align*}	
	Breaking down the covariance structure of the bootstrap statistics in the same way like that in \eqref{Expansion of CovE-E+}, we get
	\begin{align}
	&\cov(\tilde S^{(2)}_{i^*},\tilde S^{(2)}_{j^*}|X) =  - \sum_{k = i+ \delta_1}^{j - \delta_1} \Phi_{k}(i,w)\Phi_{k}(j,w)/(2\mt) \label{Equation cov Fi Fj}\\
	=&\,  - \sum_{k = i+ \delta_1}^{j - \delta_1} (\Phi_{k}^{(\epsilon)}(i,w) + O(1) )(\Phi_{k}^{(\epsilon)}(j,w)+ O(1) )/(2\mt) \nonumber\\
	 =&\,  \sum_{k = i+ \delta_1}^{j - \delta_1} \Phi_{k}^{(\epsilon)}(i,w)\Phi_{k}^{(\epsilon)}(j,w) + a_n  \sum_{k = i+ \delta_1}^{j - \delta_1} \Phi_{k}^{(\epsilon)}(i,w)/(2\mt)\nonumber \\
	 &\,+ b_n  \sum_{k = i+ \delta_1}^{j - \delta_1} \Phi_{k}^{(\epsilon)}(j,w)/(2\mt) + O(1/\mt)\,,\nonumber
	\end{align}
	where $a_n$ and $b_n$ are some bounded sequence.
	Since \eqref{Equation cov Fi Fj} holds for all $j$ satisfying $\mt+1< |i - j| \leq 2\mt+2$, 	
	\begin{align*}
	&\left\| \max_{j\in [i+\mt+2,i+2\mt+2]} \Big|\sum_{k = i+ \delta_1}^{j - \delta_1}   \Phi_{k}^{(\epsilon)}(i,w)\Big| \right\|_q \\
	\leq \,&\left\|  m' \max_{j\in [i+\mt+2,i+2\mt+2]}\Big| \sum_{k = i+ \delta_1 - m'}^{j - \delta_1+m'} \epsilon_k\Big|/\sqrt{2m'}  \right\|_q  =  O( \mt m' )\,.
	\end{align*}
	We thus have
	\[ 
	\left\|\max_{i\in B, j\in [i+\mt+2,i+2\mt+2]}  \Big| \sum_{k = i+ \delta_1}^{j - \delta_1}  \Phi_{k}^{(\epsilon)}(i,w) \Big|\right\|_q  = O(\sqrt{\mt m'} n^{1/q} ).
	\]
	Then, the following holds uniformly for all $\mt+1< |i - j| \leq 2\mt+2$ 
	\[
	\cov(\tilde S^{(2)}_{i^*},\tilde S^{(2)}_{j^*}|X) =  - \sum_{k = i+ \delta_1}^{j - \delta_1} \Phi_{k}^{(\epsilon)}(i,w)\Phi_{k}^{(\epsilon)}(j,w) / (2\mt) + O_{|q|}( n^{1/q}\sqrt{m'/\mt})\,,
	\]
	where the asymptotic notation $X_n=O_{|q|}(a_n)$ means that the random variable $X_n$ and the scalar $a_n$ satisfy $\|X_n/a_n\|_q=O(1)$. 
	Define
\[ 
S_{i,k,m'} = \sum_{l = k   }^{k +m' } \cos({   w (l - i ) } ) \epsilon_l \,. 
\]	
	Expand the quadratic term, and we have
	\begin{align*}
	&\sum_{k = i+ \delta_1}^{j - \delta_1} \Phi_{k}^{(\epsilon)}(i,w)\Phi_{k}^{(\epsilon)}(j,w) / (2\mt) \\
	 = &\,\sum_{k = i+ \delta_1}^{j - \delta_1} (S_{i,k-m',m'} -S_{i,k+1,m'} )(S_{j,k-m',m'} -S_{j,k+1,m'} )  / (4\mt m') \\
	=& \,\sum_{k = i+ \delta_1}^{j - \delta_1} S_{i,k-m',m'}S_{j,k-m',m'} / (4\mt m') +\sum_{k = i+ \delta_1}^{j - \delta_1} S_{i,k+1,m'}S_{j,k+1,m'} / (4\mt m')\\
	&- \sum_{k = i+ \delta_1}^{j - \delta_1} S_{i,k-m',m'}S_{j,k+1',m'} / (4\mt m') - \sum_{k = i+ \delta_1}^{j - \delta_1} S_{i,k+1,m'}S_{j,k-m',m'} / (4\mt m')\,.
	\end{align*} 
	By the proof of Lemma 1 of \cite{zhou2013}, we have for a fixed $j \in [\mt+m',\dots,n - \mt-m' ]$ 
	\[ 
	\left\|  \max_{|i-j| \leq \mt }\left( \sum_{k = i+ \delta_1}^{j - \delta_1} \Phi_{k}^{(\epsilon)}(i,w)\Phi_{k}^{(\epsilon)}(j,w) - \E\sum_{k = i+ \delta_1}^{j - \delta_1} \Phi_{k}^{(\epsilon)}(i,w)\Phi_{k}^{(\epsilon)}(j,w) \right)  \right\|_{q'} = O(\sqrt{m'/\mt})\,.   
	\]
	This implies that for all $i,j$, 
	\begin{align*}
	&\left\| \max_{\mt + m'\leq j \leq n-\mt-m'} \,\, \max_{|i-j| \leq \mt }\left( \sum_{k = i+ \delta_1}^{j - \delta_1} \Phi_{k}^{(\epsilon)}(i,w)\Phi_{k}^{(\epsilon)}(j,w) - \E\sum_{k = i+ \delta_1}^{j - \delta_1} \Phi_{k}^{(\epsilon)}(i,w)\Phi_{k}^{(\epsilon)}(j,w) \right)   \right\|_{q'}\\
	 = \,&O(n^{1/q'} \sqrt{m'/\mt})\,.   
	\end{align*}
	Lastly, note that the non-overlapping sums satisfy 
	\begin{align*}
	&\left| \E \sum_{k = i+ \delta_1}^{j - \delta_1} S_{i,k-m',m'}S_{j,k+1,m'} / (4\mt m') \right|\\
	 =\,&2\left| \sum_{k = i+ \delta_1}^{j - \delta_1}  \cov\left(S_{i,k-m',m',}S_{j,k+1,m'} \right)  / (4\mt m')  \right| = O(1/m').
	\end{align*}
	So we have
	\begin{align*}
	&\sum_{k = i+ \delta_1}^{j - \delta_1} \Phi_{k}^{(\epsilon)}(i,w)\Phi_{k}^{(\epsilon)}(j,w) / (2\mt) \\
	= &\,\E \sum_{k = i+ \delta_1}^{j - \delta_1} S_{i,k-m',m'}S_{j,k-m',m'} / (4\mt m') \\
	&\,+ \E \sum_{k = i+ \delta_1}^{j - \delta_1} S_{i,k+1,m'}S_{j,k+1,m'} / (4\mt m') + O(1/m').  
	\end{align*}
	Combining the previous results, we have
	\begin{align*}
	& \left|	\cov(\tilde S^{(2)}_{i^*}, \tilde S^{(2)}_{j^*}|X) -  \cov(\Theta^{(2)}_{i^*}, \Theta^{(2)}_{j^*})\right| \\
	=& \left\vert \E \sum_{k = i+ \delta_1}^{j - \delta_1} S_{i,k-m',m'}S_{j,k-m',m'}+ S_{i,k,m'}S_{j,k,m'} / (4\mt m') \right. \\
	&-\left.\cov\left(\sum_{k = i+ \delta_1  }^{j - \delta_1}  \cos({w (k - i ) }) \epsilon_k ,\sum_{k = i + \delta_1 }^{j - \delta_1}  \cos({w (k - j) }) \epsilon_k \right)/(2\mt) \right| \\
	&+ O({1/\mt}) + O(1/m') + O_{|q'|}(n^{1/q'}\sqrt{m'/\mt })\\
	\le & \left|  \E \sum_{k = i+ \delta_1}^{j - \delta_1-m'} S_{i,k,m'}S_{j,k,m'}  / (2\mt m') -  \frac{1}{2\mt}  \sum_{k = i+ \delta_1}^{j - \delta_1}  \sum_{l = i+ \delta_1}^{j - \delta_1}  \cos({ w (k - i) }) \cos({ w (l - j) })  \E \epsilon_k \epsilon_l  \right|  \\
	& + \left|  \E \sum_{k = i+ \delta_1-m'}^{i + \delta_1}  S_{i,k,m'}S_{j,k,m'}  / (4\mt m')  \right|+  \left|  \E \sum_{k = j - \delta_1 - m'}^{j - \delta_1 }  S_{i,k,m'}S_{j,k,m'}  / (4\mt m')  \right|\\
	& +O({1/\mt}) + O(1/m') + O_{|q'|}(n^{1/q'}\sqrt{m'/m}) \\
	\le &   \sum_{ |k- l| \leq m'}  \left| \frac{m'-|k-l|}{2\mt m'} - \frac{1}{2\mt} \right| 	\left| \E \epsilon_k \epsilon_l \right| + \frac{1}{2\mt}  \sum_{ |k - l| >  m'} \left| \E \epsilon_k \epsilon_l \right|\\
	&+ O(1/\mt) + O(1/m') +O_{|q'|}(n^{1/q'} \sqrt{m'/\mt}) +  O(m'/\mt)  \\ 
	=&\,  O({1/\mt}) + O(1/m') + O_{|q'|}(n^{1/q'} \sqrt{m'/\mt}) + O(m'/\mt).
	\end{align*}
	\underline{Case 3, $|i - j| \leq \mt+1 $}.  Assume $i \leq j$.  Then, 
	\[ 
	2\mt	\cov(\Theta^{(2)}_{i^*}, \Theta^{(2)}_{j^*})  =  \cov\left(E_{i,\mt}^+,E_{j,\mt}^+ \right) -  \cov\left(E_{i,\mt}^-,E_{j,\mt}^+ \right)  +  \cov\left(E_{i,\mt}^-,E_{j,\mt}^- \right) +  O(1) 
	\] 
	and  
	\[ 	
	2\mt \cov(\tilde S^{(2)}_{i^*}, \tilde S^{(2)}_{j^*}|X)  =  \left[	\sum_{k = i- \delta_2}^{i}  -  \sum_{k = i}^{j}+ \sum_{k = j}^{j+\delta_2}\right] \Phi_{k}(i,w)\Phi_{k}(j,w) \,.
	\]
	Similarly to Case 2, we get for all $|i-j| \leq \mt+1$ the following bounds hold simultaneously:
	\begin{align*}
	\Big| \cov\left(E_{i,\mt}^+,E_{j,\mt}^+ \right) /(2\mt) &-   \sum_{k = i- \delta_2}^{i} \Phi_{k}(i,w)\Phi_{k}(j,w)/(2\mt) \Big| \\
	= &\, O({1/\mt}) + O(1/m') + O_{|q'|}(n^{1/q'} \sqrt{m'/\mt}) + O(m'/\mt)\,,  \\
	\Big| \cov\left(E_{i,\mt}^-,E_{j,\mt}^+ \right) /(2\mt) & -   \sum_{k = i}^{j}\Phi_{k}(i,w)\Phi_{k}(j,w) /(2\mt)  \Big| \\
	=&\, O({1/\mt}) + O(1/m') + O_{|q'|}(n^{1/q'} \sqrt{m'/\mt}) + O(m'/\mt), 
	\end{align*}
and
\begin{align*}
	\Big|\cov\left(E_{i,\mt}^-,E_{j,\mt}^- \right)/(2\mt) &-  \sum_{k = j}^{j+\delta_2} \Phi_{k}(i,w)\Phi_{k}(j,w)/(2\mt) \Big| \\
	=&\, O({1/\mt}) + O(1/m') + O_{|q'|}(n^{1/q'} \sqrt{m'/\mt}) + O(m'/\mt)\,. 
	\end{align*}
	As a result, with Cases 1, 2 and 3, for all $i ,j \in [\mt+m'+1,\dots,n - \mt-m']$, we have 
	\begin{align*} 
	&\left| \cov(\tilde S^{(2)}_{i^*}, \tilde S^{(2)}_{j^*}|X) -  \cov(\Theta^{(2)}_{i^*}, \Theta^{(2)}_{j^*})\right|\\
	=  &\, O({1/\mt}) + O(1/m') + O_{|q'|}(n^{1/q'} \sqrt{m'/\mt}) + O(m'/\mt)\\
	= &\,O(1/m')+ O(m'/\mt) + O_{|q'|}(n^{1/q'} \sqrt{m'/\mt})\,. 
	\end{align*}
The lemma follows by a simple application of Markov's inequality.	
	  
\end{proof}

\begin{proof}[Proof of Theorem \ref{thm:s2_consistancy}]
Recall that \begin{align*}
T(B,\omega) =  &  \max_{i\in B}  \left| \sum_{l = i - \mt  }^{i}  \exp( \sqrt{-1} w( l- i)) X_{l,n} -\sum_{l = i +1 }^{i + \mt +1} \exp( \sqrt{-1} w ( l- i) ) X_{l,n} \right| /\sqrt{2\mt}.
\end{align*}
The first claim of the Theorem follows from Lemma \ref{lem:11} and the observation that $m'/\mt$ is dominated by $n^{2/q}\sqrt{m'/\mt}\log n$. Note that Lemmas \ref{lem:b8} to \ref{lem:12} imply that $T(B,\hat\omega)$  can be well approximated by $T(B,\omega)$ with asymptotically negligible errors and, conditional on the data, $\hat T(B,\hat\omega)$ can be well approximated by $\hat T(B,\omega)$ with asymptotically negligible errors with high probability. Note that under the hypothesis that there is no change point at frequency $\omega$, we have $T(B,\omega)=T^{(\epsilon)}(B,\omega)+O(\mt^{-1/2})$. By Lemma \ref{lem:10}, we have
\begin{eqnarray*}
\sup_{x\in\mathbb{R}}\Big|\pr(T^{(\epsilon)}(B,\omega)\le x)-\pr(T^{(y)}(B,\omega)\le x)\Big|\rightarrow 0\,,
\end{eqnarray*}
where $\{y_k\}$ is a centered Gaussian time series preserving the covariance structure of $\{\epsilon_k\}$.
On the other hand, by Lemma \ref{lem:11} and Proposition \ref{prop:complex_comparison}, we have, on an event with probability at least $1-C/(\log n)^{q'}$,
\begin{align*}
&\sup_{x\in\mathbb{R}}\Big|\pr(\hat T(B,\omega)\le x|X)-\pr(T^{(y)}(B,\omega)\le x)\Big|\\
\le&\,
\Big(\frac{1}{m'}+n^{1/q'}\sqrt{\frac{m'}{\mt}}\log n\Big)^{1/3}(\log n)^{7/6}+ \frac{(\log n)^{13/12}}{n}\,.
\end{align*}
Note that the range of $x$ does not need to be constrained here as the marginal variances of the components in $T^{(y)}(\omega)$ are bounded away from 0 by \eqref{eq:10_8} with $i=k$ therein. The theorem follows.
\end{proof}

\subsubsection{Estimation accuracy}

\begin{lemma}\label{lem:13}
Under the assumptions of Theorem \ref{thm:step2acc}, we have that
\begin{equation*}
\pr\left(\max_{1\leq r\leq M_k}|\hat{b}_{r,k}-b_{r,k}|\ge h_n\log \mt\right)\rightarrow 0    
\end{equation*}
for any $\omega_k\in\Omega$ such that $M_k\neq 0$ and any sequence $h_n$ that diverges to infinity at an arbitrarily slowly rate.
\end{lemma}
\begin{proof}
Without loss of generality, we assume that $M_k=1$ since other cases follow by essentially the same arguments. We shall omit the subscript $k$ in the sequel for simplicity. Write the mean function as
$\mu_i=C_0\cos(\omega i+\theta_0)+f(i/n)$ when $1\le i\le b_{1}$ and $\mu_i=C_1\cos(\omega i+\theta_1)+f(i/n)$ when $b_1<i\le n$. Define 
\[
{H}_{l,\mt}(\omega):=\sum_{k=l-\mt}^l\exp(\im \omega(k-l) )X_{k}-\sum_{k=l+1}^{l+\mt+1}\exp(\im \omega(k-l) )X_{k}\,.
\]
Recall the definitions of $H^{(\epsilon)}_{l,\mt}(\omega)$ and $H^{(\mu)}_{l,\mt}(\omega)$ in Lemmas \ref{lem:b8} and \ref{lem:b9}. Note that $H_{l,\mt}(\omega)=H^{(\epsilon)}_{l,\mt}(\omega)+H^{(\mu)}_{l,\mt}(\omega)$.
By the same argument for Lemmas \ref{lem:b8}, \ref{lem:b9} and Assumption \ref{assumption9}, we have that
\begin{align}\label{eq:HX bound lemma14}
|{H}_{b_1,\mt}(\hat\omega)|=[C_2+o(1)]\mt+o_\pr(\mt^{1/2}),
\end{align}
where $C_2=|C_0\exp(-\im\theta_0)-C_1\exp(-\im \theta_1)|$.
 Note that $C_2\ge \delta_4$ by Assumption \ref{assumption9}.

First of all, elementary calculations and Lemma \ref{lem:b9} show that
\begin{eqnarray*}
\max_{|l-b_1|>\mt}|{H}^{(\mu)}_{l,\mt}(\hat{\omega})|-|{H}^{(\mu)}_{b_1,\mt}(\hat{\omega})|=-C_2\mt+o_\pr(\mt^{1/2}).
\end{eqnarray*}
On the other hand, by the same argument for Lemmas \ref{lem:b8} and \ref{lem:10}, we have that
\begin{equation*}
\max_{|l-b_1|>\mt}|H^{(\epsilon)}_{l,\mt}(\hat{\omega})|=O_\pr((\mt\log n)^{1/2}).    
\end{equation*}
Hence, with probability converging to 1, 
\[
\max_{|l-b_1|>\mt}| H_{l,\mt}(\hat{\omega})|< H_{b_1,\mt}(\hat{\omega}),
\] 
which implies that $\pr(|\hat b_1-b_1|> \mt)\rightarrow 0$. 

Now, let $a_n$ be a diverging sequence that is dominated by $\mt$. Elementary but tedious calculations and the proof of Lemma \ref{lem:b9} yield that, uniformly for all $l$ such that $a_n< |l-b_1|<\mt$ and $\omega'$ such that $|\omega'-\omega|=O(g_n/n)$ for some $g_n>0$ diverging at an arbitrarily slowly rate, 
\begin{eqnarray}\label{eq:131}
|{H}^{(\mu)}_{l,\mt}(\omega')|^2-|{H}^{(\mu)}_{b_1,\mt}(\omega')|^2=\Upsilon_{1,n}-\Upsilon_{2,n}+O(r_n\mt^2g_n/n),
\end{eqnarray}
where 
$$\Upsilon_{1,n}=\frac{C^2_2}{4}\Big|\frac{\sin(R_nb/2)}{\sin (b/2)}\Big|^2\,,\ \ 
\Upsilon_{2,n}=\frac{C^2_2}{4}\Big|\frac{\sin(\mt b/2)}{\sin (b/2)}\Big|^2,$$
$R_n=l+\mt-b_1$, $r_n=|l-b_1|$, and $b=\omega'-\omega$. Elementary calculations using the assumption $\gamma_1<2/3$ and the fact that $x/2\le\sin(x)\le x$ for sufficiently small non-negative $x$ show that
\[
\Upsilon_{1,n}-\Upsilon_{2,n}\le -C_2^2r_n(\mt+R_n)/16\lesssim -\mt r_n
\] 
for a sufficiently large $n$. Therefore, we have that with probability approaching $1$, 
\begin{eqnarray}\label{eq:132}
|{H}^{(\mu)}_{l,\mt}(\hat\omega)|^2-|{H}^{(\mu)}_{b_1,\mt}(\hat\omega)|^2\lesssim -\mt r_n +O(r_n\mt^2g_n/n)\lesssim -\mt r_n
\end{eqnarray}
uniformly for all $l$ satisfying $|l-b_1|\in [a_n,\mt]$, where we utilized the fact that $g_n$ can approach infinity arbitrarily slowly and $\mt\ll n$ and hence $\mt^2g_n/n\ll \mt$. Furthermore, by the same argument for Lemmas \ref{lem:b8} and \ref{lem:10}, we have that
\begin{equation}
\max_{|l-b_1|\le\mt }|H^{(\epsilon)}_{l,\mt}(\hat{\omega})|=O_\pr((\mt\log \mt)^{1/2})\,. \label{eq:Hepsilon bound}   
\end{equation} 
Choose $a_n=h_n\log \mt$, where $h_n>0$ is diverging at an arbitrarily slow rate.  Then, we find that uniformly for all $l$ such that $|b_1-l|\in [a_n,\mt]$,
\begin{align*}
&|{H}_{l,\mt}(\hat\omega)|^2-|{H}_{b_1,\mt}(\hat\omega)|^2\\
\le\,& |{H}^{(\mu)}_{l,\mt}(\hat\omega)|^2-|{H}^{(\mu)}_{b_1,\mt}(\hat\omega)|^2+|H^{(\epsilon)}_{l,\mt}(\hat\omega)|^2-|H^{(\epsilon)}_{b_1,\mt}(\hat\omega)|^2+2\texttt{I}+2\texttt{II},
\end{align*}
where 
\[
\texttt{I}=|e^{\im \hat\omega l}{H}^{(\mu)}_{l,\mt}(\hat\omega)-e^{\im \hat\omega b_1}{H}^{(\mu)}_{b_1,\mt}(\hat\omega)||H^{(\epsilon)}_{l,\mt}(\hat\omega)| 
\]
and 
$$
\texttt{II}=|e^{\im \hat\omega l}H^{(\epsilon)}_{l,\mt}(\hat\omega)-e^{\im \hat\omega b_1}H^{(\epsilon)}_{b_1,\mt}(\hat\omega)|| H^{(\mu)}_{b_1,\mt}(\hat\omega)|\,.
$$
Following the arguments above, it is easy to show that 
\[
|e^{\im \hat\omega l}H^{(\epsilon)}_{l,\mt}(\hat\omega)-e^{\im \hat\omega b_1}H^{(\epsilon)}_{b_1,\mt}(\hat\omega)|=O_\pr(\sqrt{r_n\log \mt})
\] 
and 
\[
|e^{\im \hat\omega l}{H}^{(\mu)}_{l,\mt}(\hat\omega)-e^{\im \hat\omega b_1}{H}^{(\mu)}_{b_1,\mt}(\hat\omega)|=O(r_n)\,,
\] 
where the bounds are uniform across $l$ satisfying $|b_1-l|\in [a_n,\mt]$. Hence, by \eqref{eq:HX bound lemma14}, \eqref{eq:132}, and \eqref{eq:Hepsilon bound}, uniformly for all $l$ such that $|b_1-l|\in [a_n,\mt]$, we have, with probability approaching 1,
\begin{align}\label{eq:133}
&|{H}_{l,\mt}(\hat\omega)|^2-|{H}_{b_1,\mt}(\hat\omega)|^2\\
\le\,& -C_2^2\mt r_n/20+O(\mt\log \mt+\sqrt{r_n\log \mt}\mt+r_n\sqrt{\mt\log \mt}).    \nonumber
\end{align}
Observe that \eqref{eq:133} is negative for sufficiently large $n$. Hence the lemma follows.
\end{proof}

\begin{proof}[Proof of Theorem \ref{thm:step2acc}]
Part 1 of the Theorem follows by Theorem \ref{thm:s2_consistancy}. By Lemma \ref{lem:13}, we just need to show that
\begin{eqnarray}\label{eq:s2a1}
\pr(|\hat D_k|=D_k)\rightarrow1-\beta.
\end{eqnarray}
Note that, by \eqref{eq:HX bound lemma14} in Lemma \ref{lem:13}, ${H}_{b_{r,k},\mt}\ge C\mt$ with probability approaching 1. On the other hand, the critical values for the first $M_k$ steps are at most $O_\pr(\max(1,m'/\sqrt{\mt})\sqrt{\log n})$, which is dominated by $\mt$. Hence $\pr(|\hat D_k|<|D_k|)\rightarrow 0$. 

We now show that $\pr(|\hat D_k|>|D_k|)\rightarrow \beta$. By the similar arguments as those in the proof of Theorem \ref{thm:s1_accuracy}, we have that, after $M_k$ steps, the $(M_{k}+1)$-th step of change point estimation is asymptotically equivalent to performing change point test on the set $\tilde{B}_{M_{k}+1}:=B_1-\cup_{j=1}^{M_k}[b_{j,k}-h_n\log\mt-\mt,b_{j,k}+h_n\log\mt+\mt]$. Note that there is no change point on the set $\tilde{B}_{M_{k}+1}$ and each point in $\tilde{B}_{M_{k+1}}$ is at least $\mt+h_n\log \mt$ away from a change point. By the proof of Theorem \ref{thm:s2_consistancy}, we have that 
\begin{eqnarray*}
\pr(T(\tilde B_{M_{k}+1})>crit_{\beta,M_{k}+1}(\tilde B_{M_{k}+1}))\rightarrow \beta.
\end{eqnarray*}
The Theorem follows.

\end{proof}

\end{document}